\documentclass[11pt,reqno]{amsart}

\usepackage{amsmath, amsthm, amsopn, amssymb}
\usepackage{enumerate, enumitem, tikz, etoolbox, intcalc, geometry, caption, subcaption, afterpage}
\usepackage[english]{babel}

\usetikzlibrary{arrows,automata}
\usetikzlibrary{decorations,decorations.pathmorphing}

\definecolor{edge-color}{rgb}{0.85,0.85,0.85}
\definecolor{edge-on-color}{rgb}{0.8,0.2,0.2}
\definecolor{edge-bnd-color}{rgb}{0.4,0.4,0.9}
\definecolor{domain-path-color}{rgb}{0.4,0.4,0.9}
\tikzstyle edge=[color=edge-color]
\tikzstyle edge-on=[color=edge-on-color, thick]
\tikzstyle edge-off=[color=white,thick,dotted]
\tikzstyle edge-bnd=[color=edge-bnd-color, ultra thick]
\tikzstyle domain-path=[color=domain-path-color, thick, dotted]
\tikzstyle node0=[circle, fill=gray!50]
\tikzstyle node1=[circle, fill]
\tikzstyle node2=[circle, draw]

\def\hexagonEdge[#1][#2][#3][#4]{ % #3 = -1 for left edge, 0 for top edge, 1 for right edge
 \def \h {0.866cm}
 \begin{scope}[xshift=(#1) * 1.5cm, yshift=(#1+2*(#2)) *\h]
	\def \dd { \intcalcMod{\intcalcAdd{#3}{6}}{6} }
	\ifnumequal{\dd}{1}{ \draw [#4] {(0:1) -- (60:1)}; } {}
	\ifnumequal{\dd}{0}{ \draw [#4] {(60:1) -- (120:1)}; } {}
	\ifnumequal{\dd}{5}{ \draw [#4] {(120:1) -- (180:1)}; } {}
	\ifnumequal{\dd}{4}{ \draw [#4] {(180:1) -- (240:1)}; } {}
	\ifnumequal{\dd}{3}{ \draw [#4] {(240:1) -- (300:1)}; } {}
	\ifnumequal{\dd}{2}{ \draw [#4] {(300:1) -- (0:1)}; } {}
 \end{scope}
}

\def\hexagonEdges[#1][#2][#3][#4]{
 	\foreach \x/\y/\d in {#4} {
 		\hexagonEdge[(#2)+\x][(#3)+\y][\d][#1];
}}

\def\hexagonEdgesFullRow[#1][#2][#3][#4]{
 	\foreach \x/\d in {#4} {
 		\hexagonEdge[(#2)+\x][(#3)-\intcalcDiv{\x}{2}][\d][#1];
}}

\def\hexagonFullRow[#1][#2]{
	\foreach \x in {#2} {
 		\hexagon[\x][(#1)-\intcalcDiv{\x}{2}][1];
}}

\def\hexagonFullColumn[#1][#2]{
	\foreach \y in {#2} {
 		\hexagon[#1][\y-\intcalcDiv{#1}{2}][1];
}}

\def\hexagon[#1][#2]{
 \def \h {0.866cm}
 \begin{scope}[xshift=(#1) * 1.5cm, yshift=(#1+2*(#2)) *\h]
    \draw [edge] {(0:1) -- (60:1) -- (120:1) -- (180:1) -- (240:1) -- (300:1) -- cycle};
 \end{scope}
}

\def\hexagonNode[#1][#2][#3][#4]{
 \def \h {0.866cm}
 \begin{scope}[xshift=(#1) * 1.5cm, yshift=(#1+2*(#2)) *\h]
    \def\c{\intcalcMod{(#1)-(#2)}{3}}
	\def\d{#3}
	\ifnumequal{\d}{-1}{}{\ifnumequal{\d}{\c}{}{\def\c{9}}}
    \ifnumequal{\c}{0}{ \node [node0] at (0,0) { }; } {}
    \ifnumequal{\c}{1}{ \node [node1] at (0,0) { }; } {}
    \ifnumequal{\c}{2}{ \node [node2] at (0,0) { };} {}
 \end{scope}
}

\def\hexagonGrid[#1][#2][#3]{
  \foreach \i in {0,...,#1} {
   \foreach \j in {0,...,#2} {
 	\def\y{\j - \intcalcDiv{\i}{2}}
  	\hexagon[\i][\y];
	\ifnumequal{#3}{1}{\hexagonNode[\i][\y][-1][0]}{}
}}}

\def\hexagonGridOneClass[#1][#2][#3][#4]{
\begin{scope}
	\tikzstyle node0=[circle, fill=gray!50, minimum size=5pt,inner sep=0]
	\tikzstyle node1=[circle, fill=gray!50, minimum size=5pt,inner sep=0]
	\tikzstyle node2=[circle, fill=gray!50, minimum size=5pt,inner sep=0]

  \foreach \i in {0,...,#1} {
   \foreach \j in {0,...,#2} {
 	\def\y{\j - \intcalcDiv{\i}{2}}
  	\hexagon[\i][\y];
	\hexagonNode[\i][\y][#3][#4];
}}
\end{scope}}

\def\groundState[#1][#2][#3]{
  \hexagonGrid[#1][#2][#3];
  \foreach \i in {0,...,#1} {
   \foreach \j in {0,...,#2} {
	\ifnumequal{ \intcalcMod{\i- \j + \intcalcDiv{\i}{2}}{3} }{1}{\trivialLoop[\i][\j - \intcalcDiv{\i}{2}]}{}
}}}

\def\trivialLoop[#1][#2]{ \hexagonEdges[edge-on][#1][#2][0/0/0,0/0/1,0/0/2,0/0/3,0/0/4,0/0/5];}
\def\doubleLoopUp[#1][#2]{ \hexagonEdges[edge-on][#1][#2][0/1/-2,0/1/-1,0/1/0,0/1/1,0/1/2,0/0/1,0/0/2,0/0/3,0/0/4,0/0/5]; }
\def\doubleLoopRight[#1][#2]{ \hexagonEdges[edge-on][#1][#2][0/0/0,0/0/5,0/0/4,0/0/3,0/0/2,1/0/3,1/0/2,1/0/1,1/0/0,1/0/-1]; }
\def\doubleLoopLeft[#1][#2]{ \hexagonEdges[edge-on][#1][#2][0/1/0,0/1/1,1/0/0,1/0/1,1/0/2,1/0/3,1/0/4,0/1/3,0/1/4,0/1/5]; }

\newtheorem{thm}{Theorem}[section]
\newtheorem{lemma}[thm]{Lemma}
\newtheorem{prop}[thm]{Proposition}
\newtheorem{definition}[thm]{Definition}
\newtheorem{cor}[thm]{Corollary}

\newtheorem{quest}[thm]{Question}
\newtheorem{fact}[thm]{Fact}
\theoremstyle{definition}
\newtheorem*{remark}{Remark}

\newcommand{\reffig}[1] {\textsc{\ref{#1}}}

\newcommand{\cC}{\mathcal{C}}

\newcommand{\cF}{\mathcal{F}}
\newcommand{\cE}{\mathcal{E}}

\newcommand{\R}{\mathbb{R}}
\newcommand{\Z}{\mathbb{Z}}
\newcommand{\HH}{\mathbb{H}}
\newcommand{\VH}{V(\HH)}
\newcommand{\EH}{E(\HH)}
\newcommand{\T}{\mathbb{T}}
\newcommand{\sfT}{\mathsf{T}}

\renewcommand{\Pr}{\mathbb{P}}

\newcommand{\nSn}{\sqrt{n} \cdot \mathbb{S}^{n-1}}

\newcommand{\Int}[1]{\mathrm{Int}(#1)}
\newcommand{\Ext}[1]{\mathrm{Ext}(#1)}
\newcommand{\IntEdge}[1]{\mathrm{Int}^\mathrm{E}(#1)}
\newcommand{\IntVert}[1]{\mathrm{Int}^\mathrm{V}(#1)}
\newcommand{\ExtEdge}[1]{\mathrm{Ext}^\mathrm{E}(#1)}
\newcommand{\ExtVert}[1]{\mathrm{Ext}^\mathrm{V}(#1)}
\newcommand{\IntHex}[1]{\mathrm{Int}^{\mathrm{hex}}(#1)}
\newcommand{\BadEdges}{E^{\operatorname{bad}}}
\newcommand{\BadVert}{V^{\operatorname{bad}}}
\newcommand{\BadEdgesBefore}{\overline E}
\DeclareMathOperator{\dbl}{dbl}
\newcommand{\ground}{\omega_{\operatorname{gnd}}}
\DeclareMathOperator{\TrivialLoop}{\mathsf{TrivLoop}}
\newcommand{\clr}{{{\mathsf{c}}}}
\newcommand{\breakup}{{{\cC}}}
\newcommand{\shiftFunc}{R_\gamma}
\newcommand{\shift}[1]{{\shiftFunc(#1)}}
\newcommand{\LC}{\mathsf{LoopConf}}

\newcommand{\DIRup}{\,\uparrow\,}
\newcommand{\DIRdown}{\,\downarrow\,}

\title{Exponential decay of loop lengths in the loop $O(n)$ model with large $n$}
\date{\today}

\author{Hugo Duminil-Copin}
\thanks{Research of H.D.-C. is supported by Swiss FNS, the ERC AG COMPASP and the NCCR SwissMap}
\address{Hugo Duminil-Copin\hfill\break
    Universit\'e de Gen\`eve\\
    D\'epartement de math\'ematiques\\
    1211 Gen\`eve 4, Switzerland.}
\email{hugo.duminil@unige.ch}
\urladdr{http://www.unige.ch/~duminil}

\author{Ron Peled}
\thanks{Research of R.P. and Y.S. is partially supported by an ISF grant and an IRG
grant.}
\address{Ron Peled\hfill\break
    Tel Aviv University\\
    School of Mathematical Sciences\\
    Tel Aviv, 69978, Israel.}
\email{peledron@post.tau.ac.il}
\urladdr{http://www.math.tau.ac.il/~peledron}

\author{Wojciech Samotij}
\thanks{Research of W.S. is partially supported by an ISF grant.}
\address{Wojciech Samotij\hfill\break
    Tel Aviv University\\
    School of Mathematical Sciences\\
    Tel Aviv, 69978, Israel.}
\email{samotij@post.tau.ac.il}
\urladdr{http://www.math.tau.ac.il/~samotij}

\author{Yinon Spinka}
\address{Yinon Spinka\hfill\break
    Tel Aviv University\\
    School of Mathematical Sciences\\
    Tel Aviv, 69978, Israel.}
\email{yinonspi@post.tau.ac.il}
\urladdr{http://www.math.tau.ac.il/~yinonspi}

\begin{document}

\begin{abstract}
The loop $O(n)$ model is a model for a random collection of non-intersecting loops on
the hexagonal lattice, which is believed to be in the same
universality class as the spin $O(n)$ model. It has been conjectured
that both the spin and the loop $O(n)$ models exhibit exponential decay of
correlations when $n>2$. We verify this for the loop $O(n)$ model
with large parameter $n$, showing that long loops are exponentially
unlikely to occur, uniformly in the edge weight $x$. Our proof
provides further detail on the structure of typical configurations
in this regime. Putting appropriate boundary conditions, when $nx^6$
is sufficiently small, the model is in a dilute, disordered phase in
which each vertex is unlikely to be surrounded by any loops,
whereas when $nx^6$ is sufficiently large, the model is in a dense,
ordered phase which is a small perturbation of one of the three
ground states.
\end{abstract}

\maketitle

\section{Introduction}\label{sec:introduction}

After the introduction of the Ising model \cite{Len20} and Ising's
conjecture that it does not undergo a phase transition, physicists
tried to find natural generalizations of the model with richer
behavior. In~\cite{HelKra34}, Heller and Kramers described the
classical version of the celebrated quantum Heisenberg model where
spins are vectors in the (two-dimensional) unit sphere in dimension
three. Later, Stanley introduced the \emph{spin $O(n)$ model} by
allowing spins to take values in higher-dimensional spheres
\cite{Sta68}. We refer the interested reader to~\cite{DomGre76} for
a history of the subject.

Formally, a configuration of the spin $O(n)$ model on a finite graph
$G$ is an assignment $\sigma\in\Omega:=(\sqrt n \cdot\mathbb
S^{n-1})^{V(G)}$ of spins to each vertex of $G$, where $\mathbb
S^{n-1} \subseteq \R^n$ is the $(n-1)$-dimensional unit sphere and
the choice of the radius $\sqrt{n}$ serves as a convenient
normalization. The {\em Hamiltonian} of the model is defined by
\[
\mathcal{H}_{G,n}(\sigma)~:=~-\sum_{\{u,v\}\in
E(G)}\left\langle\sigma_u,\sigma_v\right\rangle,
\]
where $\left\langle\cdot,\cdot\right\rangle$ denotes the scalar
product in $\R^n$. At inverse temperature $\beta$, we define the
finite-volume Gibbs measure $\mu_{G,n,\beta}$ to be the probability measure
on $\Omega$ given by
\begin{equation*}
  d\mu_{G,n,\beta}(\sigma) := \frac{1}{Z^{\text{spin}}_{G,n,\beta}} \exp \left[-\beta \mathcal{H}_{G,n}(\sigma)\right]
  d\sigma,
\end{equation*}
where $Z^{\text{spin}}_{G,n,\beta}$, the \emph{partition function},
is given by
\begin{equation}\label{eq:Z_def}
  Z^{\text{spin}}_{G,n,\beta}:=\int_{\Omega} \exp \left[-\beta
\mathcal{H}_{G,n}(\sigma)\right]
  d\sigma
\end{equation}
and $d\sigma$ is the uniform probability measure on $\Omega$ (i.e.,
the product measure of the uniform distributions on $\sqrt n
\cdot\mathbb S^{n-1}$ for each vertex in $G$).

By taking the weak limit of measures on larger and larger subgraphs of an infinite planar lattice, such as $\Z^2$ or the hexagonal lattice $\HH$, an infinite-volume measure $\mu_{n,\beta}$ can be defined, and one may ask whether a phase transition occurs at some critical inverse temperature. From this point of view, the behavior of the model is very different for different values of $n$:
\begin{itemize}
\item
  For $n=1$, the model is simply the Ising model, which is known to undergo a phase transition between an ordered and a disordered phase, as proved by Peierls~\cite{Pei36} (refuting Ising's conjecture). The critical inverse temperature has been computed for the square and the hexagonal lattices and it is fair to say that a lot is known about the behavior of the model. We refer the reader to~\cite{Dum13,DumSmi12a,Pal07} and references therein for an overview of the recent progress on the subject.
\item
  For $n=2$, the model is the so-called XY model (first introduced in \cite{VakLar66}). Since the spin space $\mathbb S^1$ is a continuous group, the Mermin--Wagner theorem~\cite{MerWag66} guarantees that there is no phase transition between ordered and disordered phases. Still, a Kosterlitz--Thouless phase transition occurs as proved in~\cite{FroSpe81,KosTho73,McBSpe77,Tho69}. That is, below some critical inverse temperature, the spin-spin correlations $\mu_{n,\beta}[\langle \sigma_u,\sigma_v\rangle]$ decay exponentially fast in the distance between $u$ and $v$, while above this critical inverse temperature, they decay only like an inverse power of the distance.
\item
  For $n\ge 3$, it is predicted that no phase transition occurs~\cite{Pol75} and that spin-spin correlations decay exponentially fast at every positive temperature. The $n=3$ case, corresponding to the classical Heisenberg model, is of special interest. Let us mention that this prediction is part of a more general conjecture asserting that planar spin systems with non-Abelian continuous spin space do not exhibit a phase transition. As of today, the $n\ge 3$ case remains wide open. The best known results in this direction can be found in~\cite{Kup80}, where a $1/n$ expansion is performed as $n$ tends to infinity.
\end{itemize}

On the hexagonal lattice $\HH$, the spin $O(n)$ model can be related
to the so-called \emph{loop $O(n)$ model} introduced
in~\cite{DomMukNie81}. Before providing additional details on the
relation, let us define the loop $O(n)$ model. A \emph{loop} is a
finite subgraph of $\HH$ which is isomorphic to a simple cycle. A
\emph{loop configuration} is a spanning subgraph of $\HH$ in which
every vertex has even degree; see Figure~\reffig{fig:loop-configs}. The
\emph{non-trivial finite} connected components of a loop
configuration are necessarily loops, however, a loop configuration
may also contain isolated vertices and infinite simple paths. We
shall often identify a loop configuration with its set of edges,
disregarding isolated vertices. In this work, a \emph{domain} $H$ is
a non-empty finite connected induced subgraph of $\HH$ whose
complement $\VH \setminus V(H)$ induces a connected subgraph of
$\HH$ (in other words, it does not have ``holes''). For convenience,
all of our results will be stated for domains, although the
definitions and techniques may sometimes be applied in greater
generality. Given a domain $H$ and a loop configuration $\xi$, we
denote by $\LC(H,\xi)$ the collection of all loop configurations
$\omega$ that agree with $\xi$ on $\EH \setminus E(H)$. Finally,
for a domain $H$ and a loop configuration $\omega$, we denote by
$L_H(\omega)$ the number of loops in $\omega$ which \emph{intersect}
$E(H)$ and by $o_H(\omega)$ the number of edges of $\omega \cap E(H)$.

\begin{definition}\label{def:loop-model}
Let $H$ be a domain and let $\xi$ be a loop configuration. Let $n$ and $x$ be
positive real numbers. The loop $O(n)$ measure on $H$ with edge
weight $x$ and boundary conditions $\xi$ is the probability measure
$\Pr_{H,n,x}^\xi$ on $\LC(H,\xi)$ defined by
  \[
  \Pr_{H,n,x}^\xi(\omega) := \frac{x^{o_H(\omega)} n^{L_H(\omega)}}{Z_{H,n,x}^\xi}, \quad \omega\in\LC(H,\xi),
  \]
  where $Z_{H,n,x}^\xi$ is the unique constant which makes $\Pr_{H,n,x}^\xi$ a probability measure.
\end{definition}
We note that the loop $O(n)$ model is defined for any \emph{real}
$n>0$ whereas the spin $O(n)$ model is only defined for positive
\emph{integer} $n$ (the loop $O(n)$ model may be defined also with
$n=0$ by taking the limit $n\to0$, giving rise to a self-avoiding
walk model). Let us now briefly discuss the connection between the
loop and the spin $O(n)$ models (with integer $n$) on a domain $H
\subset \HH$. Rewriting the partition function
$Z^{\text{spin}}_{H,n,\beta}$ given by \eqref{eq:Z_def} using the
approximation $e^t \approx 1+t$ gives
\[
\begin{split}
  Z^{\text{spin}}_{H,n,\beta}=\int\displaylimits_{\Omega} \prod_{\{u,v\}\in E(H)} e^{\beta \langle  \sigma_u,\sigma_v \rangle} \,d\sigma
  &\approx \int\displaylimits_{\Omega} \prod_{\{u,v\}\in E(H)} (1 + \beta \langle  \sigma_u,\sigma_v \rangle) \,d\sigma \\
  &= \sum_{\omega \subset E(H)} \beta^{o_H(\omega)} \int\displaylimits_{\Omega} \prod_{\{u,v\} \in E(\omega)} \langle  \sigma_u,\sigma_v \rangle \,d\sigma.
\end{split}
\]
The integral on the right-hand side equals $n^{L_H(\omega)}$ if
$\omega\in\LC(H,\emptyset)$ and 0 otherwise; see
Appendix~\ref{sec:integrals} for the calculation. Here, the
normalization of taking spins on the sphere of radius $\sqrt{n}$ is
used. Hence, substituting $x$ for $\beta$,
\[
Z_{H,n,x}^{\text{spin}} \approx \sum_{\omega \in
\LC(H,\emptyset)} x^{o_H(\omega)} n^{L_H(\omega)} =
Z_{H,n,x}^\emptyset.
\]
In the same manner, the spin-spin correlation of $u,v\in V(H)$ may
be approximated as follows.
\begin{equation}\label{eq:relation spin loop}
  \mu_{H,n,x}[\langle  \sigma_u,\sigma_v \rangle] = \frac{ \displaystyle\int_{\Omega} \langle \sigma_u, \sigma_v \rangle \exp \left[-x \mathcal{H}_{H,n}(\sigma)\right]
  d\sigma}{Z_{H,n,x}^{\text{spin}}} \approx n\cdot \frac{\displaystyle \sum_{\lambda \in \LC(H,\emptyset,u,v)} x^{o_H(\lambda)} n^{L'_H(\lambda)} J(\lambda)}{\displaystyle \sum_{\omega \in \LC(H,\emptyset)} x^{o_H(\omega)} n^{L_H(\omega)}},
\end{equation}
where $\LC(H,\emptyset,u,v)$ is the set of spanning subgraphs of $H$
in which the degrees of $u$ and $v$ are odd and the degrees of all
other vertices are even. Here, for $\lambda \in \LC(H,\emptyset,u,v)$, $o_H(\lambda)$ is the number of edges of $\lambda$, $L'_H(\lambda)$ is the number of loops in $\lambda$ after removing an arbitrary simple path in $\lambda$ between $u$ and $v$, and
$J(\lambda):=\tfrac{3n}{n+2}$ if there are three disjoint paths in
$\lambda$ between $u$ and $v$ and $J(\lambda):=1$ otherwise (in which
case, there is a unique simple path in $\lambda$ between $u$ and
$v$); see Appendix~\ref{sec:integrals} for the calculation.

Unfortunately, the above approximation is not justified for any
$x>0$. Nevertheless, \eqref{eq:relation spin loop} provides a
heuristic connection between the spin and the loop $O(n)$ models and
suggests that both these models reside in the same universality
class. For this reason, it is natural to ask whether the prediction
about the absence of phase transition is valid for the loop $O(n)$
model.

\begin{quest}
  \label{quest:loop-On-quest}
  Does the quantity on the right-hand side of~\eqref{eq:relation spin loop} decay exponentially fast in the distance between $u$ and $v$, uniformly in the domain $H$, whenever $n>2$ and $x>0$?
\end{quest}

In this article, we partially answer this question. In
Theorem~\ref{thm:spin spin} below, we show that for all sufficiently
large $n$ and any $x>0$, the quantity on the right-hand side
of~\eqref{eq:relation spin loop} decays exponentially fast for a
large class of domains $H$. The theorem is a consequence of a more
detailed understanding of the loop $O(n)$
model. We show that for small $x$ the model is in a dilute,
disordered phase, where the sampled loop configuration is rather
sparse and the probability of seeing long loops surrounding a given
vertex decays exponentially in the length (see
Figure~\reffig{fig:loop-sample-n=8,x=0.5}). For large $x$, the same
exponential decay holds but for a different reason. There, the model
is in a dense, ordered phase, which is a perturbation of a periodic
ground state. In the ground state all loops have length $6$ and a typical
perturbation does not make them significantly longer (see
Figure~\reffig{fig:loop-sample-n=8,x=2}).

\medbreak \noindent {\bf The $x=\infty$ Model.} We shall also
consider the limit of the loop $O(n)$ model as the edge weight $x$
tends to infinity. This means restricting the model to `optimally
packed loop configurations', i.e., loop configurations having the
maximum possible number of edges.
\begin{definition}
  Let $H$ be a domain and let $\xi$ be a loop configuration.
  For $n>0$, the loop $O(n)$ measure on $H$ with edge weight $x = \infty$ and boundary conditions $\xi$ is the probability measure on $\LC(H,\xi)$ defined by
\[
\Pr_{H,n,\infty}^\xi(\omega) := \lim_{x \to \infty} \Pr_{H,n,x}^\xi(\omega) =
 \begin{cases} \frac{n^{L_{H}(\omega)}}{Z_{H,n,\infty}^\xi} &\text{if }o_{H}(\omega)=o_{H,\xi} \\ 0 &\text{otherwise} \end{cases} , \quad \omega\in\LC(H,\xi),
\]
where $o_{H,\xi} := \max \{ o_H(\omega) : \omega \in \LC(H,\xi) \}$ and $Z_{H,n,\infty}^\xi$ is the unique constant making
$\Pr_{H,n,\infty}^\xi$ a probability measure.
\end{definition}
We note that if a loop configuration $\omega\in\LC(H,\xi)$ is
\emph{fully packed}, i.e., every vertex in $V(H)$ has degree $2$,
then $\omega$ is optimally packed, i.e., $o_H(\omega)=o_{H,\xi}$.

\medbreak

Before concluding this section, let us mention that the loop $O(n)$
model with $n\le 2$ is also of great interest; see
Section~\ref{sec:discussion_open_questions} for a discussion.

%%%%%%%%%%%%%%%%%%%%%%%%%%%%%%%%%%%%%%%%%%%%%%%%%%%%%%%%%%%%%%%%%%%%%%%%%%%%%%%%

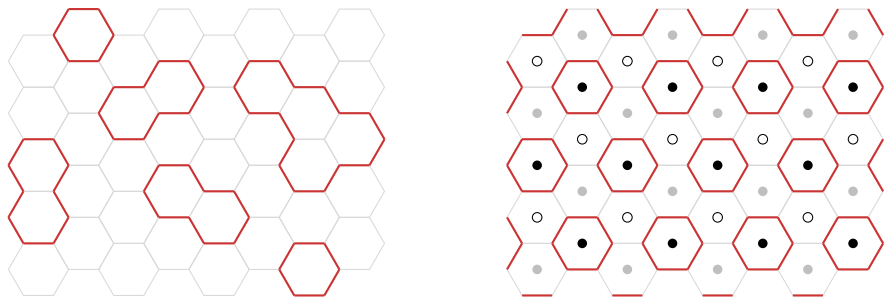
\begin{figure}
    \centering

    \begin{tabular}{lcr}

    \begin{tikzpicture}[scale=0.4, every node/.style={scale=0.4}]
        \hexagonGrid[7][4][0];
        \trivialLoop[1][4];
        \trivialLoop[6][-3];
        \doubleLoopUp[0][1];
        \doubleLoopRight[2][2];
        \doubleLoopLeft[3][-1];
        \hexagonEdges[edge-on][5][0][0/1/-3,0/1/-2,0/1/-1,0/1/0,0/1/1,1/0/0,1/0/1,2/-1/0,2/-1/1,2/-1/2,2/-1/3,1/-1/2,1/-1/3,1/-1/4,1/-1/5,1/0/-2];
    \end{tikzpicture}

    &\,\qquad\,&

    \begin{tikzpicture}[scale=0.4, every node/.style={scale=0.35}]
        \groundState[7][4][1];
        \hexagonEdges[edge-on][0][0][0/0/-1,0/1/-2,0/3/-1,0/4/-2];
        \hexagonEdges[edge-on][0][0][0/0/3,2/-1/3,4/-2/3,6/-3/3];
        \hexagonEdges[edge-on][7][-2][0/0/1,0/1/2];
        \hexagonEdges[edge-on][0][5][0/0/3,1/-1/-1,1/-1/1,2/-1/3,3/-2/-1,3/-2/1,4/-2/3,5/-3/-1,5/-3/1,6/-3/3,7/-4/-1,7/-4/1];
    \end{tikzpicture}

    \end{tabular}

    \caption{On the left, a loop configuration. On the right, a proper $3$-coloring of the triangular lattice $\T$ (the dual of the hexagonal lattice $\HH$), inducing a partition of $\T$ into three color classes $\T^0$, $\T^1$, and $\T^2$. The $0$-phase ground state $\ground^0$ is the (fully-packed) loop configuration consisting of trivial loops around each hexagon in~$\T^0$.}
    \label{fig:loop-configs}
\end{figure}

%\newgeometry{left=15mm,bottom=15mm,top=10mm}
\begin{figure}
    \centering
    \begin{subfigure}[t]{.5\textwidth}
        \includegraphics[scale=0.46]{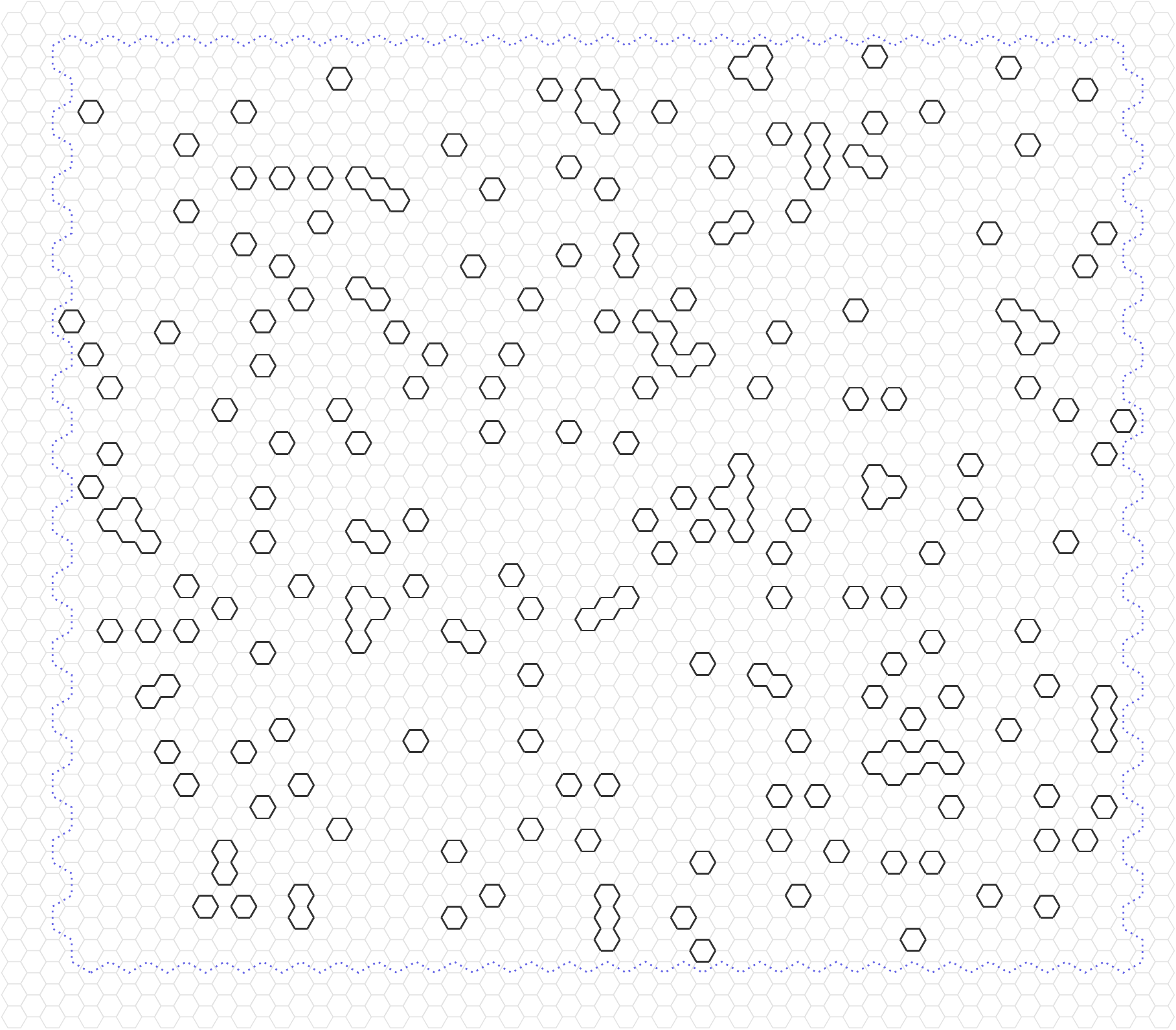}
        \caption{$n=8$ and $x=0.5$. Theorem~\ref{thm:Gibbs measures_small_x} shows that the limiting measure is unique for domains with vacant boundary conditions when $x$ is small.}
        \label{fig:loop-sample-n=8,x=0.5}
    \end{subfigure}%
    \begin{subfigure}{20pt}
        \quad
    \end{subfigure}%
    \begin{subfigure}[t]{.5\textwidth}
        \includegraphics[scale=0.46]{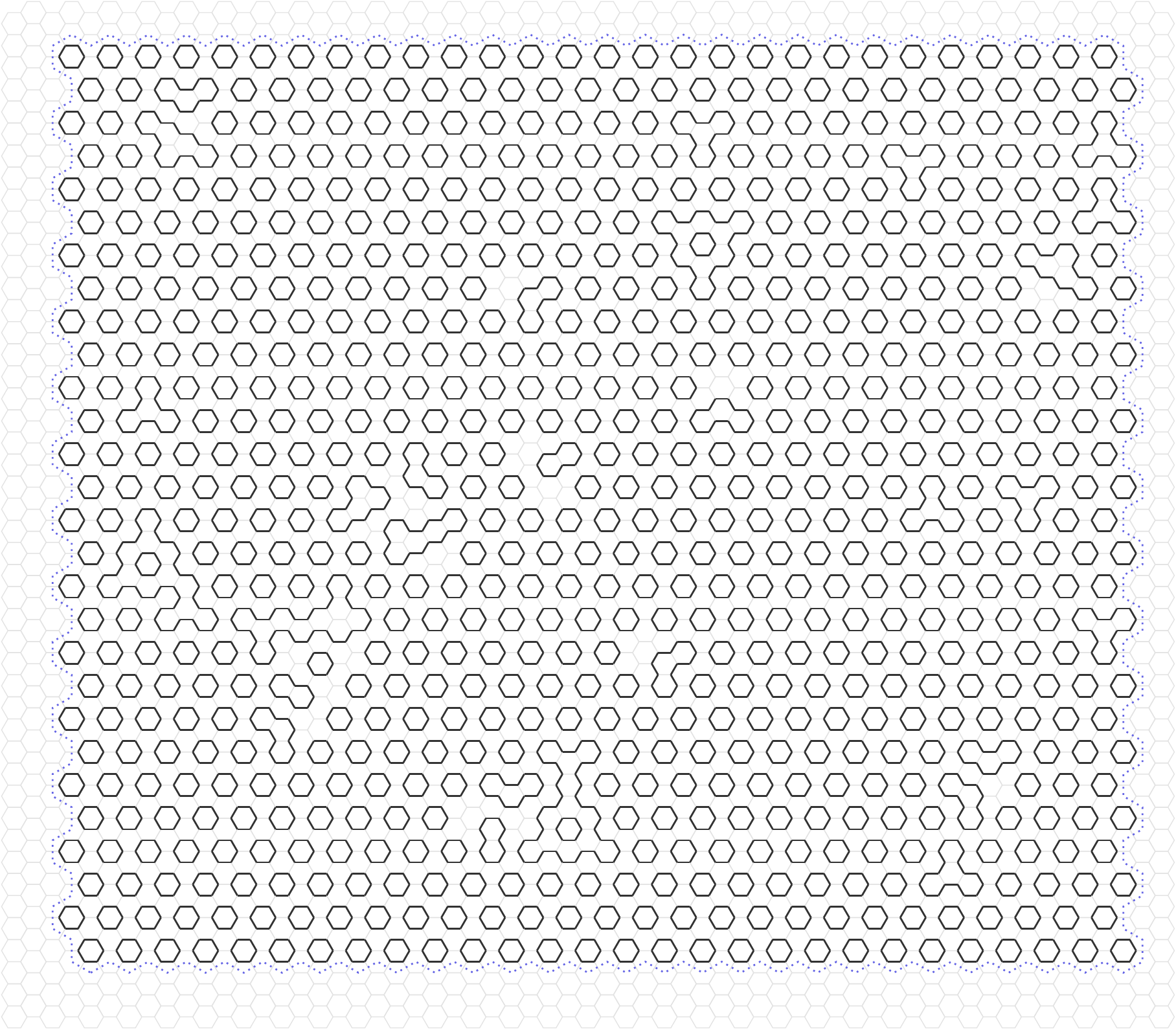}
        \caption{$n=8$ and $x=2$. Theorem~\ref{thm:small-deviations-from-the-ground-state-x-model} shows that typical configurations are small perturbations of the ground state for large $n$ and $x$.}
        \label{fig:loop-sample-n=8,x=2}
    \end{subfigure}
    \caption{Two samples of random loop configurations with large $n$. Configurations are on a $60\times45$ domain of type $0$ and are sampled via Glauber dynamics for 100 million iterations started from the empty configuration.}
    \label{fig:loop-samples-large-n}
\end{figure}
%\restoregeometry

%%%%%%%%%%%%%%%%%%%%%%%%%%%%%%%%%%%%%%%%%%%%%%%%%%%%%%%%%%%%%%%%%%%%%%%%%%%%%%%%

\subsection{Results}
\label{sec:results}

In order to state our main results, we need several more definitions
(see Figure~\reffig{fig:loop-configs} for their illustration). We
consider the triangular lattice $\T := (0,2)\Z + (\sqrt3,1)\Z$, and
view the hexagonal lattice $\HH$ as its dual lattice, obtained by
placing a vertex at the center of every face (triangle) of $\T$, so
that each edge $e$ of $\HH$ corresponds to the unique edge $e^*$ of
$\T$ which intersects $e$. Since vertices of $\T$ are identified
with faces of $\HH$, they will be called \emph{hexagons} instead of
vertices. We will also say that a vertex or an edge of $\HH$
\emph{borders} a hexagon if it borders the corresponding face of
$\HH$.

There are exactly 6 proper colorings of $\T$ with the colors
$\{0,1,2\}$. For the rest of the paper, we fix an arbitrary proper coloring and let $\T^\clr$ be the set of hexagons colored by $\clr$,
$\clr\in\{0,1,2\}$.
A \emph{trivial loop} is a loop of length exactly $6$. Define the
\emph{$\clr$-phase ground state} $\ground^\clr$ to be the
(fully-packed) loop configuration consisting of all the trivial
loops surrounding hexagons in $\T^\clr$.
We shall say that a domain $H$ is \emph{of type $\clr$}, $\clr \in
\{0,1,2\}$, if every edge $\{u,v\} \in \ground^\clr$ satisfies
either $u,v \in V(H)$ or $u,v \notin V(H)$.
Equivalently, $H$ is of type $\clr$ if and only if
\begin{equation}\label{eq:domain_of_type}
    \LC(H,\emptyset) = \{\omega\cap E(H) : \omega\in\LC(H,\ground^\clr)\}.
\end{equation}
Finally, we shall say that a loop \emph{surrounds} a vertex $u$
of $\HH$ if any infinite simple path in $\HH$ starting at $u$
intersects a vertex of this loop. In particular, if a loop passes
through a vertex then it surrounds it as well.

\begin{thm}
  \label{thm:no-large-loops}
  There exist $n_0, \alpha > 0$ such that for any $n \geq n_0$ and $x\in(0,\infty]$ the following holds. For any $\clr \in \{0,1,2\}$, any domain $H$ of type $\clr$, any $u \in V(H)$ and any integer $k > 6$, we have
  \[
  \Pr_{H,n,x}^\emptyset(\text{there exists a loop of length $k$ surrounding $u$}) \leq n^{-\alpha k}.
  \]
\end{thm}

As follows from Theorem~\ref{thm:small-deviations-from-the-ground-state-x-model} below, when $n$ and $nx^6$ are sufficiently large, it is likely that $u$ is contained in a trivial loop. Thus, the assumption that $k>6$ is necessary.
The techniques involved in the proof of Theorem~\ref{thm:no-large-loops} also imply the following result, which partially answers Question~\ref{quest:loop-On-quest}.

\begin{thm}
  \label{thm:spin spin}
  There exist $n_0,\alpha > 0$ such that for any $n \geq n_0$ and any $x>0$ the following holds.
  For any $\clr \in \{0,1,2\}$, any domain $H$ of type $\clr$ and any distinct non-adjacent $u,v \in V(H)$, we have
  \[
  \frac{\displaystyle\sum_{\lambda \in \LC(H,\emptyset,u,v)} x^{o_H(\lambda)} n^{L'_H(\lambda)} J(\lambda)}{\displaystyle\sum_{\omega \in \LC(H,\emptyset)} x^{o_H(\omega)} n^{L_H(\omega)}} \le x \cdot n^{-\alpha \, d_H(u,v)},
  \]
where $d_H(u,v)$ is the graph distance in $H$ between $u$ and $v$.
\end{thm}

%\subsection{The Gibbs measures}

Our techniques provide additional information on the (infinite-volume) Gibbs measures
of the loop $O(n)$ model. We recall the standard definition: a probability measure $\Pr$ on the set of loop configurations on $\HH$ (viewed as a subset of $\{0,1\}^{E(\HH)}$) is a Gibbs measure for the loop $O(n)$ model with edge weight $x$ if for any domain $H$ and $\Pr$-almost every loop configuration $\xi$, the distribution of the configuration $\omega$, conditioned that $\omega\in\LC(H,\xi)$, is given by $\Pr^{\xi}_{H,n,x}$.

For small parameter $x$, under vacant
boundary conditions, the model is in a dilute, disordered phase,
where loops are rare and tend to be short; see
Figure~\reffig{fig:loop-sample-n=8,x=0.5}. This is relatively simple to
show and is proved in Corollary~\ref{cl:no-large-loops-for-small-x}.
A consequence of this fact is the existence of a unique limiting
Gibbs measure when exhausting the hexagonal lattice $\HH$ via
domains with vacant boundary conditions.
\begin{thm}\label{thm:Gibbs measures_small_x}
  There exists $c>0$ such that for any $n > 0$ and $0<x\le c$ satisfying $nx^6\le c$ the following holds. Let $H_k$ be an increasing sequence of domains
  satisfying $\cup_k H_k = \HH$. Then the measures $\Pr_{H_k,n,x}^\emptyset$ converge (weakly) as $k \to \infty$ to an infinite-volume Gibbs measure $\Pr_{\HH,n,x}$ which is supported on loop configurations with no infinite paths.
\end{thm}

It follows that the limiting measure $\Pr_{\HH,n,x}$ does not depend
on the specific choice of exhausting sequence $(H_k)$ as one may
interleave two such sequences to obtain another convergent sequence.
Consequently, it also follows that $\Pr_{\HH,n,x}$ is invariant
under automorphisms of $\HH$. Our proofs apply also when one allows
$H_k$ to be arbitrary finite subgraphs of $\HH$ rather than domains,
but we do not state this explicitly as our work is mostly concerned
with domains. The restriction to vacant boundary conditions is,
however, essential for our proofs with the difficulty stemming from
the fact that non-vacant boundary conditions may force the existence
of long paths in the configuration (see
Figure~\reffig{fig:domain-unique-min-edge-loop-config}). Still, it may
be that there is a unique Gibbs measure in this regime of small $x$
and we provide a discussion of this in
Section~\ref{sec:discussion_open_questions}.

For large parameter $x$ and large $n$, the situation changes
dramatically. Here, we obtain that the model is in a dense, ordered
phase, where, under the $\ground^\clr$ boundary conditions, a typical
configuration is a perturbation of that ground state.
As a consequence of this structure, the model has at least three different limiting Gibbs measures in this regime of $n$ and $x$. We state this precisely in the following theorem.
To lighten the notation, we write $\Pr_{H,n,x}^\clr$ for
the loop $O(n)$ measure on $H$ with boundary conditions
$\ground^\clr$.

\begin{thm}\label{thm:Gibbs measures_large_x}
  There exists $C>0$ such that for any $n\ge C$ and any $x\in(0,\infty]$ satisfying $nx^6 \ge C$ the following holds. Let $H_k$ be an increasing sequence of domains satisfying $\cup_k H_k = \HH$. Then, for every $\clr\in\{0,1,2\}$, the measures
$\Pr_{H_k,n,x}^{\clr}$ converge (weakly) as $k \to \infty$ to an
infinite-volume Gibbs measure $\Pr_{\HH,n,x}^\clr$ which is supported on loop configurations with no infinite paths.
Furthermore, no one of the limiting measures is a convex combination of the other two.
\end{thm}
Similarly to before, it follows that, for each $\clr\in\{0,1,2\}$,
the limiting measure $\Pr_{\HH,n,x}^\clr$ does not depend on the
specific choice of exhausting sequence $(H_k)$ and that
$\Pr_{\HH,n,x}^\clr$ is invariant under automorphisms preserving the
set $\T^\clr$. However, as these measures are distinct for different
$\clr$, they are not invariant under {\em all} automorphisms. In
particular, if each $H_k$ is of type $\clr$, by
\eqref{eq:domain_of_type}, we have that $\Pr_{H_k,n,x}^{\emptyset}$
also converges to $\Pr_{\HH,n,x}^\clr$, in contrast to the behavior
obtained in Theorem~\ref{thm:Gibbs measures_small_x} for small $x$.
It would be interesting to determine whether every
infinite-volume Gibbs measure is a convex combination of these three
measures, i.e., whether these are the only \emph{extremal} Gibbs measures (see also Section~\ref{sec:discussion_open_questions}). As we remark at the end of the section, this is not the case for $x=\infty$.

As mentioned above, in the ordered regime (large $x$ and $n$), a
typical configuration drawn from $\Pr_{H,n,x}^\clr$ is a
perturbation of the $\clr$-phase ground state $\ground^\clr$ (see
Figure~\reffig{fig:loop-sample-n=8,x=2}). This is made precise in
the following theorem, which we state for the $\clr=0$ phase for
concreteness of our definitions. In order to measure how close
$\ground^0$ and a typical loop configuration are, we introduce the
notion of a \emph{breakup}. Fix a domain $H$ and let $\omega \in
\LC(H,\ground^0)$ be a loop configuration. Let $A(\omega)$ be the
set of vertices of $\HH$ belonging to trivial loops surrounding
hexagons in $\T^0$ and let $B(\omega)$ be the unique infinite
connected component of $A(\omega)$. For $u \in \HH$, define the
breakup $\breakup(\omega,u)$ of $u$ to be the connected component of
$\HH \setminus B(\omega)$ containing $u$, setting
$\breakup(\omega,u)=\emptyset$ if $u \in B(\omega)$. We also define
$\partial \breakup(\omega,u)$ to be the internal vertex boundary of
$\breakup(\omega,u)$, i.e., the set of vertices in
$\breakup(\omega,u)$ adjacent to a vertex not in
$\breakup(\omega,u)$ (thus in $B(\omega)$). We remark that
$\breakup(\omega,u)$ need not be contained in $H$, though it cannot
extend significantly beyond it in the sense that it is contained in
any domain of type $0$ containing $H$.

\begin{thm}\label{thm:small-deviations-from-the-ground-state-x-model}
    There exists $c>0$ such that for any $n>0$, any $x \in (0,\infty]$, any domain $H$, any $u\in V(H)$ and any positive integer $k$, we have
    \[
    \Pr_{H,n,x}^0(|\partial \breakup(\omega,u)| \ge k) \le (c n \cdot \min\{x^6,1\})^{-k/15} .
    \]
\end{thm}
One should note that the above theorem contains the implicit
assumption that $n \ge C$ and $n x^6 \ge C$, as otherwise the
statement is trivial.

In this work, we mainly study the loop $O(n)$ model with either vacant
or ground state boundary conditions. To obtain a complete picture
regarding the possible Gibbs measures, one must also study the model
for general boundary conditions. As mentioned above, understanding
the Gibbs measures in each regime of $n$ and $x$, and in particular,
determining the number of extremal Gibbs measures, is an interesting
problem. Theorem~\ref{thm:Gibbs measures_small_x} and Theorem~\ref{thm:Gibbs measures_large_x} bring us closer to this
goal, providing a partial answer in the regimes $nx^6 \le c$ and
$nx^6 \ge C$, for large $n$. In this regard, one may ask what
happens in the intermediate regime, i.e., when $c<nx^6<C$ and $n$ is
large. For instance, one may ask whether or not there is a single
transition curve, perhaps of the form $nx^6=c'$. If indeed this is the
case, it would be interesting to investigate the number of extremal
Gibbs measures on this curve, determining whether there is a unique
such Gibbs measure (as Theorem~\ref{thm:Gibbs measures_small_x} suggests for $nx^6 \le c$),
$3$ such measures (as Theorem~\ref{thm:Gibbs measures_large_x} suggests for $nx^6 \ge C$), $4$ such measures, or perhaps a different quantity (see also Section~\ref{sec:discussion_open_questions}).

\begin{remark}
  For $x=0$ and $x=\infty$, many other Gibbs measures can be constructed. For instance, for positive integers $a$ and $b$, let $H_{a,b}$ be the ``rectangle'' of width $2a+1$ and height $b$ (measured in hexagons) with the origin at the center, as in Figure~\reffig{fig:domain-unique-fully-packed-loop-config} (on the left). It is not hard to check that the configuration depicted in the figure is the unique fully-packed loop configuration (with vacant boundary conditions) inside $H_{a,b}$. Thus, the probability measure $\Pr_{H_{a,b},n,\infty}^\emptyset$ is supported on a single configuration. The measures $\Pr_{H_{a,b},n,\infty}^\emptyset$ converge (as $a,b \to \infty$) to a delta measure on the configuration of infinite vertical paths covering the entire lattice (which is a Gibbs measure of the loop $O(n)$ model with edge weight $\infty$). By considering different domains, one may construct many more examples of this nature (once again, see Figure~\reffig{fig:domain-unique-fully-packed-loop-config}).
  One may also look at the limiting model as $x$ tends to $0$, which corresponds to requiring the configuration to have the minimal number of edges. For the vacant boundary conditions, the finite-volume measure is a Dirac measure on the empty configuration. Using alternative boundary conditions, one may construct several distinct Gibbs measures (see, e.g., Figure~\reffig{fig:domain-unique-min-edge-loop-config}).
\end{remark}

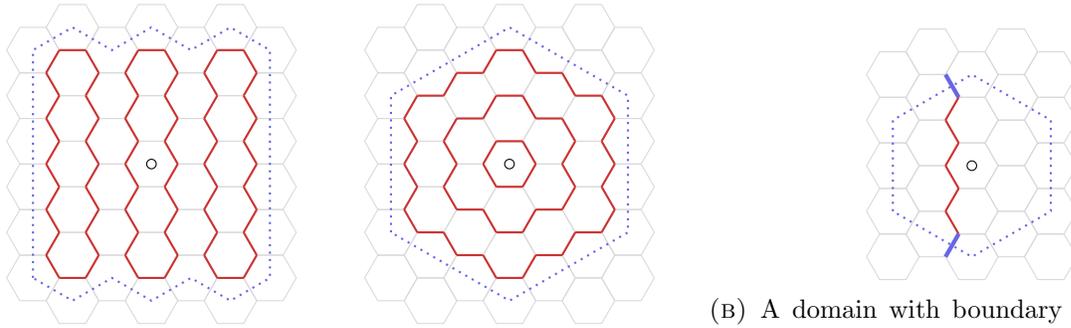
\begin{figure}
    \centering
    \begin{subfigure}{0.53\textwidth}

   \begin{tabular}{lll}

    \begin{tikzpicture}[scale=0.35, every node/.style={scale=0.35}]
        \hexagonGrid[6][5][0];
        \hexagon[1][-1][1];\hexagon[3][-2][1];\hexagon[5][-3][1];
        \hexagonEdges[edge-on][1][-1][0/1/3,0/1/2,0/1/-2,0/1/1,0/1/-1,0/2/2,0/2/-2,0/2/1,0/2/-1,0/3/2,0/3/-2,0/3/1,0/3/-1,0/4/2,0/4/-2,0/4/1,0/4/-1,0/5/2,0/5/-2,0/5/1,0/5/-1,0/5/0];
        \hexagonEdges[edge-on][3][-2][0/1/3,0/1/2,0/1/-2,0/1/1,0/1/-1,0/2/2,0/2/-2,0/2/1,0/2/-1,0/3/2,0/3/-2,0/3/1,0/3/-1,0/4/2,0/4/-2,0/4/1,0/4/-1,0/5/2,0/5/-2,0/5/1,0/5/-1,0/5/0];
        \hexagonEdges[edge-on][5][-3][0/1/3,0/1/2,0/1/-2,0/1/1,0/1/-1,0/2/2,0/2/-2,0/2/1,0/2/-1,0/3/2,0/3/-2,0/3/1,0/3/-1,0/4/2,0/4/-2,0/4/1,0/4/-1,0/5/2,0/5/-2,0/5/1,0/5/-1,0/5/0];
         \begin{scope}[yshift=-2*0.866cm, xscale=1.5, yscale=0.866]
            \node [node2] at (3,7) {};
            \draw [domain-path] {(0,2)--(0,12)--(1,13)--(2,12)--(3,13)--(4,12)--(5,13)--(6,12)--(6,2)--(5,1)--(4,2)--(3,1)--(2,2)--(1,1)--cycle };
        \end{scope}
    \end{tikzpicture}

    & \,\,\, &

    \begin{tikzpicture}[scale=0.35, every node/.style={scale=0.35}]
        \hexagonGrid[6][5][0];
        \hexagon[1][-1][1];\hexagon[3][-2][1];\hexagon[5][-3][1];
        \hexagonEdgesFullRow[edge-on][-3][1][5/0,5/1,6/0,7/0,7/5];
        \hexagonEdgesFullRow[edge-on][-3][2][3/1,4/0,4/1,5/0,5/1,6/0,7/0,7/5,8/0,8/5,9/5];
        \hexagonEdgesFullRow[edge-on][-3][3][3/1,4/5,4/1,5/5,5/1,6/0,7/5,7/1,8/5,8/1,9/5];
        \hexagonEdgesFullRow[edge-on][-3][4][3/1,4/5,4/1,5/0,5/5,6/0,6/5,6/1,7/0,7/1,8/5,8/1,9/5];
        \hexagonEdgesFullRow[edge-on][-3][5][4/0,4/5,5/0,5/5,6/0,6/5,6/1,7/0,7/1,8/0,8/1];
        \hexagonEdgesFullRow[edge-on][-3][6][6/0,6/5,6/1];

        \begin{scope}[xscale=1.5, yscale=0.866]
            \node [node2] at (3,5) {};
            \draw [domain-path] {(0,2) -- (0,8) -- (3,11) -- (6,8) -- (6,2) -- (3,-1) -- cycle };
        \end{scope}
    \end{tikzpicture}

    \end{tabular}

    \caption{Domains for which there exists a single fully-packed loop configuration (with vacant boundary conditions). Using such domains, one may obtain many weak limits of the probability measures $\Pr_{H,n,\infty}^\emptyset$.}
    \label{fig:domain-unique-fully-packed-loop-config}

    \end{subfigure}%
    \begin{subfigure}{20pt}
        \quad
    \end{subfigure}%
    \begin{subfigure}{0.42\textwidth}
        \centering
    \begin{tikzpicture}[scale=0.35, every node/.style={scale=0.35}]
        \hexagonGrid[4][4][0];
        \hexagonEdges[edge-on][1][0][0/0/1,0/1/2,0/1/1,0/2/2,0/2/1,0/3/2];
        \hexagonEdges[edge-bnd][1][0][0/0/2,0/3/1];
        \begin{scope}[xscale=1.5, yscale=0.866]
            \node [node2] at (2,4) {};
            \draw [domain-path] {(0,2) -- (0,6) -- (2,8) -- (4,6) -- (4,2) -- (2,0) -- cycle };
        \end{scope}
    \end{tikzpicture}

    \caption{A domain with boundary conditions inducing a unique loop configuration with minimal number of edges. Such domains give rise to a Gibbs measure for $x=0$ which contains an infinite interface passing near the origin.}
    \label{fig:domain-unique-min-edge-loop-config}

    \end{subfigure}

    \caption{Constructing multiple Gibbs measures when $x=0$ or $x=\infty$ through suitable domains and boundary conditions.}
\end{figure}

%
%%
%%%
%%%%%%%%%%%%%%
%%%
%%
%

\subsection{Overview of the proof}
Our proofs make use of the following simple lemma.
\begin{lemma}
  \label{lem:prob-inequality-tool}
  Let $p,q>0$ and let $E$ and $F$ be two events in a discrete probability space. If there exists a map $\sfT \colon E \to F$ such that $\Pr(\sfT(e))\ge p\cdot\Pr(e)$ for every $e\in E$, and $|\sfT^{-1}(f)|\le q$ for every $f\in F$, then
  \[ \Pr(E) \leq \frac{q}{p}\cdot\Pr(F) .\]
\end{lemma}
\begin{proof}
We have
\[ p\cdot \Pr(E)  \leq \sum_{e \in E} \Pr(\sfT(e))
 = \sum_{e \in E} \sum_{f \in F} \Pr(f) \mathbf{1}_{\{\sfT(e)=f\}}
 = \sum_{f \in F} |\sfT^{-1}(f)|\cdot\Pr(f) \leq q\cdot \Pr(F). \qedhere \]
\end{proof}
The results for small $x$ are obtained via a fairly standard, and
short, Peierls argument, by applying the above lemma to a map
which removes loops. For details, we refer the reader to
Section~\ref{sec:exponential_decay_loop_lengths}. The main novelty
of this work lies in the study of the loop $O(n)$ model for large
$x$.

In the large $x$ regime, the idea is to apply the above lemma to a
suitably defined `repair map'.
This map takes a configuration $\omega$ sampled with $0$-phase
ground state boundary conditions (or vacant boundary
conditions in a domain of type $0$) and having a large breakup and
returns a `repaired' configuration in which the breakup is
significantly reduced. The map operates by identifying regions in
which the configuration resembles one of the three ground states.
Regions resembling the $\ground^1$ state are `shifted down' by one
hexagon to resemble $\ground^0$ and similarly regions resembling
$\ground^2$ are `shifted up' by one hexagon to resemble $\ground^0$.
Regions resembling the $\ground^0$ state are left untouched. Regions
which do not resemble any of the ground states are completely
replaced by trivial loops from the $\ground^0$ state. We show that
this yields a new loop configuration, compatible with the boundary
conditions, and having much higher probability.
To finish using Lemma~\ref{lem:prob-inequality-tool}, we further show that
the number of preimages of a given loop configuration is exponentially smaller than the probability gain. This yields the main lemma of our paper, Lemma~\ref{lem:prob-outer-circuit}, from which our results for large $x$ are later deduced. The repair map is illustrated in
Figure~\reffig{fig:proof-illustration} and is formally defined in
Section~\ref{sec:repair-map} following the definitions of
`flowers', `gardens' and `clusters' which we require to make precise
the notion of resembling a ground state.

%
%%
%%%
%%%%%%%%%%%%%%
%%%
%%
%
\subsection{Graph notation}\label{sec:preliminaries}

Throughout this paper, given a graph $G$, we shall denote its vertex and edge sets by $V(G)$ and $E(G)$, respectively. If $u,v\in V(G)$ are such that $\{u,v\}\in E(G)$, we say that $u$ and $v$ are \emph{adjacent} (or {\em neighbors}) in $G$ and we drop the dependence on $G$ if it is clear from the context. For a vertex $u$ and an edge $e$ such that $u\in e$, we say that $e$ is \emph{incident} to $u$ and that $u$ is an {\em endpoint} of $e$. For $A\subset V(G)$, we define its \emph{(vertex) boundary} $\partial A$ by
\[
\partial A := \big\{u\in A ~:~ \{u,v\}\in E(G)\text{ for some }v \not\in A \big\}.
\]

The following is a standard lemma which gives a bound on the number of connected induced subgraphs of a graph.
\begin{lemma}[{\cite[Chapter~45]{Bol06}}]\label{lem:number-of-connected-graphs}
Let $G$ be a graph with maximum degree $d \ge 3$. The number of connected subsets of $V(G)$ containing a given vertex and $k$ other vertices is at most $(e(d-1))^k$.
\end{lemma}

%
%%
%%%
%%%%%%%%%%%%%%
%%%
%%
%
\subsection{Organization of the article}
The rest of the article is structured as follows.
Section~\ref{sec:main-lemma} introduces the repair map and proves
the main lemma, Lemma~\ref{lem:prob-outer-circuit}. In
Section~\ref{sec:proofs_main_theorems}, we derive our theorems. The
statements regarding large $x$ are deduced from the main lemma
whereas the parts pertaining to small $x$, being simpler, are
obtained directly. In Section~\ref{sec:discussion_open_questions}, we
discuss several directions for future research.

\subsection{Acknowledgements}
We are grateful to two anonymous referees whose comments helped to
improve the exposition and elucidate the relation of the results
with the existing literature.

%
%%
%%%
%%%%%%%%%%%%%%
%%%
%%
%

\section{Flowers, gardens and the repair map}
\label{sec:main-lemma}

This section is devoted to the formulation and proof of the main
lemma, Lemma~\ref{lem:prob-outer-circuit}. We start by stating a few
definitions in Section~\ref{sec:definitions}. In particular, we
introduce the notions of a \emph{circuit}, \emph{$\clr$-flower},
\emph{$\clr$-garden} and \emph{$\clr$-cluster}, and gather some easy
general facts about these objects. The main lemma is stated in
Section~\ref{sec:main lemma} and the remaining sections are devoted
to its proof. Section~\ref{sec:repair-map} introduces the
repair map, which will play the role of $\sfT$ in
Lemma~\ref{lem:prob-inequality-tool}. Section~\ref{sec:compare
probabilities} compares the probability of a configuration and its
image under the repair map (which corresponds to estimating $p$ in
Lemma~\ref{lem:prob-inequality-tool}). Section~\ref{sec:proof of
lemma} gathers the last ingredients (mainly an estimate for the
number of possible preimages under the repair map, which
corresponds to bounding $q$ in Lemma~\ref{lem:prob-inequality-tool})
to conclude the proof of Lemma~\ref{lem:prob-outer-circuit}.

\subsection{Definitions and gardening}
\label{sec:definitions}

A {\em circuit} is a simple closed path in $\T$, which may be viewed as a sequence of hexagons $\gamma=(\gamma_0,\dots,\gamma_m)$, $m \ge 3$, satisfying the following two properties:
\begin{itemize}[noitemsep,nolistsep]
    \item $\gamma_m=\gamma_0$ and $\gamma_i\ne \gamma_j$ for every $0\le i<j<m$,
    \item $\gamma_i$ and $\gamma_{i+1}$ are neighbors (in $\T$) for every $0\le i<m$.
\end{itemize}
Define $\gamma^*$ to be the set of edges
$\{\gamma_i,\gamma_{i+1}\}^* \in \EH$ for $0\le i<m$.

We proceed with three standard geometric facts regarding circuits
and domains. For completeness, these facts are proved in
Appendix~\ref{sec:circuits-and-domains}. The first two facts constitute a discrete version of the Jordan curve theorem.

\begin{fact}\label{fact:gamma-int-ext}
If $\gamma$ is a circuit then the removal of $\gamma^*$ splits $\HH$ into exactly two connected components, one of which is infinite, denoted by $\Ext\gamma$, and one of which is finite, denoted by $\Int\gamma$. Moreover, each of these are induced subgraphs of $\HH$.
\end{fact}

Let $\gamma$ be a circuit. We denote the vertex sets and edge sets
of $\Int\gamma,\Ext\gamma$ by $\IntVert\gamma,\ExtVert\gamma$ and
$\IntEdge\gamma,\ExtEdge\gamma$, respectively. Note that $\{
\IntVert\gamma, \ExtVert\gamma \}$ is a partition of $\VH$ and
that $\{ \IntEdge\gamma, \ExtEdge\gamma, \gamma^* \}$ is a partition
of $\EH$. We also define $\IntHex\gamma$ to be the set of faces
of $\Int\gamma$, i.e., the set of hexagons $z\in\T$ having all their
six bordering vertices in $\IntVert\gamma$. Since $\Int\gamma$ is
induced, this is equivalent to having all six bordering edges in
$\IntEdge\gamma$.

Note that, by Fact~\ref{fact:gamma-int-ext}, $\Int\gamma$ is a
domain. The converse is also true.
\begin{fact}\label{fact:circuit-domain-bijection}
    Circuits are in one-to-one correspondence with domains via $\gamma \leftrightarrow \Int\gamma$.
\end{fact}

Hence, every domain $H$ may be written as $H=\Int\gamma$ for some circuit $\gamma$.
Recalling the definition from Section~\ref{sec:results} of a domain of type $\clr\in\{0,1,2\}$, one should also note that $H$ is of type $\clr$ if and only if $\gamma \subset \T \setminus \T^\clr$.

\begin{fact}\label{fact:circuits-max}
    Let $\sigma$ and $\sigma'$ be two circuits such that $\sigma^* \cap (\sigma')^* \neq \emptyset$ or $\IntVert\sigma \cap \IntVert{\sigma'} \neq \emptyset$.
    Then there exists a circuit $\gamma \subset \sigma \cup \sigma'$ such that $\gamma^* \subset \sigma^* \cup (\sigma')^*$ and $\Int\sigma \cup \Int{\sigma'} \subset \Int{\gamma}$.

\end{fact}

%%%%%%%%%%%%%%%%%%%%%%%%%%%%%%%%%%%%%%%%%%%%%%%%%%%%%%%%%%%%%%%%%%%%%%%%%%%%%%%%

\begin{figure}
	\centering
	
	\begin{tikzpicture}[scale=0.4, every node/.style={scale=0.4}]
	\hexagonGridOneClass[12][6][2][1];
	\begin{scope}[yshift=1*0.866cm, xshift=1.5cm]
	\trivialLoop[3][-1];\trivialLoop[5][-2];\trivialLoop[7][-3];
	\trivialLoop[2][1];\trivialLoop[1][3];\trivialLoop[2][4];
	\trivialLoop[4][3];\trivialLoop[6][2];\trivialLoop[8][1];
	\trivialLoop[9][-1];\trivialLoop[8][-2];
	\trivialLoop[6][0];
	\doubleLoopUp[0][1];
	\trivialLoop[10][0];
	\doubleLoopRight[9][-4];
	
	\hexagonEdges[edge-on][3][1][0/1/-3,0/1/-2,0/1/-1,0/1/0,0/1/1,1/0/0,1/0/1,2/-1/-1,2/-1/-2,1/-1/2,1/-1/3,1/-1/4,1/-1/5,1/0/-2];
	\end{scope}
	\begin{scope}[yshift=1*0.866cm, xshift=1.5cm, xscale=1.5, yscale=0.866]
	\draw [domain-path]
	{(0,6)--(0,8)--(1,9)--(1,11)--(2,12)--(3,11)--(4,12)--(5,11)--(6,12)--(7,11)--(8,12)--(9,11)--(9,9)--(10,8)--(10,6)--(9,5)--(9,3)--(8,2)--(8,0)--(7,-1)--(6,0)--(5,-1)--(4,0)--(3,-1)--(2,0)--(2,2)--(1,3)--(1,5)--cycle };
	\end{scope}
	\end{tikzpicture}
	
	\caption{A garden. The dashed line denotes a vacant circuit $\sigma \subset \T \setminus \T^\clr$, where $\clr\in\{0,1,2\}$. The edges inside $\sigma$, along with the edges crossing $\sigma$, then comprise a $\clr$-garden of $\omega$, since every hexagon in $\T^\clr \cap \partial \IntHex\sigma$ is surrounded by a trivial loop.}
	\label{fig:garden}
\end{figure}
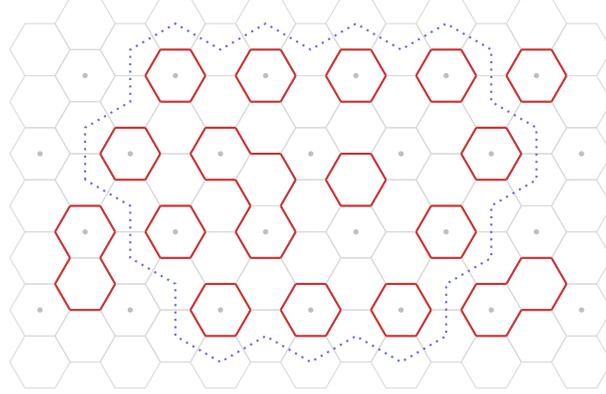

%%%%%%%%%%%%%%%%%%%%%%%%%%%%%%%%%%%%%%%%%%%%%%%%%%%%%%%%%%%%%%%%%%%%%%%%%%%%%%%%

\begin{definition}[$\clr$-flower, $\clr$-garden, vacant circuit; see Figure~\reffig{fig:garden}]
  Let $\clr\in\{0,1,2\}$ and let $\omega$ be a loop configuration. A hexagon $z\in\T^\clr$ is a {\em $\clr$-flower of} $\omega$ if it is surrounded by a trivial loop in $\omega$. A subset $E \subset \EH$ is a \emph{$\clr$-garden of $\omega$} if there exists a circuit $\sigma \subset \T \setminus \T^\clr$ such that $E=\IntEdge\sigma \cup \sigma^*$ and every $z\in \T^\clr\cap\partial \IntHex\sigma$ is a $\clr$-flower of $\omega$. In this case, we denote $\sigma(E) :=
  \sigma$. A circuit $\sigma$ is {\em vacant} in $\omega$ if $\omega\cap\sigma^*=\emptyset$.
\end{definition}

We say that $E \subset \EH$ is a garden of $\omega$ if it is a $\clr$-garden of $\omega$ for some $\clr\in\{0,1,2\}$.
We stress the fact that a garden is a {\em subset of the
edges of $\HH$}. We continue with several simple properties of
circuits, gardens and loop configurations which will be used
throughout the paper.
\begin{lemma}\label{lem:circuits-and-clusters}
    Let $\omega$ and $\omega'$ be two loop configurations.
    \begin{enumerate}[label=(\alph*), ref=\alph*, labelindent=\parindent]
        \item\label{it:circuits-and-clusters1} If $\sigma$ is a vacant circuit in $\omega$ then $\omega \cap \IntEdge\sigma$ and $\omega \cap \ExtEdge\sigma$ are loop configurations.
        \item\label{it:circuits-and-clusters2} If $E$ is a garden of $\omega$ then $\sigma(E)$ is a vacant circuit in $\omega$.
        \item\label{it:circuits-and-clusters3} If $E$ is a garden of $\omega$ then $\omega \cap E$ and $\omega \setminus E$ are loop configurations.
        \item\label{it:circuits-and-clusters4} If $\omega$ and $\omega'$ are disjoint then $\omega \cup \omega'$ is a loop configuration.
        \item\label{it:circuits-and-clusters5} If $\omega'$ is contained in $\omega$ then $\omega \setminus \omega'$ is a loop configuration.
    \end{enumerate}
\end{lemma}
\begin{proof}
    To see~\eqref{it:circuits-and-clusters1}, let $\sigma$ be a vacant circuit in $\omega$. Since any path between $\Int\sigma$ and $\Ext\sigma$ intersects $\sigma^*$, and since $\omega \cap \sigma^* = \emptyset$, every loop of $\omega$ is contained in either $\Int\sigma$ or $\Ext\sigma$, and thus,~\eqref{it:circuits-and-clusters1} follows.

    We now show~\eqref{it:circuits-and-clusters2}.
    Let $E$ be a $\clr$-garden of $\omega$, $\clr\in\{0,1,2\}$, and let $\sigma := \sigma(E)$. One of the endpoints of every edge $e\in \sigma^*$ must border a hexagon in $\T^\clr \cap \partial \IntHex\sigma$. By the definition of a $\clr$-garden, this hexagon is a $\clr$-flower, and hence, $e$ cannot belong to $\omega$. Thus, $\sigma$ is vacant in $\omega$.

    In light of~\eqref{it:circuits-and-clusters1} and~\eqref{it:circuits-and-clusters2},~\eqref{it:circuits-and-clusters3} is immediate.

    To establish~\eqref{it:circuits-and-clusters4}, it suffices to show that no vertex has degree $3$ in $\omega' \cup \omega$. Indeed, if a vertex has degree $3$ then one of the edges incident to it must be contained in both $\omega$ and $\omega'$, which is a contradiction.

    Finally, the last statement is straightforward.
\end{proof}

\begin{lemma}\label{lem:circuit-interior-degree}
    Let $\clr\in\{0,1,2\}$, let $\sigma \subset \T \setminus \T^\clr$ be a circuit, let $z \in \T^\clr$ be a hexagon and let $V(z)$ denote the six vertices in $\HH$ bordering $z$. Then
    \[ z \in \IntHex\sigma \quad\iff\quad V(z) \cap \IntVert\sigma \neq \emptyset .\]
\end{lemma}
\begin{proof}
    Recall that, by definition, $z \in \IntHex\sigma$ if and only if $V(z) \subset
    \IntVert\sigma$.
    Thus, it suffices to check that if $v \in V(z) \cap \IntVert\sigma$ and $u \in V(z)$ is adjacent to $v$ then $u \in \IntVert\sigma$. Indeed this is the case, as otherwise, $\{u,v\} \in \sigma^*$ and $z \in \sigma$, which contradicts the assumption that $\sigma \subset \T \setminus \T^\clr$.
\end{proof}

%\begin{lemma}\label{lem:size-of-boundary}
%For any simple closed path $\gamma$ in $\T \setminus \T^0$, we have
%\[ |\partial \Int(\gamma)| = |\gamma| .\]
%\end{lemma}
%\begin{proof}
%It suffices to show a bijection between $\gamma^*$ and $\partial \Int(\gamma)$.
%Observe that any edge $e \in \gamma^*$ has one endpoint in $\Int(\gamma)$ and another in $\Ext(\gamma)$. Thus, it has a unique endpoint in $\partial \Int(\gamma)$, and so we may map such an edge to this endpoint.
%Clearly, every vertex in $\partial \Int(\gamma)$ has an edge in $\gamma^*$ that is mapped to it.
%Moreover, this mapping is injective, since otherwise, there exist $v \in \Int(\gamma)$ and distinct $u,u' \sim v$ such that $(v,u),(v,u') \in \gamma^*$, contradicting the fact that $\gamma^* \subset  \E^0$.
%\end{proof}

We proceed to discuss disjointness and containment properties of
gardens.
\begin{lemma}\label{cl:gardens-same-class}
  Let $\omega$ be a loop configuration and let $E_1$ and $E_2$ be two $\clr$-gardens of $\omega$ for some $\clr\in\{0,1,2\}$. If there exists a vertex which is the endpoint of an edge in $E_1$ and an edge in $E_2$, then $E_1\cup E_2$ is contained in a $\clr$-garden of $\omega$.
\end{lemma}
\begin{proof}
    Denote $\sigma_1 := \sigma(E_1)$ and $\sigma_2 := \sigma(E_2)$.
Let us first show that necessarily
$\IntVert{\sigma_1}\cap\IntVert{\sigma_2}\neq\emptyset$ or
$\sigma_1^* \cap \sigma_2^* \neq \emptyset$. To this end, let
$v,u,w \in \VH$ be such that $\{v,u\} \in E_1$ and $\{v,w\} \in
E_2$. If $v \in \IntVert{\sigma_1}\cap\IntVert{\sigma_2}$ then we are done.
Otherwise, suppose without loss of generality that
$v\in\ExtVert{\sigma_1}$ so that $u\in\IntVert{\sigma_1}$. If also
$v\in\ExtVert{\sigma_2}$ then necessarily $w=u$ and $w\in
\IntVert{\sigma_2}$ as $\sigma_1,\sigma_2 \subset \T\setminus\T^\clr$. If
instead $v\in\IntVert{\sigma_2}$ then either $u\in\IntVert{\sigma_2}$
or $\{v,u\}\in\sigma_1^*\cap\sigma_2^*$.

By Fact~\ref{fact:circuits-max}, there exists a circuit $\gamma$ such that $\gamma^*
\subset \sigma_1^* \cup \sigma_2^*$ and $\Int{\sigma_1}
\cup \Int{\sigma_2} \subset \Int{\gamma}$. In particular,
$E_1\cup E_2 \subset E$, where $E := \IntEdge{\gamma} \cup \gamma^*$. It
remains to show that $E$ is a $\clr$-garden of $\omega$. Since, by
Lemma~\ref{lem:circuit-interior-degree}, $\T^\clr \cap
\partial \IntHex\gamma \subset
\partial \IntHex{\sigma_1} \cup \partial \IntHex{\sigma_2}$,
this follows from the assumption that $E_1$ and $E_2$ are
$\clr$-gardens of $\omega$.
\end{proof}

\begin{lemma}
  \label{cl:gardens-different-class}
  Let $\omega$ be a loop configuration, let $E_0$ be a $\clr_0$-garden of $\omega$ and let $E_1$ be a $\clr_1$-garden of $\omega$ with $\clr_0,\clr_1 \in \{0,1,2\}$ distinct. Then, either $E_0\subset E_1$, $E_1 \subset E_0$ or $E_0\cap E_1=\emptyset$.
\end{lemma}
\begin{proof}
  Assume without loss of generality that $\clr_0=0$, $\clr_1=1$ and that $E_0\cap E_1\ne \emptyset$. Denote $\sigma_0 := \sigma(E_0) \subset \T \setminus \T^0$ and $\sigma_1 := \sigma(E_1) \subset \T \setminus \T^1$. Consider an infinite path in $\HH$ beginning with some edge of $E_0 \cap E_1$ and let $e \in \EH$ be the first edge on this path that is not in $\IntEdge{\sigma_0} \cap \IntEdge{\sigma_1}$ (maybe the first edge itself). We may assume without loss of generality that $e \notin \IntEdge{\sigma_0}$. Thus, $e \in \sigma_0^*$, and, therefore, $e$ is bordered by a hexagon $z \in \T^1$ and a hexagon in $\T^2$. Since $e$ is also in $E_1$, $z$ belongs to $\IntHex{\sigma_1}$, by Lemma~\ref{lem:circuit-interior-degree}.
  Now, if $\sigma_0 \subset \IntHex{\sigma_1}$ then $E_0\subset E_1$, by Fact~\ref{fact:gamma-int-ext}.
Otherwise, there exists $\{y,y'\} \in \sigma_0^*$ such that $y\in\IntHex{\sigma_1}$ and $y' \notin \IntHex{\sigma_1}$.
In particular, $y'$ must be in $\sigma_0 \cap \sigma_1 \subset \T^2$, so that $y$ must be in  $\T^1$. Since $y$ is in $\partial \IntHex{\sigma_1}$, it must be a $1$-flower of $\omega$. But since $y$ is on $\sigma_0$, it must also be adjacent to a $0$-flower of $\omega$, which is a contradiction.
\end{proof}

\begin{definition}[$\clr$-cluster, $\clr$-cluster inside $\gamma$]
Let $\clr\in\{0,1,2\}$ and let $\omega$ be a loop configuration. A
subset $E \subset \EH$ is a {\em $\clr$-cluster} of $\omega$ if
it is a $\clr$-garden of $\omega$ and it is not contained in any
other garden of $\omega$. Let
$\gamma$ be a vacant circuit in $\omega$ and note that $\omega \cap
\IntEdge\gamma$ is a loop configuration by
Lemma~\ref{lem:circuits-and-clusters}\ref{it:circuits-and-clusters1}.
A subset $E \subset \EH$ is a {\em $\clr$-cluster} of $\omega$
{\em inside} $\gamma$ if it is a $\clr$-cluster of $\omega \cap
\IntEdge\gamma$.
\end{definition}

\begin{figure}
    \centering
    \includegraphics[scale=0.8]{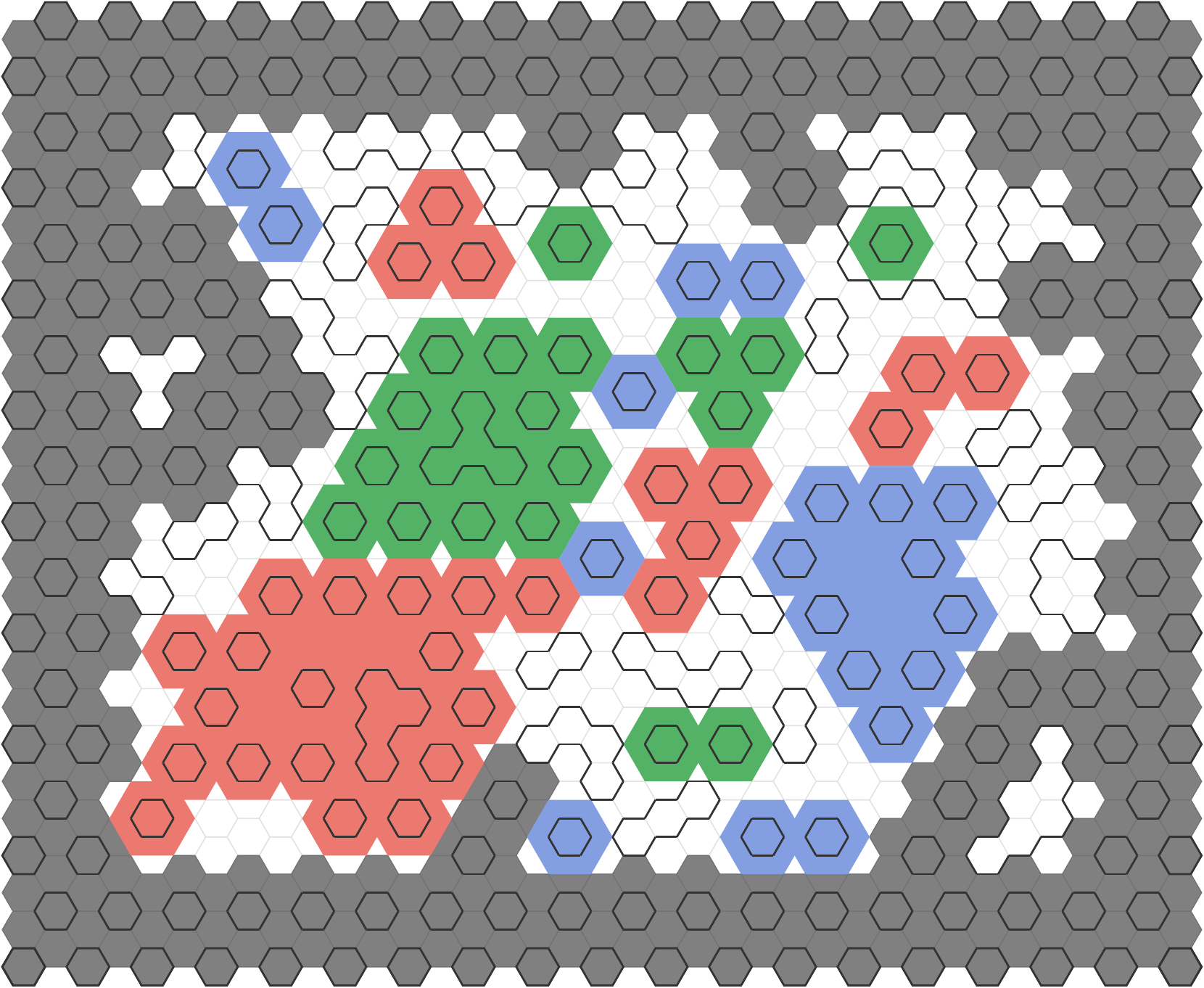}
    \caption{A loop configuration $\omega\in\LC(H,\ground^0)$. The $0$-clusters are denoted in green, the $1$-clusters in red and the $2$-clusters in blue; all taken with respect to the circuit surrounding the large unshaded domain.}
    \label{fig:breakup}
\end{figure}

We say that $E \subset \EH$ is a cluster (inside $\gamma$) if it is a $\clr$-cluster (inside $\gamma$) for some $\clr\in\{0,1,2\}$.
Once again, note that a cluster (inside $\gamma$) is a {\em subset of edges of }$\HH$. Evidently, a cluster of $\omega$ inside $\gamma$ is also a garden of $\omega$, but it is {\em not} necessarily a cluster of $\omega$. The notion of $\clr$-cluster inside $\gamma$ will be important in the definition of the repair map in Section~\ref{sec:repair-map}. Note that, by Lemma~\ref{cl:gardens-same-class} and Lemma~\ref{cl:gardens-different-class},
\begin{equation}\label{eq:clusters-are-disjoint}
\text{any two distinct clusters of $\omega$ (inside $\gamma$) are edge disjoint} ,
\end{equation}
%(even if they have different $\clr$)
and, moreover, for any $\clr\in\{0,1,2\}$,
\begin{equation}\label{eq:clusters-are-disconnected}
\text{the union of any two distinct $\clr$-clusters of $\omega$
(inside $\gamma$) is a disconnected set of edges} ,
\end{equation}
where a set of edges $E$ is said to be connected if the graph whose
vertex set is the set of endpoints of edges in $E$ and whose edge
set is $E$ is connected. Note also, that by Fact~\ref{fact:gamma-int-ext},
\begin{equation}\label{eq:clusters-are-connected}
\text{every cluster of $\omega$ (inside $\gamma$) is a connected set of edges} .
\end{equation}

\subsection{Statement of the main lemma}
\label{sec:main lemma}

We are now in a position to state the main lemma.
For a loop configuration $\omega$ and a vacant circuit $\gamma$ in
$\omega$, denote by $V(\omega,\gamma)$ the set of vertices $v \in
\IntVert\gamma$ such that the three edges of $\HH$ incident to $v$
are not all contained in the same cluster of $\omega$ inside
$\gamma$. One checks simply using Lemma~\ref{lem:circuit-interior-degree} that a vertex $v\in\IntVert\gamma$ satisfies $v\in V(\omega, \gamma)$ if and only if $v$ is incident to an edge which is not in any cluster or each of its incident edges lies in a different cluster.

For a vacant circuit $\gamma\subset \T\setminus\T^0$, the set $V(\omega, \gamma)$ specifies the deviation in $\omega$ from the $0$-phase ground state along the interior boundary of $\gamma$. Our main lemma shows that having a large deviation is exponentially unlikely.

\begin{lemma}
    \label{lem:prob-outer-circuit}
    There exists an absolute constant $c>0$ such that for any $n>0$, any $x \in (0,\infty]$, any domain $H$, any circuit $\gamma\subset \T\setminus\T^0$ and any positive integer $k$, we have
    \[
    \Pr_{H,n,x}^0\big(\partial\IntVert\gamma\subset V(\omega,\gamma)\text{ and }|V(\omega,\gamma)|\ge k \,\mid\, \text{$\gamma$ vacant} \big) \le (cn \cdot \min\{x^6,1\})^{-k/15} .
    \]
\end{lemma}

The reader should first have in mind the simpler case of the lemma
in which $H = \Int\gamma$. In this case the boundary conditions may
equivalently be taken to be vacant. The lemma is stated in greater
generality, allowing, in particular, for $\gamma$ to leave the
domain $H$, i.e., for $\Int\gamma\not\subset H$. This additional
flexibility is used in the proofs of Theorem~\ref{thm:Gibbs
measures_large_x} and
Theorem~\ref{thm:small-deviations-from-the-ground-state-x-model} to
handle the case of domains without a type.

One should note that Lemma~\ref{lem:prob-outer-circuit} contains the implicit assumption that $n \ge n \cdot \min\{x^6,1\} \ge C$, as otherwise its statement is trivial.

\subsection{Definition of the repair map}
\label{sec:repair-map}

For the remainder of this section, we fix a circuit $\gamma\subset \T\setminus \T^0$ and set $H := \Int\gamma$.
 Consider a loop configuration $\omega$ such that $\gamma$ is vacant in $\omega$. The idea of the repair map is to modify $\omega$ as follows:
\begin{itemize}[noitemsep,nolistsep]
\item
  Edges in $1$-clusters inside $\gamma$ are shifted down ``into the $0$-phase''.
\item
  Edges in $2$-clusters inside $\gamma$ are shifted up ``into the $0$-phase''.
\item
  Edges in $0$-clusters inside $\gamma$ are left untouched.
\item
  The remaining edges which are not inside (the shifted) clusters, but are in the interior of $\gamma$ (these edges will be called {\em bad}), are overwritten to ``match'' the $0$-phase ground state, $\ground^0$.
\end{itemize}
See Figure~\reffig{fig:proof-illustration} for an illustration of this
map.

In order to formalize this idea, we need a few definitions. A
\emph{shift} is a graph automorphism of $\T$ which maps every
hexagon to one of its neighbors. We henceforth fix a shift $\DIRup$
which maps $\T^0$ to $\T^1$ (and hence, maps $\T^1$ to $\T^2$ and
$\T^2$ to $\T^0$), and denote its inverse by $\DIRdown$. A shift
naturally induces mappings on the set of vertices and the set of
edges of $\HH$. We shall use the same symbols, $\DIRup$ and
$\DIRdown$, to denote these mappings. Recall from
Section~\ref{sec:results} that $\T$ has a coordinate system given by
$(0,2)\Z + (\sqrt3,1)\Z$ and that $(\T^0, \T^1, \T^2)$ are the color
classes of an arbitrary proper $3$-coloring of $\T$. In our figures
we make the choice that $(0,0) \in \T^0$ and $(0,2) \in \T^1$ so
that $\DIRup$ is the map $(a,b) \mapsto (a,b+2)$.

For a loop configuration $\omega \in \LC(H,\emptyset)$ and $\clr \in
\{0,1,2\}$, let $E^\clr(\omega)\subset \EH$ be the union of all
$\clr$-clusters of $\omega$. Note that, since $H=\Int\gamma$, for
$\omega \in \LC(H,\emptyset)$, the notions of a $\clr$-cluster and a
$\clr$-cluster inside $\gamma$ coincide. For $\omega \in
\LC(H,\emptyset)$, define also
\begin{align}
  \label{eq:def-bad-edges}
  \BadEdges(\omega) &:= (\IntEdge\gamma \cup \gamma^*) \setminus \big( E^0(\omega) \cup E^1(\omega)^{\DIRdown} \cup E^2(\omega)^{\DIRup} \big),
\\
  \label{eq:def-bad-edges-star}
  % \BadEdges_*(\omega,\gamma)
  \BadEdgesBefore(\omega) &:= (\IntEdge\gamma\cup\gamma^*)\setminus \big( E^0(\omega) \cup E^1(\omega) \cup E^2(\omega) \big).
\end{align}
Note that, by~\eqref{eq:clusters-are-disjoint}, $\{ E^0(\omega),
E^1(\omega), E^2(\omega), \BadEdgesBefore(\omega) \}$ is a partition
of $\IntEdge\gamma \cup \gamma^*$. Thus, Lemma~\ref{lem:circuits-and-clusters} implies that
\begin{equation}\label{eq:pairwise-disjoint-loop-configs}
\text{$\omega \cap E^0(\omega)$, $\omega \cap E^1(\omega)$, $\omega \cap E^2(\omega)$ and $\omega \cap \BadEdgesBefore(\omega)$ are pairwise disjoint loop configurations.}
\end{equation}
See Figure~\reffig{fig:breakup} and Figure~\reffig{fig:proof-illustration} for an illustration of these notions. Finally, we define the
\emph{repair map}
\[ \shiftFunc \colon \LC(H,\emptyset) \to \LC(H,\emptyset) \]
by
  \begin{equation*}
    \label{eq:def-repair-map}
      \shift\omega :=  \big(\omega \cap E^0(\omega)\big)  \cup \big(\omega \cap E^1(\omega)\big)^{\DIRdown} \cup \big(\omega \cap E^2(\omega)\big)^{\DIRup}\cup\big(\ground^0 \cap \BadEdges(\omega)\big)  .
  \end{equation*}
The fact that the mapping is well-defined, i.e., that $\shift\omega$
is indeed in $\LC(H,\emptyset)$, is not completely straightforward.
This follows from the following proposition, together with the
simple property in
Lemma~\ref{lem:circuits-and-clusters}\ref{it:circuits-and-clusters4}.

%
% proof illustration
%

\afterpage{
\newgeometry{left=15mm,bottom=40mm,top=10mm}
\pagestyle{empty}
\begin{figure}
    \centering
    \begin{subfigure}[t]{.5\textwidth}
        \includegraphics[scale=0.75]{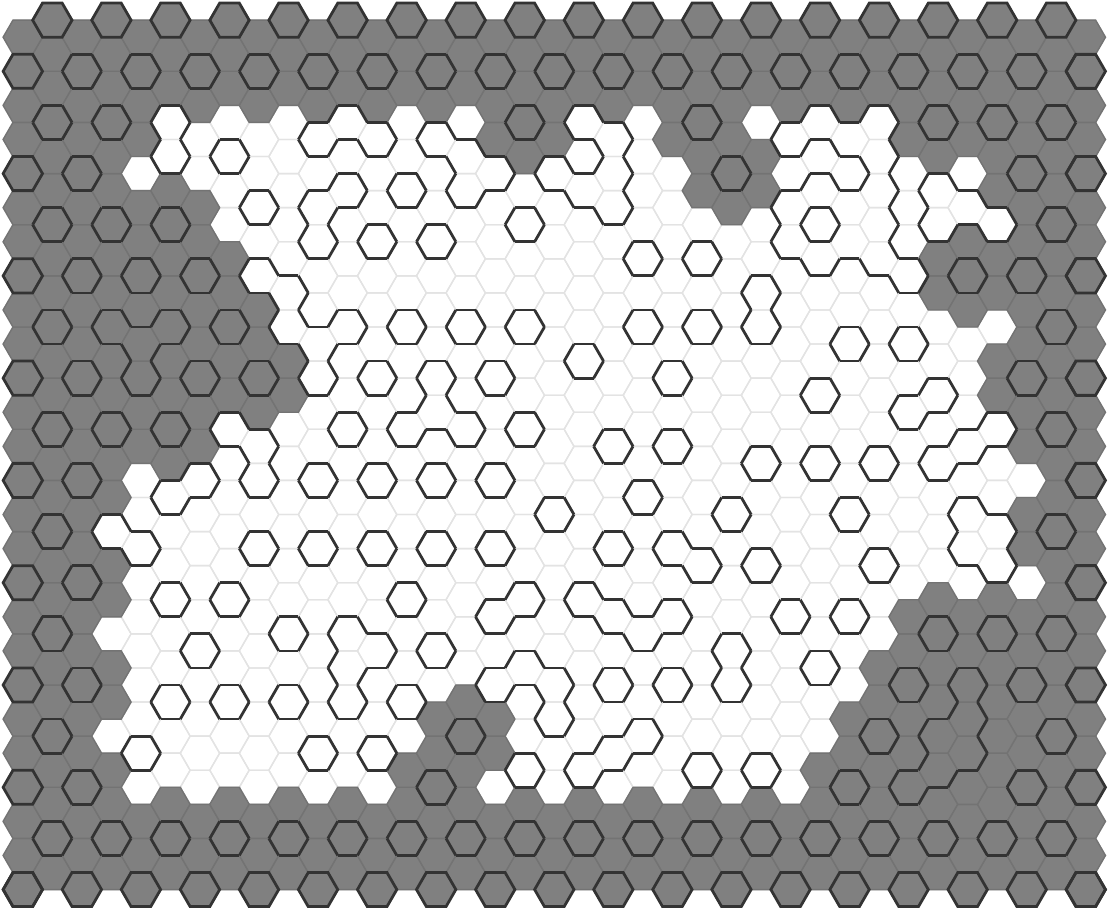}
        \caption{The breakup is found by exploring $0$-flowers from the boundary.}
        \label{fig:proof-illustration-1}
    \end{subfigure}%
    \begin{subfigure}{25pt}
        \quad
    \end{subfigure}%
    \begin{subfigure}[t]{.5\textwidth}
        \includegraphics[scale=0.75]{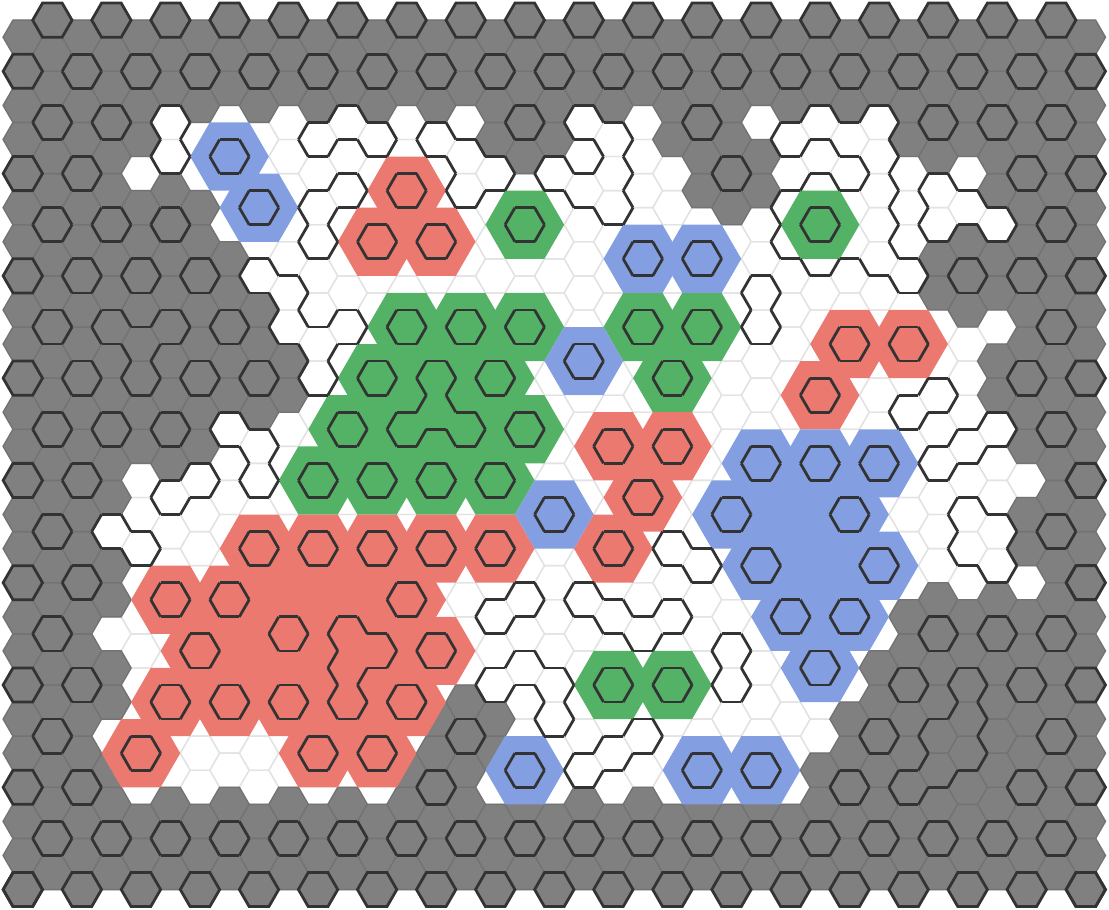}
        \caption{The clusters are found within the breakup.}
        \label{fig:proof-illustration-2}
    \end{subfigure}
    \medbreak
    \begin{subfigure}[t]{.5\textwidth}
        \includegraphics[scale=0.75]{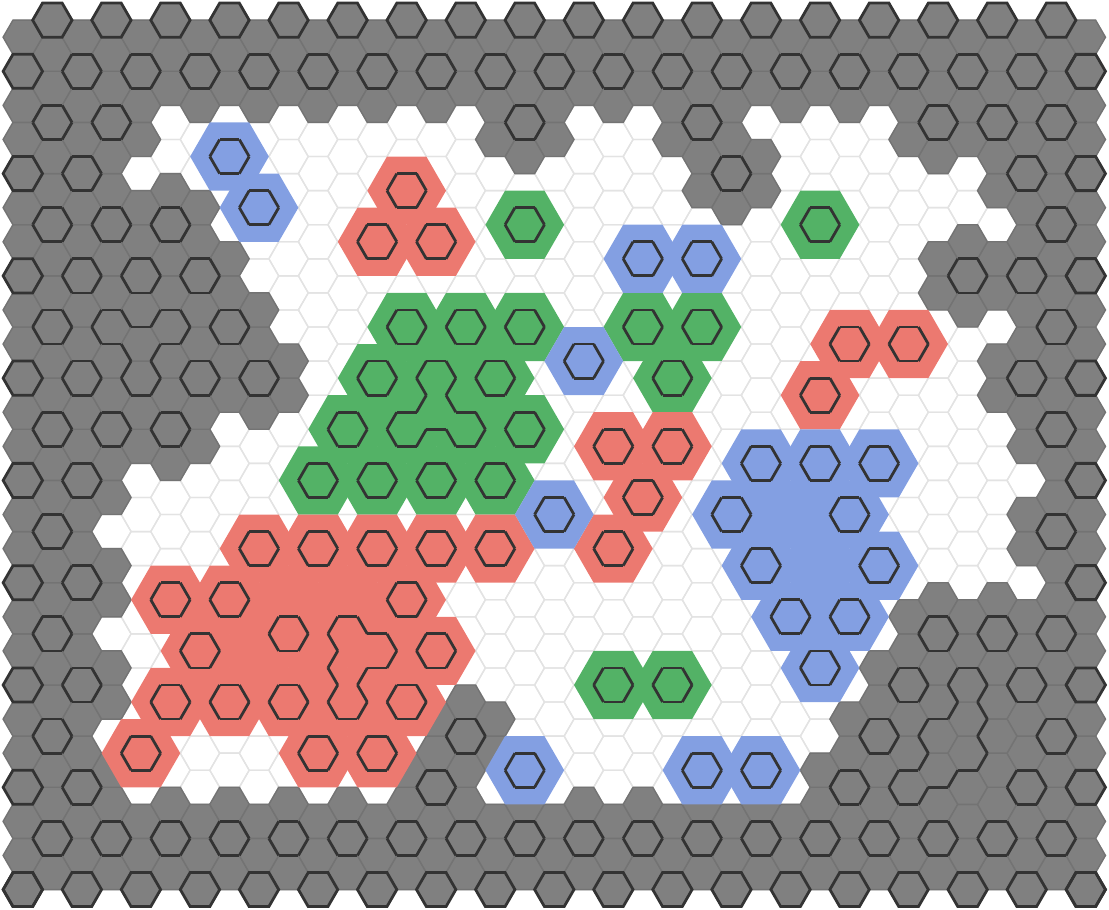}
        \caption{Bad edges are discarded.}
        \label{fig:proof-illustration-3}
    \end{subfigure}%
    \begin{subfigure}{25pt}
        \quad
    \end{subfigure}%
    \begin{subfigure}[t]{.5\textwidth}
        \includegraphics[scale=0.75]{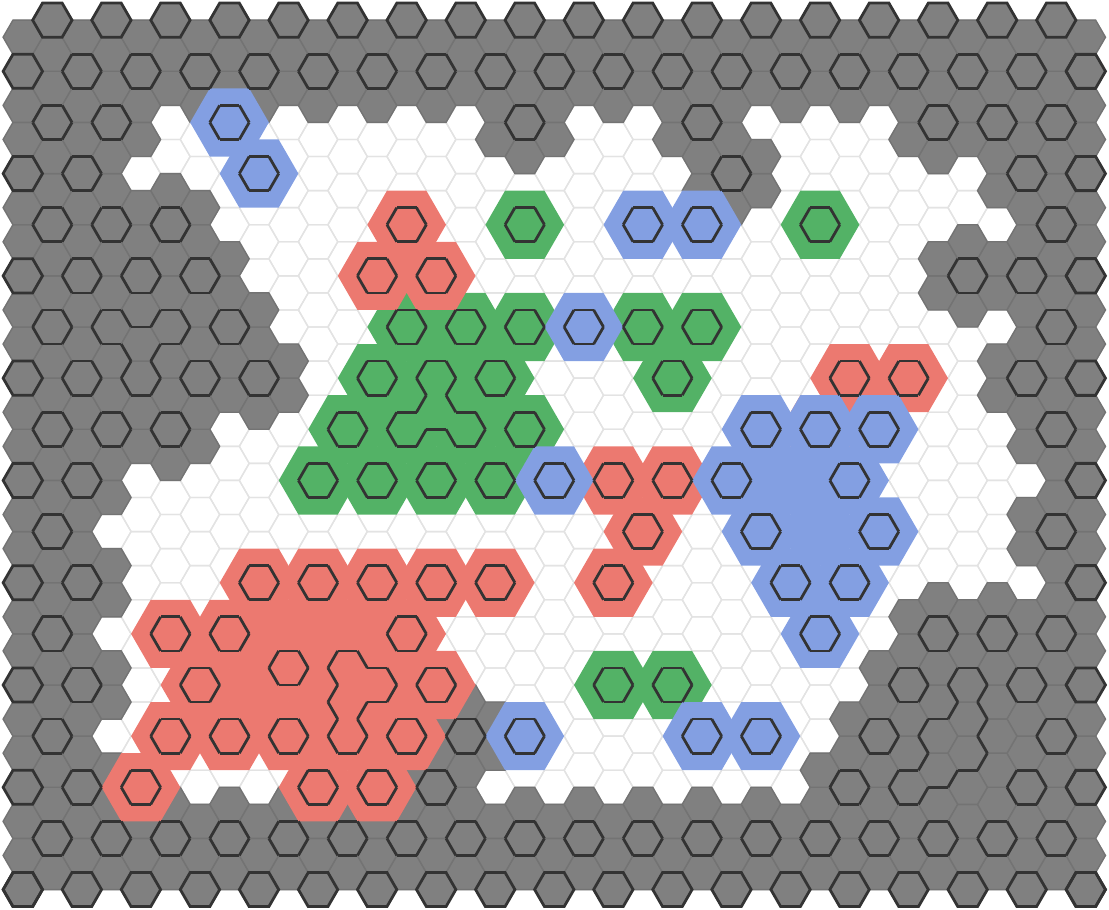}
        \caption{The clusters are shifted into the $0$-phase.
            %Specifically, $1$-clusters are shifted down, $2$-clusters are shifted up and $0$-clusters are not left untouched
        }
        \label{fig:proof-illustration-4}
    \end{subfigure}
    \medbreak
    \begin{subfigure}[t]{.5\textwidth}
        \includegraphics[scale=0.75]{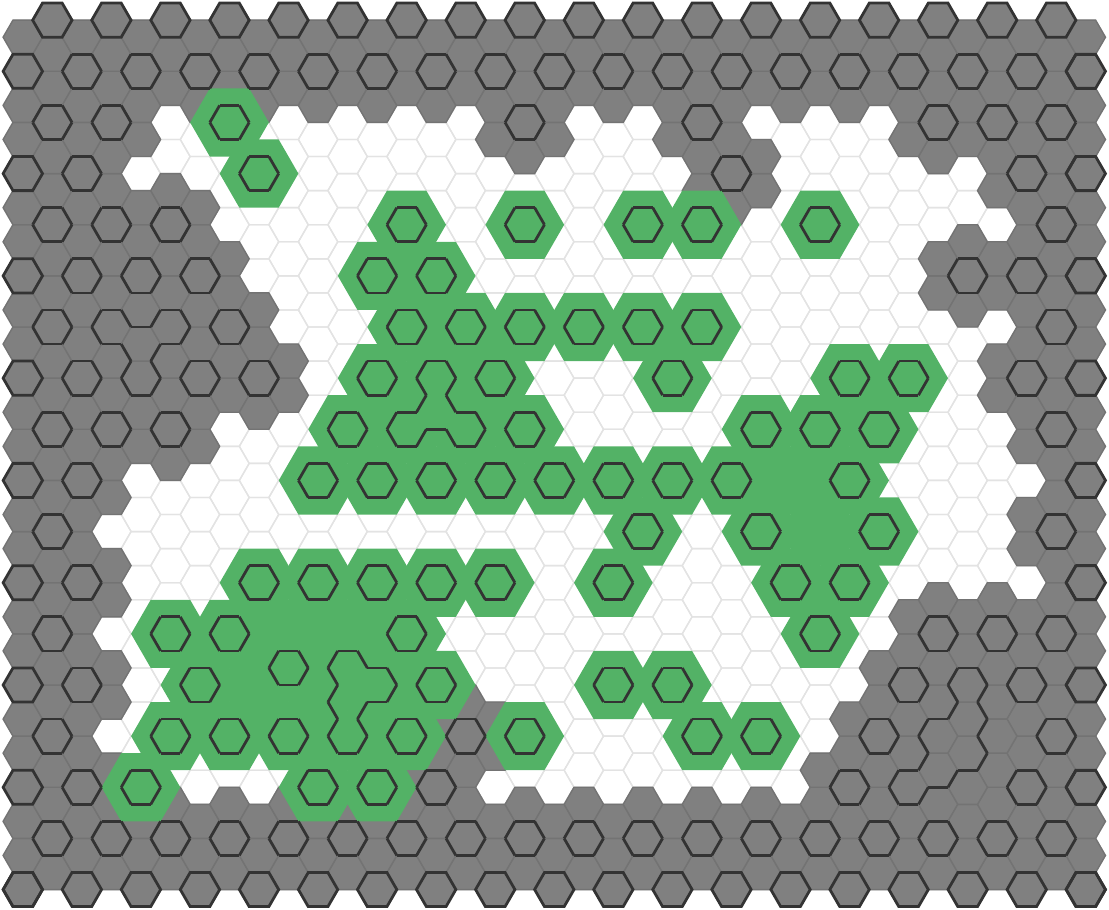}
        \caption{The empty area outside the shifted clusters is now compatible with the $0$-phase ground state.}
        \label{fig:proof-illustration-5}
    \end{subfigure}%
    \begin{subfigure}{25pt}
        \quad
    \end{subfigure}%
    \begin{subfigure}[t]{.5\textwidth}
        \includegraphics[scale=0.75]{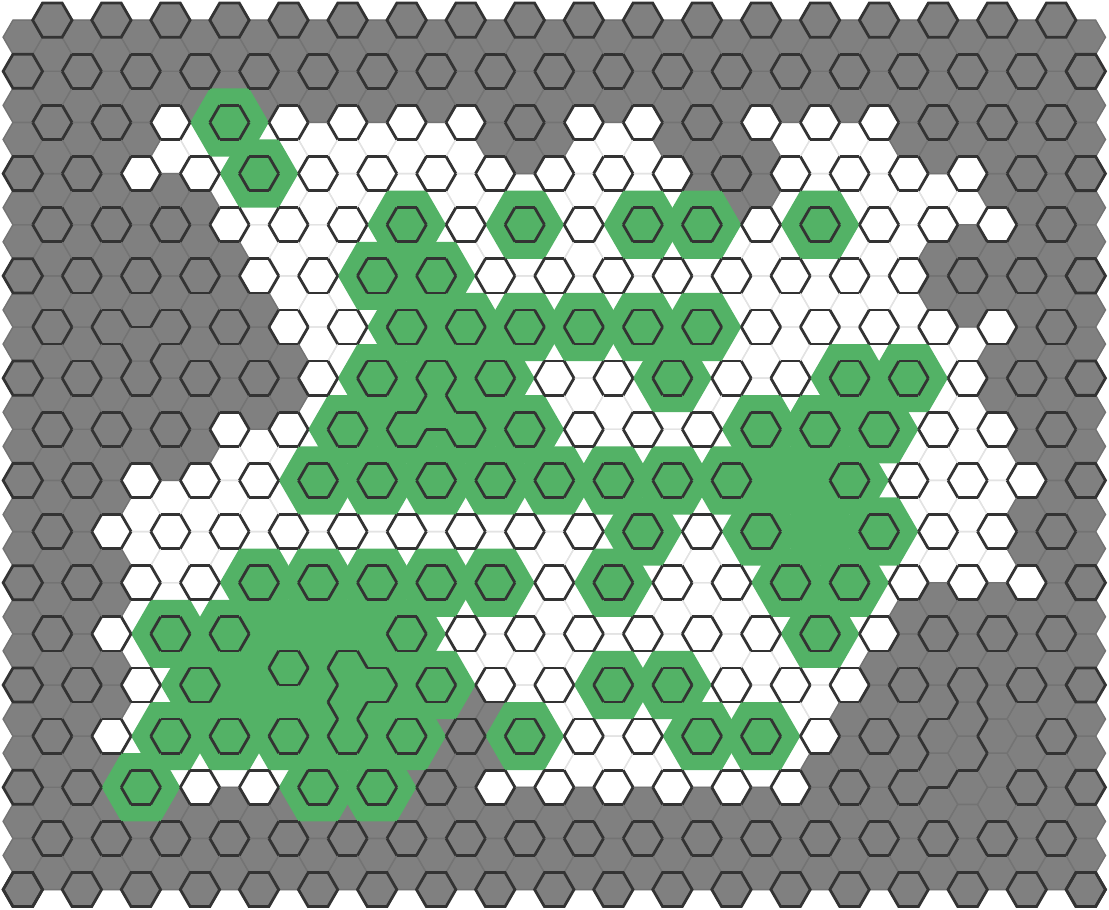}
        \caption{Trivial loops are packed in the empty area outside the shifted clusters.}
        \label{fig:proof-illustration-6}
    \end{subfigure}
    \caption{An illustration of finding the breakup and applying the repair map in it. The initial loop configuration is modified step-by-step, resulting in a loop configuration with many more loops and at least as many edges. Formal definitions are in Section~\ref{sec:repair-map}.}
    \label{fig:proof-illustration}
\end{figure}
\restoregeometry\clearpage }

\begin{prop}\label{prop:repair-map-disjoint-loop-configs}
  Let $\omega\in \LC(H,\emptyset)$. Then $\omega \cap E^0(\omega)$, $(\omega \cap E^1(\omega))^{\DIRdown}\cup(\omega \cap E^2(\omega))^{\DIRup}$ and $\ground^0 \cap \BadEdges(\omega)$ are pairwise disjoint loop configurations in $\LC(H,\emptyset)$.
\end{prop}
We require the following simple geometric lemma.
\begin{lemma}\label{lem:circuits-and-shifts}
    Let $\sigma \subset \T \setminus \T^0$ and $\sigma' \subset \T \setminus \T^1$ be circuits.
    \begin{enumerate}[label=(\alph*), ref=\alph*, labelindent=\parindent]
        \item If $\Int{\sigma'} \subset \Int{\sigma}$ then $\Int{\sigma'}^{\DIRdown} \subset \Int{\sigma}$. \label{it:circuits-and-shifts1}
        \item If $\Int{\sigma'} \subset \Ext{\sigma}$ then $\Int{\sigma'}^{\DIRdown} \subset \Ext{\sigma}$. \label{it:circuits-and-shifts2}
        \item If $\IntVert{\sigma'} \cap \IntVert{\sigma} = \emptyset$ then $\IntVert{\sigma'}^{\DIRdown} \cap \IntVert{\sigma} = \emptyset$. \label{it:circuits-and-shifts3}
    \end{enumerate}
\end{lemma}
\begin{proof}
    We first prove~\eqref{it:circuits-and-shifts1}.
    The assumption that $\Int{\sigma'} \subset \Int{\sigma}$
    implies that $\IntHex{\sigma'} \subset \IntHex{\sigma}$.
    By Lemma~\ref{lem:circuit-interior-degree}, any vertex $v$ in $\IntVert{\sigma'}$ borders a hexagon in $\IntHex{\sigma'}$.
    Thus, it suffices to show that $\IntHex{\sigma'}^{\DIRdown} \subset
    \IntHex{\sigma}$. Assume towards a contradiction that there exists
    a hexagon $z \in \IntHex{\sigma'}$ such that $z^{\DIRdown} \notin
    \IntHex{\sigma}$. In such case, by Fact~\ref{fact:gamma-int-ext}, $z^{\DIRdown}$ must be in $\sigma\cap
    \sigma'\subset\T^2$, and consequently, $z\in\T^0$. Therefore, as $z
    \in \IntHex{\sigma'}$ and $\sigma' \subset
    \T \setminus \T^1$, Lemma~\ref{lem:circuit-interior-degree} implies that the three neighbors of $z$ in $\T^1$ belong to
    $\IntHex{\sigma'}\subset \IntHex{\sigma}$. Now, Lemma~\ref{lem:circuit-interior-degree} implies that $z^{\DIRdown}$ has three neighbors in
    $\T^0\cap \IntHex{\sigma}$. In particular, the six vertices bordering
    $z^{\DIRdown}$ belong to $\IntVert\sigma$, implying that
    $z^{\DIRdown} \in \IntHex{\sigma}$, which is a contradiction.

    The proof of~\eqref{it:circuits-and-shifts2} is very similar to
    that of~\eqref{it:circuits-and-shifts1} and so we omit it.

    Finally, by Fact~\ref{fact:gamma-int-ext}, \eqref{it:circuits-and-shifts3} is equivalent to~\eqref{it:circuits-and-shifts2}.
\end{proof}

\begin{proof}[Proof of Proposition~\ref{prop:repair-map-disjoint-loop-configs}]
  For the sake of brevity, throughout the proof, we drop $\omega$ from the notation of the above sets and write $\BadEdges$, $E^0$, $E^1$ and $E^2$.

  \medbreak\noindent
  \textsc{Step 1: $\omega \cap E^0$, $(\omega \cap E^1)^{\DIRdown}\cup(\omega \cap E^2)^{\DIRup}$ and $\ground^0 \cap \BadEdges$ are contained in $\IntEdge\gamma$.}
  \smallbreak

  Since $\gamma$ is vacant in both $\omega$ and $\ground^0$, it follows that $\omega \cap E^0$ and $\ground^0 \cap \BadEdges$ are contained in $\IntEdge\gamma$.
  It remains to show that $(\omega \cap E^1)^{\DIRdown}$ and $(\omega \cap E^2)^{\DIRup}$ are contained in $\IntEdge\gamma$.
  %For $(\omega\cap E^1)^{\DIRdown}$, we argue as follows. Let $E$ be a $1$-cluster of $\omega$ and denote $\sigma:=\sigma(E) \subset \T\setminus\T^1$. We wish to prove that $\IntHex\sigma^{\DIRdown} \subset \IntHex\gamma$. Assume towards a contradiction that there exists a hexagon $z \in \IntHex\sigma$ such that $z^{\DIRdown} \notin \IntHex\gamma$. In such case $z^{\DIRdown}$ must be in $ \gamma\cap \sigma\subset\T^2$. As a consequence, $z\in\T^0$. Therefore, as $z \in \IntHex\sigma$, the three neighbors of $z$ in $\T^1$ belong to $\IntHex\sigma$ (they cannot belong to $\sigma$ as $\sigma \subset \T \setminus \T^1$), and hence also to $\IntHex\gamma\supset\IntHex\sigma$. This implies that $z^{\DIRdown}$ has three neighbors in $\T^0\cap(\IntHex\gamma\cup\gamma)$. It must therefore be in $\IntHex\gamma$, which is a contradiction. This proves that $(\omega\cap E^1)^{\DIRdown} \subset \widehat E$. A similar argument shows that $(\omega\cap E^2)^{\DIRup} \subset \widehat E$.
  We show this only for $(\omega \cap E^1)^{\DIRdown}$, as the other case is symmetric.
  Let $E$ be a $1$-cluster of $\omega$.
  We must show that $(\omega \cap E)^{\DIRdown} \subset \IntEdge\gamma$.
  Since, by Lemma~\ref{lem:circuits-and-clusters}\ref{it:circuits-and-clusters2}, $\omega \cap E \subset
  \IntEdge{\sigma(E)} \subset \IntEdge\gamma$, this follows from
  Lemma~\ref{lem:circuits-and-shifts}\ref{it:circuits-and-shifts1}.

  \medbreak\noindent
  \textsc{Step 2: $\omega \cap E^0$, $(\omega \cap E^1)^{\DIRdown}\cup(\omega \cap E^2)^{\DIRup}$ and $\ground^0 \cap \BadEdges$ are pairwise disjoint.}
  \smallbreak

  By definition, $\BadEdges$ (and therefore $\ground^0 \cap \BadEdges$) is disjoint from the first two sets.
  It remains to show that $\omega\cap E^0$ is disjoint from $(\omega \cap E^1)^{\DIRdown}$ and $(\omega \cap E^2)^{\DIRup}$.
  We show this only for $\omega\cap E^0$ and $(\omega \cap E^1)^{\DIRdown}$, as the other case is symmetric.
  %Denote $\sigma := \sigma(E^0)$ and $\sigma' := \sigma(E^1)$.
  Let $E$ and $E'$ be $0$- and $1$-clusters of $\omega$,
  respectively. We must show that $(\omega \cap E) \cap
  (\omega \cap E')^{\DIRdown} = \emptyset$.
  By Lemma~\ref{lem:circuits-and-clusters}\ref{it:circuits-and-clusters2},
  $(\omega \cap E) \cap
  (\omega \cap E')^{\DIRdown} \subset \IntEdge{\sigma(E)} \cap
  \IntEdge{\sigma(E')}^{\DIRdown}$, which is empty by~\eqref{eq:clusters-are-disjoint} and
  Lemma~\ref{lem:circuits-and-shifts}\ref{it:circuits-and-shifts3}.

  \medbreak\noindent
  \textsc{Step 3: $\omega \cap E^0$, $(\omega \cap E^1)^{\DIRdown}\cup(\omega \cap E^2)^{\DIRup}$ and $\ground^0 \cap \BadEdges$ are loop configurations.}
  \smallbreak

    We first show that $\ground^0\cap\BadEdges$ is a loop configuration.
    Observe that $E^0 \cup (E^1)^{\DIRdown} \cup (E^2)^{\DIRup}$ is the union of $\IntEdge{\sigma} \cup \sigma^*$
    for a collection of circuits $\sigma \subset \T \setminus \T^0$.
    Since every circuit $\sigma \subset \T \setminus \T^0$ is vacant in $\ground^0$, Lemma~\ref{lem:circuits-and-clusters} implies that $\ground^0 \cap (E^0 \cup (E^1)^{\DIRdown} \cup (E^2)^{\DIRup})$ is a loop configuration, and thus, also that $\ground^0\cap\BadEdges=(\ground^0 \setminus (E^0 \cup (E^1)^{\DIRdown} \cup (E^2)^{\DIRup})) \cap \IntEdge\gamma$ is a loop configuration.

    Since $\omega \cap E^0$ is a loop configuration, by~\eqref{eq:pairwise-disjoint-loop-configs}, it remains only to check that $(\omega \cap E^1)^{\DIRdown}\cup(\omega \cap E^2)^{\DIRup}$ is a loop configuration.
    In light of~\eqref{eq:pairwise-disjoint-loop-configs} and Lemma~\ref{lem:circuits-and-clusters}(\ref{it:circuits-and-clusters4},\ref{it:circuits-and-clusters5}), it suffices to show that $(\omega \cap E^1)^{\DIRdown}\cap(\omega \cap E^2)^{\DIRup}$ is a loop configuration. For convenience, we prove this separately in the next lemma.
\end{proof}

 For a hexagon $z \in \T$, we denote by $E(z)$ the six edges bordering $z$.
 We call a hexagon $z \in \T$ \emph{double-clustered} for $\omega$
 if $E(z^{\DIRup}) \subset E^1(\omega)$ and $E(z^{\DIRdown}) \subset E^2(\omega)$.
 Denote by $\dbl(\omega)$ the subset of all hexagons in $\IntHex\gamma$
 that are double-clustered for $\omega$.

\begin{lemma}\label{lem:E12-intersection}
    Let $\omega\in \LC(H,\emptyset)$. Then $\dbl(\omega) \subset \T^0$ and $(\omega \cap E^1(\omega))^{\DIRdown}\cap(\omega \cap E^2(\omega))^{\DIRup}$ consists solely of the trivial loops surrounding the hexagons in $\dbl(\omega)$. That is,
    \[ (\omega \cap E^1(\omega))^{\DIRdown}\cap(\omega \cap E^2(\omega))^{\DIRup} = \bigcup_{z \in \dbl(\omega)} E(z) .\]
\end{lemma}
\begin{proof}
    Let $z \in \dbl(\omega)$. Then $z^{\DIRup} \in \IntHex{\sigma(E_1)}$ and $z^{\DIRdown} \in \IntHex{\sigma(E_2)}$, where $E_1$ and $E_2$ are $1$- and $2$-clusters of $\omega$, respectively. It follows from Lemma~\ref{lem:circuit-interior-degree} and~\eqref{eq:clusters-are-disjoint} that $z \in \T^0$ and that $z \notin \IntHex{\sigma(E_1)} \cup \IntHex{\sigma(E_2)}$. Thus, $z^{\DIRup}$ is a $1$-flower of $\omega$ and $z^{\DIRdown}$ is a $2$-flower of $\omega$. In particular, $E(z) \subset (\omega \cap E^1(\omega))^{\DIRdown}\cap(\omega \cap E^2(\omega))^{\DIRup}$.

    For the opposite containment, let $e \in (\omega \cap E^1(\omega))^{\DIRdown}\cap(\omega \cap E^2(\omega))^{\DIRup}$. Then $e^{\DIRup} \in \IntEdge{\sigma(E_1)} \cup \sigma(E_1)^*$ and $e^{\DIRdown} \in \IntEdge{\sigma(E_2)} \cup \sigma(E_2)^*$, where $E_1$ and $E_2$ are $1$- and $2$-clusters of $\omega$, respectively. Since, by Lemma~\ref{lem:circuits-and-clusters}\ref{it:circuits-and-clusters2}, $\sigma(E_1)$ and $\sigma(E_2)$ are vacant in $\omega$, we have $e^{\DIRup} \in \IntEdge{\sigma(E_1)}$ and $e^{\DIRdown} \in \IntEdge{\sigma(E_2)}$. In particular, both endpoints of $e^{\DIRup}$ belong to $\IntVert{\sigma(E_1)}$ and both endpoints of $e^{\DIRdown}$ belong to $\IntVert{\sigma(E_2)}$. Therefore, by Lemma~\ref{lem:circuit-interior-degree}, $e$ must border a hexagon $z$ in $\T^0$, and $E(z^{\DIRup}) \subset E_1$ and $E(z^{\DIRdown}) \subset E_2$. Thus, $z \in \dbl(\omega)$.
\end{proof}

The next lemma shows that certain boundary conditions are preserved by the repair map.

\begin{lemma}\label{lem:range-of-repair-map-new}
    Let $H'$ be a domain and denote $\cE := \LC(H,\emptyset) \cap \LC(H',\ground^0 \cap \IntEdge\gamma)$.
    Then $\shift{\cE} \subset \cE$.
\end{lemma}
\begin{proof}
    Let $\omega \in \cE$ and denote $\omega' := \shift\omega$.
    Set $F := \IntEdge{\gamma} \setminus E(H')$ and note that $\cE = \{ \tilde\omega \in \LC(H,\emptyset) : \tilde\omega \cap F = \ground^0 \cap F \}$. In fact, one easily checks that $\cE = \{ \tilde\omega \in \LC(H,\emptyset) : \ground^0 \cap F \subset \tilde\omega \}$. Thus, by Proposition~\ref{prop:repair-map-disjoint-loop-configs}, it suffices to show that $\ground^0 \cap F \subset \omega'$.

    Let us first show that $F$ is disjoint from $E^1(\omega)$ and $E^2(\omega)$.
    To this end, let $e \in F$ and consider an infinite simple path in $E(H')^c$ starting from $e$. Observe that no vertex on this path borders a 1- or 2-flower of $\omega$.
    On the other hand, by the definition of a cluster, if $e$ belongs to a 1- or 2-cluster of $\omega$, then any such path must have such a vertex. Hence, $e \notin E^1(\omega) \cup E^2(\omega)$.

    Towards showing that $\ground^0 \cap F \subset \omega'$, let $e \in \ground^0 \cap F$ and note that $e$ borders a hexagon $z \in \T^0$. By Lemma~\ref{lem:circuit-interior-degree}, $E(z)$ is contained in either $E^0(\omega)$, $E^1(\omega)^{\DIRdown}$, $E^2(\omega)^{\DIRup}$ or $\BadEdges(\omega)$.
    In the first case, $E(z) \subset \omega \cap E^0(\omega) \subset \omega'$.
    In the second case, $z^{\DIRup} \in \IntHex{\sigma(E)}$ for some $1$-cluster $E$ of $\omega$.
Since $e \notin E^1(\omega)$, we have $z^{\DIRup} \in \partial\IntHex{\sigma(E)}$. Thus, $z^{\DIRup}$ is a $1$-flower of $\omega$ and $E(z) \subset (\omega \cap E^1(\omega))^{\DIRdown} \subset \omega'$.
The third case is similar to the second case.
Finally, in the last case, $E(z) \subset \ground^0 \cap \BadEdges(\omega) \subset \omega'$.
\end{proof}

\subsection{Comparing the probabilities of $\shift\omega$ and $\omega$}
\label{sec:compare probabilities}

As in Section~\ref{sec:repair-map}, we henceforth fix a circuit $\gamma\subset \T\setminus \T^0$ and denote $H:=\Int\gamma$.
Our goal now is to compare the probabilities of $\shift\omega$ and $\omega$.
Recall the definition of $V(\omega,\gamma)$ from Section~\ref{sec:main lemma}.
Denote by $V'(\omega,\gamma)$ the vertices in $V(\omega,\gamma)$ which are isolated in $\omega$ (i.e., which are incident to no edges in $\omega$).

\begin{prop}\label{cor:shift-increases-weight-of-configuration}
  Let $n \ge 1$, let $x \in (0,\infty]$ and let $\omega \in \LC(H,\emptyset)$. Then
  \[
  \Pr_{H,n,x}^\emptyset(\shift\omega) \ge n^{\frac{|V(\omega,\gamma)|}{15} + \frac{|V'(\omega,\gamma)|}{10}} \cdot x^{|V'(\omega,\gamma)|} \cdot \Pr_{H,n,x}^\emptyset(\omega).
  \]
  In particular, if $nx^6 \ge 1$ then
    \[
    \Pr_{H,n,x}^\emptyset(\shift\omega) \ge (n \cdot \min\{x^6,1\})^{\frac{|V(\omega,\gamma)|}{15}} \cdot (\max\{x,1\})^{|V'(\omega,\gamma)|} \cdot \Pr_{H,n,x}^\emptyset(\omega).
    \]
\end{prop}

The proof of Proposition~\ref{cor:shift-increases-weight-of-configuration} is based on showing that applying the repair map can only increase the number of loops and edges and estimating carefully the amounts by which they increase.

We begin with two preliminary lemmas. Denote by $\BadVert(\omega)$ the subset of $\IntVert\gamma$ composed of endpoints of edges in $\BadEdges(\omega)$.
Recall the definition of $\dbl(\omega)$ just prior to Lemma~\ref{lem:E12-intersection}.

\begin{lemma}
  \label{cl:bad-edges-formula}
  For any $\omega\in\LC(H,\emptyset)$, we have
  \[
  |\BadVert(\omega)| = |V(\omega,\gamma)| + 6 \cdot |\dbl(\omega)|.
  \]
\end{lemma}
\begin{proof}
  As before, set $E^\clr:=E^\clr(\omega)$ for $\clr\in\{0,1,2\}$. Let $U:=\IntVert\gamma\setminus V(\omega,\gamma)$ be the set of vertices whose three incident edges are contained in one of the sets $E^0$, $E^1$ or $E^2$. Let $U':=\IntVert\gamma\setminus \BadVert(\omega)$ be the set of vertices whose three incident edges are contained in one of the sets $E^0$, $(E^1)^{\DIRdown}$ or $(E^2)^{\DIRup}$. The lemma will follow if we show that $|U| - |U'| = 6 \cdot |\dbl(\omega)|$.

  For $E \subset \EH$, denote by $\Int E$ the set of vertices whose 3 incident edges belong to $E$. Then
  \begin{equation}\label{eq:bad-edges-formula-1}
  \begin{aligned}
    U &= \Int{E^0}\cup\Int{E^1}\cup\Int{E^2}, \\
    U' &= \Int{E^0} \cup \Int{E^1}^{\DIRdown} \cup \Int{E^2}^{\DIRup}.
  \end{aligned}
  \end{equation}
  We now show that
  \begin{equation}\label{eq:bad-edges-formula-2}
    \Int{E^0} \cap \Int{E^1}^{\DIRdown} = \emptyset \qquad\text{and}\qquad \Int{E^0} \cap \Int{E^2}^{\DIRup} = \emptyset .
  \end{equation}
  Note that, for a garden $E$, we have $\Int E=\IntVert{\sigma(E)}$. Thus, it follows from~\eqref{eq:clusters-are-disjoint} and Lemma~\ref{lem:circuits-and-shifts}\ref{it:circuits-and-shifts3} that if $E$ and $E'$ are $0$- and $1$-clusters of $\omega$, respectively, then $\Int{E} \cap \Int{E'}^{\DIRdown} = \emptyset$.
  On the other hand, $\Int{E^\clr} = \cup\,\Int{E}$ over all $\clr$-clusters $E$ of $\omega$ in $\gamma$, as follows from~\eqref{eq:clusters-are-disconnected} and~\eqref{eq:clusters-are-connected}.
  We therefore conclude that $\Int{E^0} \cap \Int{E^1}^{\DIRdown}=\emptyset$. By symmetry, we also have $\Int{E^0} \cap \Int{E^2}^{\DIRup}=\emptyset$.

  Using the inclusion-exclusion principle, we obtain
  \begin{align*}
    |U'|
    &= |\Int{E^0}| + |\Int{E^1}^{\DIRdown}| + |\Int{E^2}^{\DIRup}| - |\Int{E^1}^{\DIRdown} \cap \Int{E^2}^{\DIRup}|
     &&\text{\small by~\eqref{eq:bad-edges-formula-1} and~\eqref{eq:bad-edges-formula-2}} \\
    &= |\Int{E^0}| + |\Int{E^1}| + |\Int{E^2}| - |\Int{E^1}^{\DIRdown} \cap \Int{E^2}^{\DIRup}| \\
    &= |U| - |\Int{E^1}^{\DIRdown} \cap \Int{E^2}^{\DIRup}|.
     &&\text{\small by~\eqref{eq:bad-edges-formula-1} and~\eqref{eq:clusters-are-disjoint}}
  \end{align*}
  Finally, observe that, by Lemma~\ref{lem:circuit-interior-degree}, $\Int{E^1}^{\DIRdown} \cap \Int{E^2}^{\DIRup}$ is precisely the set of vertices that border the hexagons in $\dbl(\omega)$ and that each such vertex is incident to a unique double-clustered hexagon (since $\dbl(\omega) \subset \T^0$, by Lemma~\ref{lem:E12-intersection}). Consequently,
  \[
  |\Int{E^1}^{\DIRdown} \cap \Int{E^2}^{\DIRup}| = 6 \cdot |\dbl(\omega)| . \qedhere
  \]
\end{proof}

For our next lemma, we require the following definition. A
\emph{functional on loops} is a map $\phi$ that assigns a real
number to each loop in $\HH$. We say that $\phi$ is
\emph{$\DIRup$-invariant} if $\phi(L^{\DIRup})=\phi(L)$ for every
loop $L$ and $\phi(L)=\phi(L')$ for any two trivial loops $L$ and
$L'$. Given such a functional, we extend $\phi$ to finite loop
configurations $\omega$ by summing over all the loops, i.e., by
setting
\[
\phi(\omega) := \sum_{\text{loops $L$ in $\omega$}} \phi(L).
\]

Recall the definition of $\BadEdgesBefore(\omega)$
from~\eqref{eq:def-bad-edges-star} and the repair map from Section~\ref{sec:repair-map}.
Let $\TrivialLoop\subset\HH$ denote a trivial loop.

\begin{lemma}\label{cl:shift-general-measure-formula}
  For any $\omega\in\LC(H,\emptyset)$ and any $\DIRup$-invariant functional $\phi$ on loops, we have
  \[
  \phi(\shift\omega) - \phi(\omega) = \phi(\TrivialLoop) \cdot \tfrac{|V(\omega,\gamma)|}{6} - \phi(\omega \cap \BadEdgesBefore(\omega)).
  \]
\end{lemma}

\begin{proof}
  As before, set $E^\clr:=E^\clr(\omega)$ for $\clr\in\{0,1,2\}$ and $\BadEdges:=\BadEdges(\omega)$. Recall from Proposition~\ref{prop:repair-map-disjoint-loop-configs} that each loop of $\shift\omega$ belongs to one of the following pairwise disjoint loop configurations: $\omega \cap E^0$, $\ground^0 \cap \BadEdges$, or $(\omega \cap E^1)^{\DIRdown}\cup(\omega \cap E^2)^{\DIRup}$. Thus, the definition of a functional implies that \begin{equation}\label{eq:shift-general-measure-formula-1}
    \phi(\shift\omega) = \phi(\omega \cap E^0) + \phi(\ground^0 \cap \BadEdges)+ \phi\big((\omega \cap E^1)^{\DIRdown} \cup (\omega \cap E^2)^{\DIRup}\big) .
  \end{equation}

  We claim that $\ground^0 \cap \BadEdges$ consists of $|\BadVert(\omega)|/6$ trivial loops.
  As $\ground^0 \cap \BadEdges$ is a loop configuration and $\ground^0$ is a fully-packed loop configuration (i.e., every vertex has degree $2$) containing only trivial loops, it suffices to show that each vertex in $\BadVert(\omega)$ is incident to at least two edges in $\BadEdges$.
  We may write
  \[ \BadEdges = (\IntEdge\gamma \cup \gamma^*) \setminus \bigcup_i (\IntEdge{\sigma_i} \cup \sigma_i^*) = \bigcap_i \ExtEdge{\sigma_i} \setminus \ExtEdge\gamma \]
  for some circuits $\sigma_i \subset \T \setminus \T^0$.
  Let $v \in \BadVert(\omega)$ and let $z$ be the hexagon in $\T^0$ which $v$ borders.
  By Lemma~\ref{lem:circuit-interior-degree}, the six edges bordering $z$ must belong to $\IntEdge\gamma$ and to $\ExtEdge{\sigma_i}$ for each $i$. Hence, they belong to $\BadEdges$, and, in particular, two edges incident to $v$ belong to $\BadEdges$, as required.

  Thus, the $\DIRup$-invariance of $\phi$ implies
  \begin{equation}
    \label{eq:shift-general-measure-formula-3}
    \phi(\ground^0 \cap \BadEdges) = \phi(\TrivialLoop) \cdot |\BadVert(\omega)| / 6 .
  \end{equation}
  By Lemma~\ref{lem:E12-intersection}, the inclusion-exclusion principle and the $\DIRup$-invariance of $\phi$, we have that
  \begin{equation}
    \label{eq:shift-general-measure-formula-2}
    \begin{split}
      \phi\big((\omega \cap E^1)^{\DIRdown} \cup (\omega \cap E^2)^{\DIRup}\big)
      &= \phi((\omega \cap E^1)^{\DIRdown})
      + \phi((\omega \cap E^2)^{\DIRup}) - \phi\big((\omega \cap E^1)^{\DIRdown} \cap (\omega \cap E^2)^{\DIRup}\big) \\
      &= \phi(\omega \cap E^1)
      + \phi(\omega \cap E^2)- \phi(\TrivialLoop) \cdot |\dbl(\omega)|.
    \end{split}
  \end{equation}
  Using identities \eqref{eq:shift-general-measure-formula-1}, \eqref{eq:shift-general-measure-formula-3}, \eqref{eq:shift-general-measure-formula-2} and Lemma~\ref{cl:bad-edges-formula}, we obtain
  \[
  \phi(\shift\omega)
  = \phi(\omega \cap E^0)
  + \phi(\omega \cap E^1)
  + \phi(\omega \cap E^2) + \phi(\TrivialLoop) \cdot |V(\omega,\gamma)| / 6 .
  \]
  Finally, by~\eqref{eq:pairwise-disjoint-loop-configs},
  \[
  \phi(\omega)
  = \phi(\omega \cap E^0)
  + \phi(\omega \cap E^1)
  + \phi(\omega \cap E^2)
  + \phi(\omega \cap \BadEdgesBefore(\omega)),
  \]
  and the lemma follows by subtracting the last two displayed equations.
\end{proof}

\begin{proof}[Proof of Proposition~\ref{cor:shift-increases-weight-of-configuration}]
  Fix a loop configuration $\omega \in \LC(H,\emptyset)$.
  Lemma~\ref{cl:shift-general-measure-formula} applied to the $\DIRup$-invariant functionals $\phi_1$ and $\phi_2$ defined by
  \[
  \phi_1(L):=|E(L)| \quad\text{and}\quad \phi_2(L):=1 \quad \text{for every loop $L$}
  \]
  implies (respectively) that
  \begin{align}
  \Delta o &:= o_H(\shift\omega) - o_H(\omega) = |V(\omega,\gamma)| - |\omega \cap \BadEdgesBefore(\omega)|, \label{eq:boundO}\\
  \Delta L &:= L_H(\shift\omega) - L_H(\omega) = |V(\omega,\gamma)| / 6 - L_H(\omega \cap \BadEdgesBefore(\omega)) \label{eq:boundL}.
  \end{align}
    Since every trivial loop of $\omega$ is contained in a cluster, there are no trivial loops of $\omega$ in $\BadEdgesBefore(\omega)$. Hence, as any non-trivial loop contains at least $10$ edges,
    \[
    L_H(\omega \cap \BadEdgesBefore(\omega)) \le |\omega \cap \BadEdgesBefore(\omega)| / 10 .
    \]
  Furthermore, the simple observation that $V(\omega,\gamma) \setminus V'(\omega,\gamma)$ is precisely the set of endpoints of edges in $\omega \cap \BadEdgesBefore(\omega)$, and the fact that $\omega \cap \BadEdgesBefore(\omega)$ is a loop configuration, by~\eqref{eq:pairwise-disjoint-loop-configs}, imply that
\[ |\omega \cap \BadEdgesBefore(\omega)| = |V(\omega,\gamma) \setminus V'(\omega,\gamma)|.\]
  Substituting these in~\eqref{eq:boundO} and \eqref{eq:boundL}, we obtain
  \[ \Delta o = |V'(\omega,\gamma)| \qquad\text{and}\qquad \Delta L \ge \tfrac{|V(\omega,\gamma)|}{15} + \tfrac{|V'(\omega,\gamma)|}{10} .\]
  Therefore, as $n \ge 1$ by assumption,
  \[
  \frac{\Pr_{H,n,x}^\emptyset(\shift\omega)}{\Pr_{H,n,x}^\emptyset(\omega)} = \frac{x^{o_H(\shift\omega)} \cdot n^{L_H(\shift\omega)}}{x^{o_H(\omega)} \cdot n^{L_H(\omega)}} = x^{\Delta o} \cdot n^{\Delta L} \ge x^{|V'(\omega,\gamma)|} \cdot n^{\frac{|V(\omega,\gamma)|}{15} + \frac{|V'(\omega,\gamma)|}{10}} . \qedhere
  \]
\end{proof}

\subsection{Proof of the main lemma}
\label{sec:proof of lemma}

In this section, we prove Lemma~\ref{lem:prob-outer-circuit}. Recall the definition of $V(\omega,\gamma)$ from Section~\ref{sec:main lemma}. Let us start with two technical lemmas regarding the connectedness of $V(\omega,\gamma)$. Let $\HH^{\times}$ be the graph obtained from $\HH$ by adding an edge between each pair of opposite vertices of every hexagon, so that $\HH^{\times}$ is a $6$-regular non-planar graph.

\begin{figure}
    \centering
    \begin{tikzpicture}[scale=0.5, every node/.style={scale=0.5}]
    \hexagon[0][0];\hexagon[1][0];\hexagon[2][-1];
    \node [circle,draw,fill=white, inner sep=1mm] at (0.5,\h) { };
    \node [circle,draw,fill=white, inner sep=1mm] at (2.5,\h) { };
    \node [circle,draw,fill, inner sep=1mm] at (1,0) { };
    \node [circle,draw,fill, inner sep=1mm] at (2,0) { };
    \node [scale=1.6] at (0.2,1.4*\h) {$u'$};
    \node [scale=1.6] at (2.9,1.4*\h) {$v'$};
    \node [scale=1.6] at (0.6,-0.2*\h) {$u$};
    \node [scale=1.6] at (2.4,-0.2*\h) {$v$};
    \node [scale=1.6] at (1.5,1.3*\h) {$z$};
    \begin{scope}[xscale=1.5, yscale=0.866]
    \draw [domain-path,-] {(0,0)--(1,1)--(2,0) };
    \end{scope}
    \end{tikzpicture}
    \caption{If a circuit $\gamma$ lies in $\T\setminus\T^0$ then any three consecutive hexagons on $\gamma$ are in the depicted constellation up to rotation and reflection (with $\gamma$ denoted by the dotted line).
        The set of vertices in $\partial\IntVert\gamma$ bordering the hexagon $z$ is then either the set $\{u,v\}$ or the set $\{u', v'\}$, and in both cases, constitutes an edge of $\HH^{\times}$. The same is true for $\partial\ExtVert\gamma$.}
    \label{fig:left-right-turns}
\end{figure}
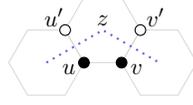

\begin{lemma}
  \label{cl:boundary-of-internal-vertices-is-connected}
  Let $\gamma \subset \T \setminus \T^0$ be a circuit. Then $\partial \IntVert\gamma$ and $\partial \ExtVert\gamma$ are connected in $\HH^{\times}$.
\end{lemma}

\begin{proof}
 Suppose $\gamma = (z_0,\ldots, z_m)$. Set $U$ to be either $\partial \IntVert\gamma$ or $\partial \ExtVert\gamma$ and let $U_i$ be the set of vertices in $U$ which border the hexagon $z_i$. The connectivity of $U$ in $\HH^{\times}$ is a consequence of the following statements:
\begin{enumerate}[label=(\alph*), ref=\alph*, labelindent=\parindent]
\item $U = \cup_{i} U_i$.
\item $U_i\cap U_{i+1}\neq\emptyset$ for $0\le i < m$.
\item $U_i$ is connected in $\HH^{\times}$ for all $i$.
\end{enumerate}
The first and second properties follow from Fact~\ref{fact:gamma-int-ext}. For the third property note that the only constellation up to rotation and reflection of three consecutive hexagons $z_{i-1}, z_i, z_{i+1}\in\T \setminus \T^0$ (where the indices are taken modulo $m$) on $\gamma$  is as depicted in Figure~\reffig{fig:left-right-turns}, so that the set $U_i$ has size $2$ and constitutes an edge in $\HH^{\times}$.
\end{proof}

\begin{lemma}\label{cl:bad-vertices-are-connected}
  Let $\omega$ be a loop configuration and let $\gamma\subset\T\setminus\T^0$ be a vacant circuit in $\omega$. If $\partial\IntVert\gamma\subset V(\omega,\gamma)$ then $V(\omega,\gamma)$ is connected in $\HH^{\times}$.
\end{lemma}

\begin{proof}
Let $E_1,\dots,E_m$ denote the clusters of $\omega$ inside $\gamma$ and write $\sigma_i := \sigma(E_i)$.
The connectivity of $V(\omega,\gamma)$ in $\HH^{\times}$ is a consequence of the following statements:
\begin{enumerate}[label=(\alph*), ref=\alph*, labelindent=\parindent]
\item $V(\omega,\gamma) = \IntVert\gamma \setminus \cup_i \IntVert{\sigma_i}$.
\item $\IntVert\gamma$ is connected in $\HH$.
\item $\partial \ExtVert{\sigma_i}$ is connected in $\HH^{\times}$ for all $i$.
\item $\partial \ExtVert{\sigma_i} \subset V(\omega,\gamma)$ for all $i$.
\end{enumerate}
The first property follows from the definition of $V(\omega,\gamma)$, the second from Fact~\ref{fact:gamma-int-ext} and the third from Lemma~\ref{cl:boundary-of-internal-vertices-is-connected} (and symmetry).
For the fourth property, note that $\partial \ExtVert{\sigma_i} \cap \IntVert\gamma \subset V(\omega,\gamma)$ by~\eqref{eq:clusters-are-disjoint}, and $\partial\ExtVert{\sigma_i} \subset \IntVert\gamma$ by the assumption that $\partial\IntVert\gamma \subset V(\omega,\gamma)$.
\end{proof}

\begin{lemma}
    \label{lem:prob-outer-circuit-intermediate}
    There exist absolute constants $C,c>0$ such that for any $n \ge C$ and any $x \in (0,\infty]$ satisfying $nx^6 \ge C$ the following holds.
    Let $\gamma\subset \T\setminus\T^0$ be a circuit, let $H'$ be a domain and set $\cE:=\LC(H', \ground^0 \cap \IntEdge\gamma)$.
    Then, for any integers $k \ge \ell \ge 0$, we have
    \[
    \Pr_{\Int\gamma,n,x}^\emptyset\big(\partial\IntVert\gamma\subset V(\omega,\gamma),~ |V(\omega,\gamma)|\ge k \text{ and } |V'(\omega,\gamma)|\ge \ell \,\mid\, \cE \big) \le \frac{(cn \cdot \min\{x^6,1\})^{-k/15}}{\max\{x^{\ell},1\}} .
    \]
\end{lemma}
\begin{proof}
    Let $\gamma\subset \T\setminus\T^0$ be a circuit and denote $H := \Int\gamma$.
    Let $n>0$ and let $x \in (0,\infty]$. We may assume throughout the proof that $n \cdot \min\{x^6,1\}$ is sufficiently large, as otherwise the statement is trivial.
    We shall show that for any $\emptyset \neq V \subset \IntVert\gamma$,
    \begin{equation}\label{eq:b}
    \Pr_{H,n,x}^\emptyset(V(\omega,\gamma) = V \text{ and } |V'(\omega,\gamma)| \ge \ell \,\mid\, \cE) \le \frac{(2\sqrt{2})^{|V|} \cdot (n \cdot \min\{x^6,1\})^{-|V|/15}}{\max\{x^\ell,1\}}.
    \end{equation}
    In light of Lemma~\ref{cl:bad-vertices-are-connected} and Lemma~\ref{lem:number-of-connected-graphs}, Lemma~\ref{lem:prob-outer-circuit-intermediate} will then follow from~\eqref{eq:b} by summing over all sets $V$ with $\partial \IntVert\gamma \subset V \subset \IntVert\gamma$ such that $V$ is connected in $\HH^\times$ and has cardinality at least $k$.

    In order to prove~\eqref{eq:b}, we shall apply Lemma~\ref{lem:prob-inequality-tool} to the (restricted) repair map
    \[
    \shiftFunc \colon \{ \omega \in \LC(H,\emptyset) \cap \cE ~:~ V(\omega,\gamma)=V \text{ and } |V'(\omega,\gamma)| \ge \ell \} \to \LC(H,\emptyset) \cap \cE ,\]
    which, by Lemma~\ref{lem:range-of-repair-map-new}, is well-defined.
    By Proposition~\ref{cor:shift-increases-weight-of-configuration}, we may take $p:=(n \cdot \min\{x^6,1\})^{|V|/15} \cdot \max\{x^\ell,1\}$. It remains to estimate, for each $V$, the maximum number of preimages under $\shiftFunc$ of a given loop configuration.

    Let $\omega$ be such that $V(\omega,\gamma) = V$ and let $E(V)$ be the set of edges with both endpoints in $V$. We claim that the set $\omega \setminus E(V)$ may be reconstructed from $\shift\omega$ and $V$. Indeed, $\omega \subset \IntEdge\gamma$ since $\omega \in \LC(H,\emptyset)$, and, for every $e \in \IntEdge\gamma \setminus E(V)$, we may determine whether $e \in \omega$ in the following way.
    Since $e$ has an endpoint $u_0 \in \IntVert\gamma \setminus V$, we see that $e$ belongs to a $\clr$-cluster $E$ of $\omega$ for some $\clr\in\{0,1,2\}$. In this case, $\omega \cap E$ equals either $\shift\omega \cap E$, $\shift\omega^{\DIRup} \cap E$ or $\shift\omega^{\DIRdown} \cap E$, depending on whether $\clr=0$, $\clr=1$ or $\clr=2$, respectively. Hence, it suffices to determine $\clr$ from $V$. To this end, consider a path from $u_0$ to $V$ in $\Int\gamma$, and let $\{u,v\}$ be the first edge on this path such that $u \notin V$ and $v \in V$.
    Observe that $u \in \IntVert{\sigma(E)}$ and $v \in \ExtVert{\sigma(E)}$ since $\partial \ExtVert{\sigma(E)} \cap \IntVert\gamma \subset V \subset \ExtVert{\sigma(E)}$ by~\eqref{eq:clusters-are-disjoint} and the definition of $V(\omega,\gamma)$.
    Thus, $\{u,v\} \in \sigma(E)^*$. Finally, since $\sigma(E) \subset \T \setminus \T^\clr$, we see that $\clr$ is the unique element in $\{0,1,2\}$ such that $y,z \notin \T^\clr$, where $\{y,z\}^*=\{u,v\}$.

    In conclusion, since given $V(\omega,\gamma)=V$, $\shift\omega$ uniquely determines $\omega \setminus E(V)$, the number of preimages of a given loop configuration $\shift\omega$ is at most the number of subsets of $E(V)$. Since there are at most $3|V|/2$ edges with both endpoints in $V$, there are at most $2^{3|V|/2}$ subsets of $E(V)$. Thus, Lemma~\ref{lem:prob-inequality-tool} implies~\eqref{eq:b}.
\end{proof}

\begin{proof}[Proof of Lemma~\ref{lem:prob-outer-circuit}]
Let $A$ be the event that $\partial\IntVert\gamma\subset V(\omega,\gamma)$ and $|V(\omega,\gamma)|\ge k$.
Denote $\cE:=\LC(H, \ground^0 \cap \IntEdge\gamma)$.
Using the fact that $\gamma$ is vacant in $\ground^0$, the domain Markov property implies that
\[ \Pr_{H,n,x}^0(A \mid \text{$\gamma$ vacant}) = \Pr_{\Int\gamma,n,x}^\emptyset(A \mid \cE).\]
Thus, the result follows from Lemma~\ref{lem:prob-outer-circuit-intermediate}.
\end{proof}

%
%%
%%%
%%%%%%%%%%%%%%
%%%
%%
%

\section{Proofs of main theorems}
\label{sec:proofs_main_theorems}

Throughout this section, we continue to use the notation introduced
in Section~\ref{sec:definitions}. The proofs of the main theorems
mostly rely on the main lemma, Lemma~\ref{lem:prob-outer-circuit}.

\subsection{Exponential decay of loop
lengths}\label{sec:exponential_decay_loop_lengths}

As mentioned in the introduction, the results for small $x$ follow
via a Peierls argument. The following lemma gives an upper bound on
the probability that a given collection of loops appears in a random
loop configuration.

\begin{lemma}\label{lem:given-loops-are-unlikely}
    Let $H$ be a domain and let $\xi$ be a loop configuration.
    Then, for any $n>0$, any $x>0$ and any $A \in \LC(H,\emptyset)$, we have
    \[ \Pr_{H,n,x}^\xi(A \subset \omega) \le n^{L_H(A)}x^{o_H(A)} .\]
\end{lemma}
\begin{proof}
    Consider the map
    \[ \sfT \colon \{ \omega \in \LC(H,\xi) :  A \subset \omega \} \to \LC(H,\xi) \]
    defined by
    \[ \sfT(\omega) := \omega \setminus A .\]
    Clearly, $\sfT$ is well-defined (see Lemma~\ref{lem:circuits-and-clusters}\ref{it:circuits-and-clusters5}) and injective.
    Moreover, since $L_H(\sfT(\omega)) = L_H(\omega) - L_H(A)$ and $o_H(\sfT(\omega)) = o_H(\omega) - o_H(A)$, we have
    \[ \Pr_{H,n,x}^\xi(\sfT(\omega)) = \Pr_{H,n,x}^\xi(\omega) \cdot n^{-L_H(A)} x^{-o_H(A)} .\]
    Hence, the statement follows from Lemma~\ref{lem:prob-inequality-tool}.
\end{proof}

Recall the notion of a loop surrounding a vertex given prior to Theorem~\ref{thm:no-large-loops}.

\begin{cor}\label{cl:no-large-loops-for-small-x}
  For any $n>0$, any $x>0$, any domain $H$, any vertex $u \in V(H)$ and any positive integer $k$, we have
  \[
  \Pr_{H,n,x}^\emptyset(\text{there exists a loop of length $k$ surrounding $u$}) \leq k n (2x)^k .
  \]
\end{cor}

\begin{proof}

Denote by $a_k$ the number of simple paths of length $k$ in $\HH$
starting at a given vertex. Clearly, $a_k \leq 3 \cdot 2^{k-1}$. It
is then easy to see that the number of loops of length $k$
surrounding $u$ is at most $k a_{k-1} \le k 2^k$. Thus, the result
follows by the union bound and
Lemma~\ref{lem:given-loops-are-unlikely}.
\end{proof}

Our main lemma, Lemma~\ref{lem:prob-outer-circuit}, shows that for a given circuit $\gamma$ (with a type) it is unlikely that the set $V(\omega,\gamma)$ is large. The set $V(\omega,\gamma)$ specifies deviations from the ground states which are `visible' from $\gamma$, i.e., deviations which are not `hidden' inside clusters. In Theorem~\ref{thm:no-large-loops}, we claim that it is unlikely to see long loops surrounding a given vertex. Any such long loop constitutes a deviation from all ground states. Thus, the theorem would follow from the main lemma (in the main case, when $x$ is large) if the long loop was captured in $V(\omega,\gamma)$. Our next lemma bridges the gap between the main lemma and the theorem, by showing that even when a deviation is not captured by $V(\omega,\gamma)$, there is necessarily a smaller circuit $\sigma$ which captures it in $V(\omega,\sigma)$.

\begin{lemma}\label{lem:existence-of-good-outer-circuit}
Let $\omega$ be a loop configuration, let $\clr\in\{0,1,2\}$ and let $\gamma \subset \T\setminus\T^\clr$ be a vacant circuit in $\omega$.
Let $U \subset \IntVert\gamma$ be non-empty and connected and assume that no vertex in $U$ belongs to a trivial loop in $\omega$. Then there exists $\clr'\in\{0,1,2\}$ and a circuit $\sigma \subset \T\setminus\T^{\clr'}$ such that $\Int\sigma \subset \Int\gamma$, $\sigma$ is vacant in $\omega$ and $U \cup \partial \IntVert\sigma \subset V(\omega,\sigma)$.
\end{lemma}
\begin{proof}
We prove the lemma by induction on $|\IntVert\gamma|$. We consider two cases.

Assume first that $\partial\IntVert\gamma \subset V(\omega,\gamma)$.
If $U \subset V(\omega,\gamma)$ then we are done, with $\sigma = \gamma$.
Otherwise, since $U$ is connected and no vertex in $U$ belongs to a trivial loop in $\omega$ it follows that $U$ is disjoint from $V(\omega, \gamma)$. Thus, using again the connectedness of $U$ and \eqref{eq:clusters-are-disjoint}, there is a cluster $E$ of $\omega$ inside $\gamma$ which contains all edges incident to vertices in $U$. Denote $\gamma' := \sigma(E)$ and observe that $\Int{\gamma'} \subsetneq \Int\gamma$ and that $\gamma'$ is vacant in $\omega$ by Lemma~\ref{lem:circuits-and-clusters}\ref{it:circuits-and-clusters2}. Hence, the lemma follows by applying the induction hypothesis with $\gamma'$ replacing $\gamma$.

Assume now that $\partial\IntVert\gamma \setminus V(\omega,\gamma) \neq \emptyset$. Let $u \in \partial\IntVert\gamma \setminus V(\omega,\gamma)$ and note that $u$ necessarily borders a $\clr$-flower $z$ of $\omega$. Consider the subgraph $H'$ induced by the vertices of $H$ which do not border $z$. Observe that $U \subset V(H')$ and, while $H'$ is not necessarily connected, each of its connected components is a domain of type $\clr$. Let $\gamma'$ be the circuit corresponding to the domain containing $U$. Now $\Int{\gamma'} \subset H' \subsetneq \Int\gamma$ and $\gamma'$ is vacant in $\omega$ as $\gamma$ is vacant and $z$ is a $\clr$-flower. Thus, the lemma follows by applying the induction hypothesis with $\gamma'$ replacing $\gamma$.
\end{proof}

\subsection*{Proof of Theorem~\ref{thm:no-large-loops}}
\label{sec:11}

Suppose that $n_0$ is a sufficiently large constant, let $n \ge n_0$ and let $x \in (0,\infty]$ be arbitrary. Let $\clr\in\{0,1,2\}$, let $H$ be a domain of type $\clr$ and let $u \in V(H)$. We shall estimate the probability that, in a random loop configuration drawn from $\Pr_{H,n,x}^\emptyset$, the vertex $u$ is surrounded by a non-trivial loop of length $k$. We consider two cases, depending on the relative values of $n$ and $x$.

Suppose first that $nx^6 < n^{1/50}$. Since $n \ge n_0$, we may
assume that $2x \le n^{-4/25}$ and that $kn^{-k/120} \le 1$ for all
$k>0$. By Corollary~\ref{cl:no-large-loops-for-small-x}, for every
$k \ge 7$,
\[
\begin{split}
  \Pr_{H,n,x}^\emptyset(\text{there exists a loop of length $k$ surrounding $u$}) & \le k n (2x)^k \le k n^{1-4k/25} \\
  & \le k n^{-k/60} \le n^{-k/120}.
\end{split}
\]

We now assume that $nx^6 \ge n^{1/50}$. Since $n \ge n_0$, we may assume that $n \cdot \min\{x^6,1\}$ is sufficiently large for our arguments to hold.
Let $L \subset H$ be a non-trivial loop of length $k$ surrounding $u$. Note that, if $\omega \in \LC(H,\emptyset)$ has $L \subset \omega$ then, by Lemma~\ref{lem:existence-of-good-outer-circuit}, for some $\clr'\in\{0,1,2\}$, there exists a circuit $\sigma \subset \T\setminus\T^{\clr'}$ such that $\Int\sigma \subset H$, $\sigma$ is vacant in $\omega$ and $V(L) \cup \partial \IntVert\sigma \subset V(\omega,\sigma)$.
Using the fact that $H$ is of type $\clr$ and the equivalence~\eqref{eq:domain_of_type}, the domain Markov property and Lemma~\ref{lem:prob-outer-circuit} imply that for every fixed circuit $\sigma \subset \T\setminus\T^{\clr'}$ with $\Int\sigma \subset H$,
\[
\Pr_{H,n,x}^\emptyset(\sigma\text{ vacant and }V(L)\cup\partial\IntVert\sigma \subset V(\omega,\sigma)) \le (c n \cdot \min\{x^6,1\})^{-|V(L)\cup \partial\IntVert\sigma|/15}.
\]
Thus, denoting by $\mathcal{G}(u)$ the set of circuits $\sigma$ contained in $\T\setminus\T^{\clr'}$ for some $\clr'\in\{0,1,2\}$ and having $u\in\IntVert\sigma$, we obtain
\[
\begin{split}
  \Pr_{H,n,x}^\emptyset(L \subset \omega)
  & \le \sum_{\sigma \in \mathcal{G}(u)} (c n \cdot \min\{x^6,1\})^{-|V(L)\cup\partial\IntVert\sigma|/15}\\
  & \le \sum_{\ell = 1}^{\infty} D^{\ell} (cn \cdot \min\{x^6,1\})^{-\max\{\ell,k\}/15} \\
  &\le (c' n \cdot \min\{x^6,1\})^{-k/15} ,
\end{split}
\]
where we used the facts that the length of a circuit $\sigma$ such that $|\partial \IntVert\sigma|=\ell$ is at most $3\ell$, that the number of circuits $\sigma$ of length at most $3\ell$ with $u \in \IntVert{\sigma}$ is bounded by $D^\ell$ for some sufficiently large constant $D$, and in the last inequality we used the assumption that $n \cdot \min\{x^6,1\}$ is sufficiently large.
Since the number of loops of length $k$ surrounding a given vertex is smaller than $k 2^k$ (see the proof of Corollary~\ref{cl:no-large-loops-for-small-x}), our assumptions that $nx^6 \ge n^{1/50}$ and $n \ge n_0$ yield
\[
\begin{split}
  \Pr_{H,n,x}^\emptyset(\text{there exists a loop of length $k$ surrounding $u$})
  & \le k 2^k (c' n^{1/50})^{-k/15} \le n^{-k/800}.
\end{split}
\]

%
%%
%%%
%%%%%%%%%%%%%%
%%%
%%
%

\subsection*{Proof of Theorem~\ref{thm:spin spin}} The proof is
very similar to that of Theorem~\ref{thm:no-large-loops}. The main
difference is the following replacement of
Lemma~\ref{lem:prob-outer-circuit}. Recall that in every $\lambda \in
\LC(H,\emptyset,u,v)$, there is a simple path between $u$ and $v$. Let $p(\lambda,u,v)$ be such a path and denote $\omega_\lambda := \lambda \setminus E(p(\lambda,u,v))$, so
that $\omega_\lambda \in \LC(H,\emptyset)$ and $L'_H(\lambda) = L_H(\omega_\lambda)$.
For a circuit $\gamma$ for which $\Int\gamma \subset H$ and for a positive integer $k$, let $\mathcal E(H,u,v,\gamma,k)$ be the set of configurations
$\lambda \in\LC(H,\emptyset,u,v)$ such that
\begin{itemize}
    \item $\gamma$ is vacant in $\omega_\lambda$;
    \item $V(p(\lambda,u,v)) \setminus \{u,v\}$ and $\partial\IntVert\gamma$ are contained in $V(\omega_\lambda,\gamma)$;
    \item $|V(\omega_\lambda,\gamma)|\ge k$.
\end{itemize}
For $\omega \in \LC(H,\emptyset)$ and $\lambda \in \LC(H,\emptyset,u,v)$, denote
\begin{align*}
\phi_{H,n,x}(\omega) &:= x^{o_H(\omega)}n^{L_H(\omega)} ,\\
\phi_{H,n,x}(\lambda) &:= x^{o_H(\lambda)}n^{L'_H(\lambda)} J(\lambda) .
\end{align*}

\begin{lemma}\label{lem:prob-outer-circuit-loopconfig-with-path}
    There exist absolute constants $C,c>0$ such that for any $n \ge C$ and $x\in(0,\infty)$ satisfying $n x^6 \ge C$ the following holds.
    For any domain $H$, any $\clr\in\{0,1,2\}$, any circuit $\gamma\subset \T\setminus\T^\clr$ for which $\Int\gamma \subset H$, any distinct vertices $u,v \in V(H)$ and any positive integer $k$, we have
    \[\sum_{\lambda\in \mathcal E(H,u,v,\gamma,k)} \phi_{H,n,x}(\lambda) \le x (cn \cdot \min\{x^6,1\})^{-k/15} \sum_{\omega\in \LC(H,\emptyset)} \phi_{H,n,x}(\omega) .\]
\end{lemma}

\begin{proof}
By symmetry, it suffices to consider the case that $\clr = 0$.
For $\ell \ge 0$, let $\mathcal E_\ell$ denote the set of $\lambda \in \mathcal E(H,u,v,\gamma,k)$ having $|V(p(\lambda,u,v)) \setminus \{u,v\}| = \ell$ and set $\mathcal E'_\ell := \{ \omega_\lambda : \lambda \in \mathcal E_\ell \}$.
Since $V(p(\lambda,u,v)) \setminus \{u,v\} \subset V(\omega_\lambda,\gamma)$, we have $|V'(\omega_\lambda,\gamma)| \ge \ell$ for any $\lambda \in \cE_\ell$.
Therefore, by Lemma~\ref{lem:prob-outer-circuit-intermediate},
\[ \sum_{\omega \in \mathcal E'_\ell} \phi_{H,n,x}(\omega) \le (cn \cdot \min\{x^6,1\})^{-k/15} \cdot \max\{x,1\}^{-\ell} \cdot \sum_{\omega\in \LC(H,\emptyset)} \phi_{H,n,x}(\omega) .\]
Since $J(\lambda) \le 3$ and $|E(p(\lambda,u,v))|=\ell+1$ for any $\lambda \in \cE_\ell$, we have
\[ \phi_{H,n,x}(\lambda) = \phi_{H,n,x}(\omega_\lambda) \cdot x^{|E(p(\lambda,u,v))|} J(\lambda) \le \phi_{H,n,x}(\omega_\lambda) \cdot 3x \cdot \max\{x,1\}^\ell .\]
Thus, noting that for every $\omega \in \cE'_\ell$,
\[ |\{ \lambda \in \cE_\ell ~:~ \omega_\lambda = \omega \}| \le \text{\#(simple paths of length $\ell+1$ from $u$ to $v$)} \le 3 \cdot 2^{\ell-1} \le 2^{\ell+1} ,\]
we obtain
\[ \sum_{\lambda \in \mathcal E_\ell} \phi_{H,n,x}(\lambda) \le 3x \cdot 2^{\ell+1} \cdot (cn \cdot \min\{x^6,1\})^{-k/15} \cdot \sum_{\omega\in \LC(H,\emptyset)} \phi_{H,n,x}(\omega) .\]
Finally, the lemma follows by summing over $0 \le \ell \le k$.
\end{proof}

We shall also require the following replacement of Corollary~\ref{cl:no-large-loops-for-small-x}.
\begin{lemma}\label{lem:no-large-loops-for-small-x-path}
  Let $n>0$ and $0<x\le\frac{1}{8}$. For any domain $H$ and any distinct $u,v \in V(H)$, we have
  \[
  \sum_{\lambda \in \LC(H,\emptyset,u,v)} \phi_{H,n,x}(\lambda) \le 3(2x)^{d_H(u,v)} \sum_{\omega \in \LC(H,\emptyset)} \phi_{H,n,x}(\omega).
  \]
\end{lemma}
\begin{proof}
The number of possibilities for a simple path of length $k$ from $u$ to $v$ is at most $3\cdot2^{k-2}$. Consideration of the map $\lambda \mapsto \omega_\lambda$, the fact that $J(\lambda)\le 3$ and summation over all possibilities for $p(\lambda,u,v)$ now shows that the ratio of the sums appearing in the lemma is bounded above by
\begin{equation*}
\sum_{k\ge d_H(u,v)} 3 x^k(3\cdot 2^{k-2}) = \frac{9}{4}\cdot\frac{(2x)^{d_H(u,v)}}{1-2x}\le 3 (2x)^{d_H(u,v)}.\qedhere
\end{equation*}
\end{proof}

We now proceed along the
same lines as the proof of Theorem~\ref{thm:no-large-loops}. Suppose first that $nx^6 < n^{1/2}$. Since $n \ge n_0$, the theorem follows as an immediate consequence of Lemma~\ref{lem:no-large-loops-for-small-x-path}. Suppose now that $nx^6 \ge n^{1/2}$. For each $\lambda\in \LC(H,\emptyset,u,v)$, by Lemma~\ref{lem:existence-of-good-outer-circuit} applied to $\omega_\lambda$, there exists a circuit $\sigma \subset \T\setminus\T^{\clr'}$ for some $\clr'\in\{0,1,2\}$ such that $\Int\sigma \subset H$ and $\lambda\in  \mathcal E(H,u,v,\sigma,k_\sigma)$, where $k_\sigma := \max\{d_H(u,v)-1, |\partial \IntVert\sigma|\}$. The theorem now follows with a similar calculation as in Theorem~\ref{thm:no-large-loops}, by summing over all possibilities for the circuit $\sigma$ and applying Lemma~\ref{lem:prob-outer-circuit-loopconfig-with-path} with $\gamma = \sigma$ and $k = k_\sigma$.

%
%%
%%%
%%%%%%%%%%%%%%
%%%
%%
%

\subsection{Small perturbation of ground state}\label{sec:proof-of-theorem-small-dev}

\subsection*{Proof of Theorem~\ref{thm:small-deviations-from-the-ground-state-x-model}}
 By definition, the subgraph
of $\HH$ induced by $\breakup(\omega,u)$ is a domain when it is
non-empty. Let $\Gamma(\omega,u)$ be the circuit satisfying
$\breakup(\omega,u)=\IntVert{\Gamma(\omega,u)}$. It follows that
$\Gamma(\omega,u)$ is vacant and contained in $\T\setminus \T^0$. To
see this, note that the edge boundary of $B(\omega)$ consists only
of edges $\{v,w\}$ such that $w$ borders a $0$-flower
$y$ and $v$ is the unique neighbor of $v$ not bordering $y$; in particular, $\{v,w\}$ borders a hexagon from $\T^1$ and a
hexagon from $\T^2$ and $\{v,w\} \not\in \omega$. Furthermore,
$\partial \breakup(\omega,u)\subset V(\omega,\Gamma(\omega,u))$.
This follows as $\Gamma(\omega,u)$ is vacant in
$\omega$ and, by the definition of $B(\omega)$, no vertex of $\partial\IntVert{\Gamma(\omega,u)}$
belongs to a trivial loop surrounding a hexagon in $\T^0$.

Now, denoting by $\mathcal{G}_k(u)$ the set of circuits $\gamma
\subset \T \setminus \T^0$ having $u\in\IntVert\gamma$
and $|\partial \IntVert\gamma| \ge k$,
Lemma~\ref{lem:prob-outer-circuit} implies that
\[
\begin{split}
\Pr_{H,n,x}^0(|\partial \breakup(\omega,u)| \geq k)&= \sum_{\gamma \in \mathcal{G}_k(u)}\Pr_{H,n,x}^0(\Gamma(\omega,u)=\gamma)\\
&\le \sum_{\gamma \in \mathcal{G}_k(u)}\Pr_{H,n,x}^0(\gamma\text{ vacant and }\partial\IntVert\gamma\subset V(\omega,\gamma))\\
&\leq \sum_{\gamma \in \mathcal{G}_k(u)} (c n \cdot \min\{x^6,1\})^{-|\partial\IntVert\gamma|/15} \\
&\leq \sum_{\ell\ge k} D^{\ell}(c n \cdot \min\{x^6,1\})^{-\ell/15}
\leq (c' n \cdot \min\{x^6,1\})^{-k/15},
\end{split}
\]
where $c',D$ are positive constants. In the final inequality, we used the facts that the length of a circuit $\gamma$ such that $|\partial \IntVert\gamma|=\ell$ is at most $3\ell$, and that the number of circuits of length at most $3\ell$ surrounding $u$ is bounded by $D^\ell$ for some sufficiently large constant $D$.

%
%%
%%%
%%%%%%%%%%%%%%
%%%
%%
%

\subsection{Limiting Gibbs measures}
Before proving the last two theorems, we require the following two lemmas.
We say that a circuit $\gamma$
\emph{surrounds} a subgraph $A \subset \HH$ if $A \subset
\Int\gamma$ and that $\gamma$ is \emph{inside} $A$ if $\Int\gamma
\subset A$. We say that a circuit $\gamma$ \emph{contains} a circuit
$\sigma$ if $\Int\sigma \subset \Int\gamma$.

\begin{lemma}\label{lem:marginal-distribution-given-vacant-circuit}
Let $H$ and $H'$ be two domains, let $A \subset H \cap H'$ be a non-empty subgraph and let $\xi$ and $\xi'$ be loop configurations.
Let $n>0$ and $x\in(0,\infty]$. Let $\omega \sim \Pr_{H,n,x}^\xi$ and $\omega' \sim \Pr_{H',n,x}^{\xi'}$ be independent.
Denote by $\Omega$ the event that there exists a circuit surrounding $A$ and inside $H \cap H'$ which is vacant in both $\omega$ and $\omega'$. Assume that $\Omega$ has positive probability. Then, conditioned on $\Omega$, the marginal distributions of $\omega$ and $\omega'$ on $A$ are equal.
%Moreover, this marginal distribution is a convex combination of the marginal distributions of $\Pr_{\Int\gamma,n,x}^\emptyset$ on $A$, where $\gamma\in\Gamma$.
\end{lemma}
\begin{proof}
In this proof, a doubly-vacant circuit is a circuit which is vacant
in both $\omega$ and $\omega'$. Let $\mathcal G$ denote the
collection of circuits surrounding $A$ and inside $H \cap H'$. Let
$\sigma,\sigma' \in \mathcal G$ be
doubly-vacant circuits. Then, since both circuits surround $A$,
$\Int{\sigma} \cap \Int{\sigma'} \neq \emptyset$. By
Fact~\ref{fact:circuits-max}, there exists a circuit $\gamma$ having $\gamma^* \subset \sigma^* \cup
(\sigma')^*$ which contains both $\sigma$ and $\sigma'$. Clearly,
$\gamma$ is doubly-vacant, surrounds $A$ and is inside $H \cap H'$,
and hence $\gamma \in \mathcal G$. Thus, we have a notion of the
``outermost'' doubly-vacant circuit in $\mathcal G$. On $\Omega$,
define $\Gamma$ to be this circuit. Then, we claim that, for any
circuit $\gamma \in \mathcal G$ for which the event $\Omega \cap
\{\Gamma=\gamma\}$ has positive probability, conditioned on $\Omega
\cap \{ \Gamma=\gamma \}$, the marginal distribution of
$(\omega,\omega')$ on $A^2$ is the same as the marginal distribution
of two independent loop configurations sampled from
$\Pr_{\Int\gamma,n,x}^\emptyset$. Indeed, since the event $\Omega
\cap \{ \Gamma = \gamma \}$ is determined by $\omega \setminus
\IntEdge\gamma$ and $\omega' \setminus \IntEdge\gamma$, this follows from
the domain Markov property.
%Therefore,
%\[ \Pr(\omega_{|A},~\omega'_{|A} \mid \Omega) = \sum_{\gamma\in\Gamma} \Pr(\tau=\gamma \mid \Omega) \cdot \Pr_{\Int\gamma,n,x}^\emptyset(\omega_{|A}) \cdot \Pr_{\Int\gamma,n,x}^\emptyset(\omega'_{|A}) .\]
%In particular, conditioned on $\Omega$, the marginal distributions of $\omega$ and $\omega'$ on $A$ are equal and are given by the convex combination
%\[ \sum_{\gamma\in\Gamma} \Pr(\tau=\gamma \mid \Omega) \cdot \Pr_{\Int\gamma,n,x}^\emptyset(\cdot_{|A}) . \qedhere \]
\end{proof}

\begin{lemma}\label{lem:limit-is-gibbs}
Let $H_k$ be an increasing sequence of domains such that $\cup_k H_k = \HH$ and let $\xi_k$ be a sequence of loop configurations. Let $n>0$ and $x\in(0,\infty]$ and assume that $\Pr_{H_k,n,x}^{\xi_k}$ converges (weakly) as $k \to \infty$ to an infinite-volume measure $\Pr$ which is supported on loop configurations with no infinite paths. Then $\Pr$ is a Gibbs measure for the loop $O(n)$ model with edge weight $x$.
\end{lemma}

\begin{proof}
For a domain $H$, denote by $\cF_H$ the sigma algebra generated by the events $\{ e \in \omega \}$ for $e \in E(H)$. For a loop configuration $\tau$, let $\mathcal{E}_m^\tau$ be the event that $\omega$ and $\tau$ coincide on $E(H_m) \setminus E(H)$. By L\'evy's zero-one law, $\Pr$ is a Gibbs measure if and only if for every domain $H$ and every $A\in\cF_H$,
\[
\lim_{m \to \infty} \Pr(A \mid \mathcal{E}_m^\tau) = \Pr_{H,n,x}^{\tau}(A)\quad\text{for $\Pr$-almost every $\tau$}.
\]
Fix a domain $H$ and $A\in\cF_H$. By the definition of $\Pr$, we need to show that
\[ \lim_{m \to \infty} \lim_{k \to \infty} \Pr_{H_k,n,x}^{\xi_k}(A \mid \mathcal{E}_m^\tau) = \Pr_{H,n,x}^{\tau}(A)\quad\text{for $\Pr$-almost every $\tau$}.\]
Indeed, for any $\tau$ having a vacant circuit $\gamma$ with $H\subset\Int\gamma$, the domain Markov property implies that $\Pr_{H_k,n,x}^{\xi_k}(A \mid \mathcal{E}_m^\tau) = \Pr_{H,n,x}^{\tau}(A)$ for large enough $m$ and $k\ge m$. As $\Pr$ is supported on loop configurations with no infinite paths, such a circuit exists for $\Pr$-almost every $\tau$ (consider the smallest domain containing $V(H)$ and all the connected components of $\tau$ which intersect $V(H)$ and apply Fact~\ref{fact:circuit-domain-bijection}).
\end{proof}

%
%%
%%%
%%%%%%%%%%%%%%
%%%
%%
%

\subsection*{Proof of Theorem~\ref{thm:Gibbs
measures_small_x}}\label{sec:unique-limit-for-free-boundary-conditions}

We start with a lemma.

\begin{lemma}\label{lem:empty space}
    Let $n>0$ and $x>0$. For any two domains $H$ and $H'$, any vertex $u\in \VH$ and any positive integer $k$, we have
    \[
    \Pr(\text{the connected component of $u$ in $\omega\cup\omega'$ has exactly $k$ edges}) \le (9e \max\{n^{1/6},1\} x)^k,
    \]
    where $\omega\sim\Pr_{H,n,x}^\emptyset$ and $\omega'\sim\Pr_{H',n,x}^\emptyset$ are independent.
\end{lemma}

\begin{proof}
    We may assume that $\max\{n^{1/6},1\}x \le 1$, since the statement is trivial otherwise.
    Let $\mathcal C_k$ be the set of connected subgraphs of $\HH$ that have exactly $k$ edges and contain $u$. For $S\in\mathcal C_k$, call a pair of loop configurations $(A,A')$ {\em compatible with $S$} if $E(A)\cup E(A')=E(S)$.
    Let $\mathsf S$ be the connected component of $u$ in $\omega\cup\omega'$. Then
    \[
    \begin{split}
    \Pr(|E(\mathsf S)|= k) & \le \sum_{S\in\mathcal C_k}\ \sum_{(A,A')\text{ compatible with }S}\Pr(A\subset\omega, ~ A'\subset\omega')\\
    &\le \sum_{S\in\mathcal C_k}\ \sum_{(A,A')\text{ compatible with }S}(\max\{n^{1/6},1\}x)^{o_{H}(A)+o_{H'}(A')}\\
    &\le (9e)^k(\max\{n^{1/6},1\}x)^k.
    \end{split}
    \]
    The second inequality follows from Lemma~\ref{lem:given-loops-are-unlikely} and the facts that $\omega$ and $\omega'$ are independent and that any loop consists of at least six edges.
    The last inequality follows from the following three facts:
    \begin{itemize}[noitemsep,nolistsep]
        \item
        $o_H(A)+o_{H'}(A')\ge |E(S)|= k$ and $\max\{n^{1/6},1\}x \le 1$;
        \item
        the number of possible pairs of loop configurations $(A,A')$ compatible with $S$ is bounded by $3^k$ (since each edge in $S$ must be in either $A$, $A'$ or in both);
        \item $|\mathcal C_k|$ is bounded by $3 (3e)^{k-1}\le (3e)^k$ (apply Lemma~\ref{lem:number-of-connected-graphs} to the 4-regular line graph of $\HH$, using an edge incident to $u$ as the given vertex). \qedhere
    \end{itemize}
\end{proof}

    Let us conclude the proof of Theorem~\ref{thm:Gibbs measures_small_x}.
    Assume that $9e \max\{n^{1/6},1\} x \le 1/e$.
    Let $H$ and $H'$ be two domains and let $A\subset B\subset H \cap H'$ be two sub-domains. Let $\omega\sim\Pr_{H,n,x}^\emptyset$ and $\omega'\sim\Pr_{H',n,x}^\emptyset$ be independent. Let $\mathcal E$ be the event that the union of the connected components of the vertices of $A$ in the graph $\omega\cup\omega'$ intersects $\VH\setminus V(B)$. Lemma~\ref{lem:empty space} implies that
    \begin{equation}
    \label{eq:rr}
    \Pr(\mathcal E)\le \sum_{v \in V(A)} \sum_{k=d(\{v\},\VH\setminus V(B))}^\infty (9e \max\{n^{1/6},1\} x)^k \le 2|V(A)| \cdot e^{-d(A,\VH\setminus V(B))} ,
    \end{equation}
    where $d(E,F)$ is the minimum of the graph distances between a vertex in $E$ and a vertex in $F$.

    Let us now show that, on the complement of $\mathcal E$, there exists a circuit $\gamma$ surrounding $A$ and inside $H \cap H'$ which is vacant in both $\omega$ and $\omega'$.
We first define the notion of the \emph{outer circuit} of a
non-empty finite connected subset $U$ of $\VH$. Let $U'$ be the
unique infinite connected component of $\VH \setminus U$ and let
$U'' := \VH\setminus U'$. Evidently, the subgraph of $\HH$ induced by $U''$ is a domain containing
$U$. The outer circuit $\sigma$ of $U$ is then the circuit
corresponding to this domain, i.e., $U''
= \IntVert\sigma$, which exists by
Fact~\ref{fact:circuit-domain-bijection}. Note also that
$\partial U'' \subset \partial U$ and that if $U$ is
contained in some domain then $U''$ is also contained in the same
domain.

    Let $\mathsf D$ be the union of the connected components of vertices of $A$ in $\omega\cup\omega'$.
    Let $\gamma$ be the outer circuit of $V(A) \cup \sf D$, and note that, on the complement of $\mathcal E$, $\gamma$ is inside $B$.
    Let us show that $\gamma$ is vacant in both $\omega$ and $\omega'$.
    To this end, let $e=(u,v) \in \gamma^*$ be an edge with $u \in V(A) \cup \sf D$ and $v \notin V(A) \cup \sf D$. Assume first that $u \in \sf D$.
    Clearly $e \notin \omega \cup \omega'$, as otherwise, $v$ would also belong to $\sf D$. Assume now that $u \in V(A) \setminus \sf D$. Then, by definition of $\sf D$, $u$ is not contained in a loop of neither $\omega$ nor $\omega'$. In particular, $e$ does not belong to neither $\omega$ nor $\omega'$. Thus, $\gamma$ is vacant in both $\omega$ and $\omega'$.

    Thus, by Lemma~\ref{lem:marginal-distribution-given-vacant-circuit}, the total variation between the measures $\Pr_{H,n,x}^\emptyset(\cdot_{|A})$ and $\Pr_{H',n,x}^\emptyset(\cdot_{|A})$ is at most $\Pr(\mathcal E)$.
    In light of~\eqref{eq:rr}, by taking $B$ large enough, we may make $\Pr(\mathcal E)$ arbitrarily small.
    This implies the convergence of the measures $\Pr_{H_k,n,x}^\emptyset(\cdot_{|A})$ towards a limit. Since this holds for any domain $A$, we have established the convergence of $\Pr_{H_k,n,x}^\emptyset$ as $k \to \infty$ towards an infinite-volume measure $\Pr_{\HH,n,x}$.

The fact that $\Pr_{\HH,n,x}$ is supported on loop
configurations with no infinite paths is an immediate consequence of Corollary~\ref{cl:no-large-loops-for-small-x}.
Indeed, the corollary shows that in the measure $\Pr_{H_k,n,x}^\emptyset$, the probability that a given vertex is contained in a loop of length $m$ tends to zero with $m$, uniformly in $k$.
Finally, the fact that $\Pr_{\HH,n,x}$ is a Gibbs measure follows from Lemma~\ref{lem:limit-is-gibbs}.

\subsection*{Proof of Theorem~\ref{thm:Gibbs measures_large_x}}

Let us first assume that the convergence to the limiting measures $\{\Pr_{\HH,n,x}^{\clr}\}_{\clr\in\{0,1,2\}}$ holds and deduce the properties of these measures when $n \cdot \min\{x^6,1\}$ is sufficiently large.
By Theorem~\ref{thm:small-deviations-from-the-ground-state-x-model}, if $n \cdot \min\{x^6,1\}$ is sufficiently large then, for any $z \in \T^0$,
\[ \Pr_{\HH,n,x}^0(z\text{ is surrounded by a trivial loop}) > 1/2 .\]
Since $\Pr_{\HH,n,x}^1$ and $\Pr_{\HH,n,x}^2$ are the measures induced by applying the shifts $\DIRdown$ and $\DIRup$, respectively, to $\Pr_{\HH,n,x}^0$, the same statement holds for any $\Pr_{\HH,n,x}^\clr$ with $z \in \T^\clr$.
Thus, since adjacent hexagons cannot both be surrounded by trivial loops simultaneously, we conclude that the measures $\{ \Pr_{\HH,n,x}^{\clr} \}_{\clr\in\{0,1,2\}}$ are not convex combinations of one another. Next, the fact that $\Pr_{\HH,n,x}^\clr$ is supported on loop
configurations with no infinite paths is an immediate consequence of Theorem~\ref{thm:no-large-loops} (by using~\eqref{eq:domain_of_type} and applying the convergence result with an exhausting sequence of domains of type $\clr$).
Finally, the fact that $\Pr_{\HH,n,x}^\clr$ is a Gibbs measure follows from Lemma~\ref{lem:limit-is-gibbs}.

It remains to show that, for any $\clr\in\{0,1,2\}$, $\Pr_{H_k,n,x}^\clr$ converges as $k \to \infty$ to an infinite-volume  measure $\Pr_{\HH,n,x}^\clr$. Without loss of generality, we may assume that $\clr=0$.
The proof bears similarity with the proof of Theorem~\ref{thm:Gibbs measures_small_x}.

We start with a lemma. Recall the definition of $B(\omega)$ and
$\breakup(\omega,u)$ from Section~\ref{sec:results} and recall the definition of $\HH^{\times}$ from Section~\ref{sec:proof of lemma}. For a domain
$H$ and a loop configuration $\omega \in \LC(H,\ground^0)$, set
$\breakup(\omega):=\VH\setminus B(\omega)=\cup_{u\in \VH}
\breakup(\omega,u)$. Note that, by definition, every two breakups
$\breakup(\omega,u)$ and $\breakup(\omega,v)$, where $u,v \in
\VH$, are either equal or their union is disconnected in
$\HH^\times$ (as the definition implies that if a vertex belongs to
$\breakup(\omega)$ then all vertices bordering the same hexagon in
$\T^0$ also belong to $\breakup(\omega)$). Thus, every connected
component of $\breakup(\omega)$ is a breakup of some vertex, and
every $\HH^\times$-connected component of $\partial\breakup(\omega)$
is the boundary of a breakup of some vertex, i.e., equals
$\partial\breakup(\omega,u)$ for some $u \in \VH$ (recall that
this set is $\HH^\times$-connected, by
Lemma~\ref{cl:boundary-of-internal-vertices-is-connected}).

\begin{lemma}\label{lem:empty spaceA}
There exists an absolute constant $c>0$ such that for any $n>0$ and $x\in(0,\infty]$ the following holds.
For any two domains $H$ and $H'$, any vertex $u\in \VH$ and any positive integer $k$,
\[ \Pr(\text{the $\HH^\times$-connected component of $u$ in }\partial\breakup(\omega)\cup \partial\breakup(\omega')\text{ has cardinality }k)\le (c n \cdot \min\{x^6,1\})^{-k/15} ,\]
where $\omega\sim\Pr_{H,n,x}^0$ and $\omega'\sim\Pr_{H',n,x}^0$ are independent.
\end{lemma}

\begin{proof}
Let $\mathcal C_k$ be the set of $\HH^\times$-connected subsets of
$\VH$ of cardinality $k$ containing $u$. For $S\in\mathcal C_k$,
call a pair $(A,A')$ of subsets of $\VH$ {\em compatible with
$S$} if $A\cup A'=S$. We write $A \prec \breakup(\omega)$ if $A$ is
the union of some $\HH^\times$-connected components of $\partial
\breakup(\omega)$, or equivalently, if every $\HH^\times$-connected
component of $A$ is equal to $\partial\breakup(\omega,v)$ for some
$v \in \VH$. Now, we claim that for each fixed $A$, we have
\begin{equation}\label{eq:A_subset_of_breakup}
\Pr_{H,n,x}^0(A \prec \breakup(\omega)) \le (c n \cdot \min\{x^6,1\})^{-|A|/15}.%\\
\end{equation}
To see this, note that for the probability to be positive, $A$ needs
to be a union of $\partial\IntVert{\gamma_i}$ for a collection of
circuits $\gamma_i\subset\T\setminus\T^0$ with disjoint interiors. Moreover, on the event $A \prec \breakup(\omega)$ these circuits are necessarily vacant in $\omega$.
Therefore, by conditioning on all of the $\gamma_i$ being vacant, we may
apply the domain Markov property and
Theorem~\ref{thm:small-deviations-from-the-ground-state-x-model} to
obtain the estimate~\eqref{eq:A_subset_of_breakup}. Similarly, for
each fixed $A'$ we have that
\begin{equation*}
  \Pr_{H',n,x}^0(A' \prec \breakup(\omega')) \le (c n \cdot \min\{x^6,1\})^{-|A'|/15}.
\end{equation*}
We may assume that $c n \cdot \min\{x^6,1\} \ge 1$, since the
statement is trivial otherwise. Let $\mathsf S$ be the
$\HH^\times$-connected component of $u$ in $\partial\breakup(\omega)
\cup \partial\breakup(\omega')$. Then
\begin{align*}
\Pr(|\mathsf S|= k)&\le \sum_{S\in\mathcal C_k}\ \sum_{(A,A')\text{ compatible with }S} \Pr(A \prec \breakup(\omega), ~ A' \prec \breakup(\omega'))\\
&\le \sum_{S\in\mathcal C_k}\ \sum_{(A,A')\text{ compatible with }S}(c n \cdot \min\{x^6,1\})^{-(|A|+|A'|)/15}\\
&\le (15e)^k (c n \cdot \min\{x^6,1\})^{-k/15}.
\end{align*}
In the second inequality we used the fact that $\omega$ and $\omega'$ are independent.
The last inequality follows from the following three facts:
\begin{itemize}[noitemsep,nolistsep]
    \item $|A|+|A'|\ge |S|= k$ and $cn \cdot \min\{x^6,1\} \ge 1$;
    \item the number of possible pairs $(A,A')$ compatible with $S$ is bounded by $3^k$ (since each vertex in $S$ is either in $A$, in $A'$ or in both);
    \item $|\mathcal C_k|$ is bounded by $(5e)^{k-1}\le (5e)^k$ (apply Lemma~\ref{lem:number-of-connected-graphs} to the $6$-regular graph $\HH^\times$). \qedhere
\end{itemize}
\end{proof}

Let us conclude the proof of Theorem~\ref{thm:Gibbs measures_large_x}.
Let $c>0$ be the minimum between the constants from the statements of Lemma~\ref{lem:empty spaceA} and Theorem~\ref{thm:small-deviations-from-the-ground-state-x-model}, and assume that $cn \cdot \min\{x^6,1\} \ge e^{15}$.
Let $H$ and $H'$ be two domains and let $A \subset B \subset H \cap H'$ be two domains of type $0$. Let $\omega\sim\Pr_{H,n,x}^0$ and $\omega'\sim\Pr_{H',n,x}^0$ be independent. Let $\mathcal E$ be the event that the union of $\HH^\times$-connected components of vertices in $A$ in $\partial\breakup(\omega) \cup \partial\breakup(\omega')$ intersects $\VH \setminus V(B)$. Lemma~\ref{lem:empty spaceA} implies that
\[ \Pr(\mathcal E)
\le \sum_{u \in V(A)} \sum_{k=d(\{u\},\VH\setminus V(B))}^\infty (c n \cdot \min\{x^6,1\})^{-k/15}
\le 2|V(A)| \cdot e^{-d(A,\VH\setminus V(B))},\]
where $d(E,F)$ is the minimum of the graph distances between a vertex in $E$ and a vertex in $F$.
Let $\mathcal E'$ be the event that $A$ is contained in either $\breakup(\omega)$ or $\breakup(\omega')$, i.e., that $A$ is contained entirely in one breakup (of either $\omega$ or $\omega'$). Denote by $\rho(m)$ the smallest possible size of $\partial U$ for a finite subset $U \subset \VH$ of size at least $m$.
Then Theorem~\ref{thm:small-deviations-from-the-ground-state-x-model} implies that
\[ \Pr(\mathcal E') \le 2(c n \cdot \min\{x^6,1\})^{-\rho(|V(A)|)/15} \le 2e^{-\rho(|V(A)|)} .\]

Let us now show that, on the complement of $\mathcal E \cup \mathcal
E'$, there exists a circuit $\gamma \subset \T\setminus\T^0$
surrounding $A$ and inside $H \cap H'$ which is vacant in both
$\omega$ and $\omega'$. We require the following simple geometric
claim. For brevity, in the rest of the proof we identify a domain
with its set of vertices.
\begin{equation}\label{eq:S_T_boundaries_claim}
  \begin{aligned}
    &\text{If $S,T$ are two domains of type $0$ with $S\not\subset T$
    and $T\not\subset S$ such that $S\cup T$ is connected,}\\[-0.4em]
    &\text{then $\partial{S}\cup\partial T$ is
    $\HH^\times$-connected. If, in addition, $S\cap T\neq\emptyset$ then also $\partial S\cap  T\neq\emptyset$.}
  \end{aligned}
\end{equation}
To see this, note first that $\partial S$ and $\partial T$ are
$\HH^\times$-connected by Fact~\ref{fact:circuit-domain-bijection}
and Lemma~\ref{cl:boundary-of-internal-vertices-is-connected}. If
$S\cap T=\emptyset$ then the assumption that $S\cup T$ is connected
implies that a vertex of $\partial S$ is adjacent to a vertex of
$\partial T$ yielding that $\partial S\cup\partial T$ is
$\HH^\times$-connected. Assume that $S\cap T\neq \emptyset$. By
considering a path in $T$ from $T\setminus S$ to $T\cap S$ it
follows that $\partial S\cap T\neq\emptyset$. Similarly, considering
a path in $T^c$ from $S\setminus T$ to $(S\cup T)^c$ shows that
$\partial S\setminus T\neq\emptyset$.
Finally, by considering a $\HH^\times$-path in $\partial S$ from $\partial S\cap T$ to $\partial S \setminus T$, we see that either $\partial S \cap \partial T \neq \emptyset$ or a vertex of $\partial S$ is adjacent to a vertex of $\partial T$. In either case, we conclude that $\partial S\cup\partial T$ is $\HH^\times$-connected.

Recall the notion of the outer circuit of a non-empty finite
connected subset $U$ of $\VH$ from the proof of
Theorem~\ref{thm:Gibbs measures_small_x}. Let $\sf D$ be the union of $A$ and of the connected components of $\breakup(\omega) \cup \breakup(\omega')$ that intersect $A$. Let $\gamma$ be the outer
circuit of $\sf D$. It follows that $\gamma \subset
\T\setminus\T^0$ and that $\gamma$ is vacant in both $\omega$ and
$\omega'$. Indeed, $\gamma\subset\T\setminus\T^0$ since $A$ is a
domain of type $0$ and, by the definition of the breakup, each of the $\breakup(\omega,u)$
is a domain of type $0$. Thus, no edge of $\gamma^*$ can belong to
$\omega\cup\omega'$ since otherwise both its endpoints would belong
to a breakup.

We claim that, on the complement of $\mathcal E \cup \mathcal E'$,
$\gamma$ is inside $B$. By the definition of $\gamma$ and since $B$
is a domain, it suffices to show that ${\sf D} \subset B$. On the
complement of $\mathcal E'$, we may write $\sf D$ as the
union of domains $D_i$ of type $0$ such that no one contains
another, $D_0=A$ and each $D_i$, $i \neq 0$, is a breakup of either
$\omega$ or $\omega'$. Let ${\sf D}'$ be the union of $A$ and of the $\HH^\times$-connected components of $\partial\breakup(\omega)\cup
\partial\breakup(\omega')$ that intersect $A$. On the complement of $\mathcal E$,
we have ${\sf D}' \subset B$.
By~\eqref{eq:S_T_boundaries_claim}, $\cup_i
\partial D_i$ is $\HH^\times$-connected and if $D_i\cap A\neq\emptyset$ then $\partial D_i\cap A\neq\emptyset$. Thus
$\cup_i \partial D_i\subset {\sf D}'$. We conclude
that $\partial {\sf D} \subset \cup_i
\partial D_i \subset B$, whence ${\sf D} \subset
B$ as we wanted to show.

Thus, by Lemma~\ref{lem:marginal-distribution-given-vacant-circuit}, the total variation between the measures $\Pr_{H,n,x}^0(\cdot_{|A})$ and $\Pr_{H',n,x}^0(\cdot_{|A})$ is at most $\Pr(\mathcal E \cup \mathcal E')$.
In particular, fixing a subgraph $A' \subset A$, the same holds for the measures $\Pr_{H,n,x}^0(\cdot_{|A'})$ and $\Pr_{H',n,x}^0(\cdot_{|A'})$.
Since $\rho(m)$ clearly tends to infinity as $m$ tends to infinity, by first taking $A$ large enough and then taking $B$ large enough, we may make $\Pr(\mathcal E \cup \mathcal E')$ arbitrarily small.
This implies the convergence of the measures $\Pr_{H_k,n,x}^0(\cdot_{|A'})$ towards a limit. Since this holds for any finite subgraph $A'$ of $\HH$, we have established the convergence of $\Pr_{H_k,n,x}^0$ as $k \to \infty$ towards an infinite-volume measure $\Pr_{\HH,n,x}^0$.

\section{Discussion and open questions}
\label{sec:discussion_open_questions}

In this work, we investigate the structure of loop configurations in the loop $O(n)$ model
with large parameter $n$. We show that the chance of having a loop
of length $k$ surrounding a given vertex decays exponentially in
$k$. In addition, we show, under appropriate boundary conditions,
that if $nx^6$ is small, the model is in a dilute, disordered phase
whereas if $nx^6$ is large, configurations typically resemble one of
the three ground states. In this section, we briefly discuss several
future research directions.

%\newgeometry{left=15mm,bottom=15mm,top=10mm}
\begin{figure}
    \centering
    \begin{subfigure}[t]{.5\textwidth}
        \includegraphics[scale=0.46]{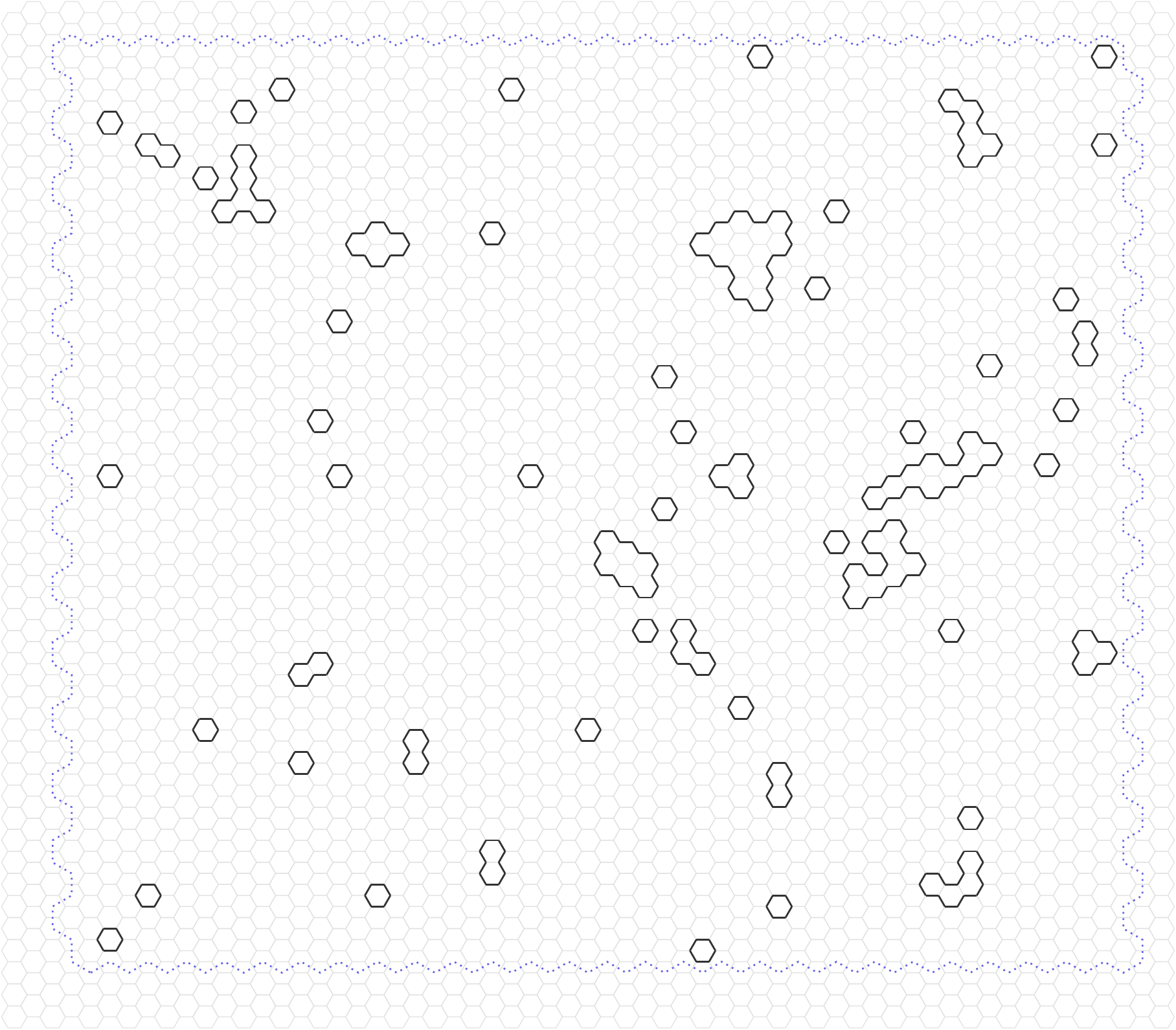}
        \caption{$n=0.8$ and $x=0.55$.}
        \label{fig:loop-sample-n=0.8,x=0.55}
    \end{subfigure}%
    \begin{subfigure}{20pt}
        \quad
    \end{subfigure}%
    \begin{subfigure}[t]{.5\textwidth}
        \includegraphics[scale=0.46]{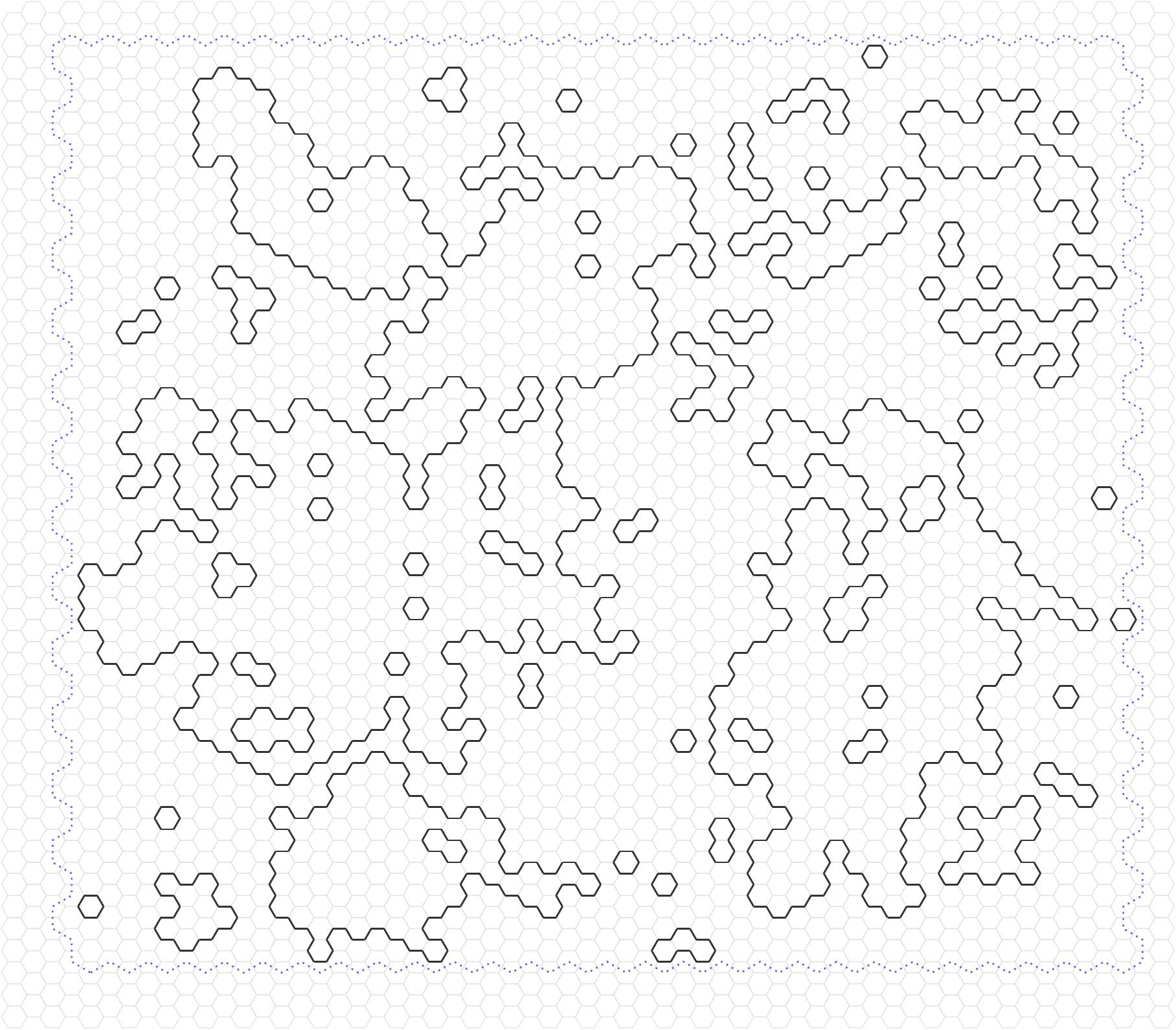}
        \caption{$n=0.8$ and $x=0.6$.}
        \label{fig:loop-sample-n=0.8,x=0.6}
    \end{subfigure}
    \medbreak
    \begin{subfigure}[t]{.5\textwidth}
        \includegraphics[scale=0.46]{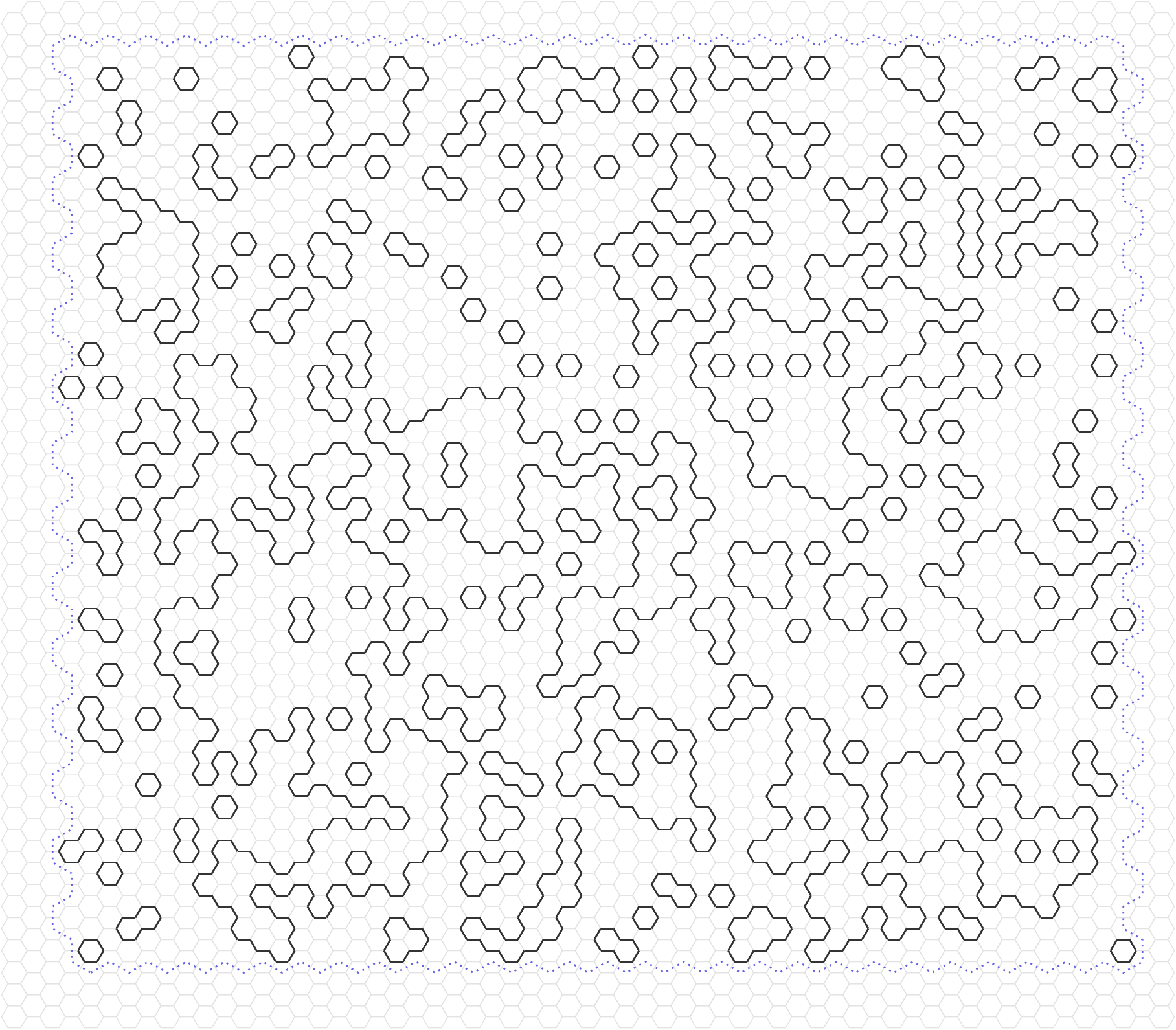}
        \caption{$n=2$ and $x=1/\sqrt{2} \approx 0.707$.}
        \label{fig:loop-sample-n=2,x=1/sqrt2}
    \end{subfigure}%
    \begin{subfigure}{20pt}
        \quad
    \end{subfigure}%
    \begin{subfigure}[t]{.5\textwidth}
        \includegraphics[scale=0.46]{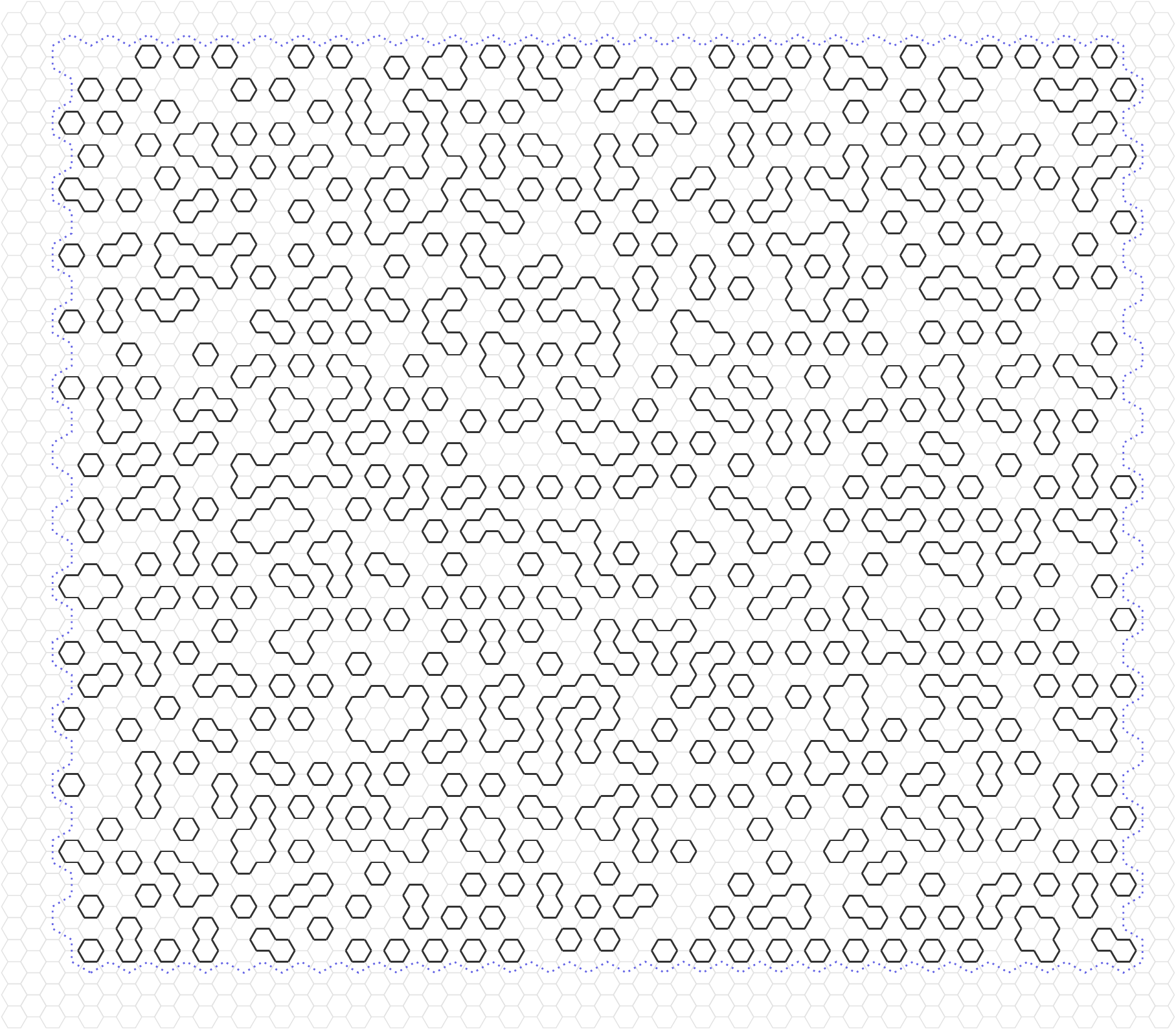}
        \caption{$n=8$ and $x=1$.}
        \label{fig:loop-sample-n=8,x=1}
    \end{subfigure}
    \caption{A few samples of random loop configurations. Configurations are on a $60\times45$ domain of type $0$ and are sampled via Glauber dynamics for 100 million iterations started from the empty configuration. The conjectured phase transition point for $n=0.8$ is $x_c = 1/\sqrt{2 + \sqrt{2-0.8}}\approx 0.568$ and for $n=2$ is $x_c=1/\sqrt{2} \approx 0.707$. Theorem~\ref{thm:no-large-loops} shows that long loops are exponentially unlikely for large $n$.}
    \label{fig:loop-samples-small-n}
\end{figure}
%\restoregeometry

\subsection*{Spin O(n)}
As described in the introduction, the loop $O(n)$ model can be
viewed as an approximation of the spin $O(n)$ model, with the length
of loops related to the spin-spin correlation function. Thus, our
results prove an analogue of the well-known conjecture that
spin-spin correlations decay exponentially (in the distance between
the sites) in the planar spin $O(n)$ model with $n\ge 3$, at any
positive temperature. Proving the conjecture itself remains a
tantalizing challenge.

\subsection*{Small $n$}
Studying the loop $O(n)$ model for small values of $n$ is of great
interest. It is predicted that the model displays critical behavior
only when $n\le 2$. There, it is expected to undergo a
Kosterlitz--Thouless phase transition at $x_c = 1 /
\sqrt{2+\sqrt{2-n}}$, see~\cite{Nie82}, and exhibit conformal
invariance when $x\ge x_c$. Mathematical results on this are
currently restricted to the cases $n=1$ and $n=0$, which correspond
to the \emph{Ising model} and the \emph{self-avoiding walk},
respectively. For these two cases, the critical values have been
identified rigorously in \cite{KimJos74} and \cite{DumSmi12},
respectively. In the $n=1$ case, the model has been
proved~\cite{CheDumHon14,CheSmi12} to be conformally invariant at
$x_c=1/\sqrt3$. For $n=1$ and $x=\infty$ the height function of the
model may be viewed as a uniformly chosen lozenge tiling of a domain
in the plane. This viewpoint leads to a determinantal process, the
\emph{dimer model}, which has been analyzed in great detail (see,
e.g., \cite{Ken04} for an introduction). Conformal invariance has
also been proved for the \emph{double dimer model} which is closely
related to the case $n=2$ and $x=\infty$ (see~\cite{Ken14}).

Our results are limited to the case $n\ge n_0$ and understanding the
various behaviors for small values of $n$ remains a beautiful
mathematical challenge. To give a taste of the different
possibilities, we provide some simulation results in Figure~\reffig{fig:loop-samples-small-n}.

\subsection*{Extremality and uniqueness of the Gibbs measures}
When $n\ge n_0$ and $nx^6\ge C$, we prove that the model has at least
three different Gibbs measures, distinguished by a choice of a
sublattice of the triangular lattice. Are these the only extremal
Gibbs measures in this regime (i.e., is every other measure a convex
combination of these three measures)? Such a result would be in the spirit of
the Aizenman--Higuchi theorem \cite{Aiz80, Hig81} which proves that
the only extremal Gibbs measures for the 2D Ising model are the two
pure states. This theorem was recently extended to the $q$-state
Potts model in \cite{CoqDumIof12}.

For small values of $\max\{n,1\}x^6$, we prove the existence of a
limiting Gibbs measure when exhausting space via an increasing
sequence of domains with vacant boundary conditions. Is this Gibbs
measure unique for each choice of $n$ and $x$ in this regime?
Intuitively, the difficulty in proving this lies in dealing with
domains with boundary conditions which force an interface (i.e.,
part of a loop) through the domain (similarly to the situation in
Figure~\reffig{fig:domain-unique-min-edge-loop-config}). If this
interface passes near the origin with non-negligible probability,
one would obtain a limiting Gibbs measure having an infinite path
with positive probability. However, one expects interfaces to follow
diffusive scaling, similarly to random walk paths, and as such
should have negligible probability to pass close to the origin when
the domain is large. Making such an intuition rigorous is quite
non-trivial and was recently carried out successfully in
\cite{CoqDumIof12} for planar Potts models. Adapting the ideas in
\cite{CoqDumIof12} to the loop $O(n)$ model poses a challenge
as these rely on specific properties of the Potts model. Roughly,
the strategy in \cite{CoqDumIof12} proceeds by showing that when
starting from a large domain $H$ with arbitrary boundary conditions,
only a uniformly bounded number of interfaces will reach the
boundary of a smaller sub-domain $H'$. Then it is shown that these
bounded number of interfaces follow diffusive scaling as in the
intuition above. The first part, bounding the number of interfaces
between the boundary of $H$ and $H'$, may possibly be carried out
for the loop $O(n)$ model by using
Lemma~\ref{lem:prob-inequality-tool}; configurations with many long
interfaces may be `rewired', erasing most of these interfaces and
replacing them with short connections along the boundary of $H$,
yielding configurations with much higher probability. The second
part, however, showing the diffusive scaling, remains a major
obstacle.

\subsection*{The hard-hexagon model}
Our results shed light on the Gibbs measures of the loop $O(n)$
model when $n\ge n_0$ and either $nx^6\le c$ or $nx^6\ge C$. The
structure for $n\ge n_0$ and $c\le nx^6\le C$ remains unclear; see Figure~\reffig{fig:loop-sample-n=8,x=1} and Figure~\reffig{fig:loop-samples-large-n}. Is
there a single $x_c(n)$ at which the model transitions from the
dilute, disordered phase to the dense, ordered phase? What happens
when $x = x_c(n)$?

An intuition for this question may be obtained by considering a
limiting model as $n$ tends to infinity. As noted already in the
paper \cite{DomMukNie81} where the loop $O(n)$ model was introduced,
taking the limit $n\to\infty$ and $nx^6\to\lambda$ leads formally to
the hard-hexagon model. As loops of length longer than $6$ become
less and less likely in this limit, hard-hexagon configurations
consist solely of trivial loops, with each such loop contributing a
factor of $\lambda$ to the weight. Thus, the hard-hexagon model is
the hard-core lattice gas model on the triangular lattice $\T$ with
fugacity $\lambda$. For this model, Baxter \cite{Bax80} (see also
\cite[Chapter 14]{Bax89}) computed the critical fugacity
\begin{equation*}
\lambda_c = \left(2\cos\left(\frac{\pi}{5}\right)\right)^5 = \frac{1}{2}\left(11 + 5\sqrt{5}\right) \approx 11.09017,
\end{equation*}
and showed that as $\lambda$ increases beyond the threshold
$\lambda_c$, the model undergoes a fluid-solid phase transition from
a homogeneous phase in which the sublattice occupation frequencies
are equal to a phase in which one of the three sublattices is
favored. Additional information is obtained on the critical behavior
including the fact that the mean density of hexagons is equal for
each of the three sublattices \cite[Equation (13)]{Bax80} and the
fact that the transition is of second order \cite[Equation
(9)]{Bax80}. Baxter's arguments use certain assumptions on the model
which appear not to have been mathematically justified. Still, this
exact solution may suggest that the loop $O(n)$ model with large $n$
will also have a unique transition point $x_c(n)$, that $nx_c(n)^6$
will converge to $\lambda_c$ as $n$ tends to infinity and that the
transition in $x$ is of second order, with the model having a unique
Gibbs state when $x = x_c(n)$.

\subsection*{Square-lattice random-cluster model and dilute Potts model}
We start with a somewhat informal description of the square-lattice
random-cluster model and refer the interested reader to
\cite{grimmett2006random, duminil2015geometric} for more details. The random-cluster model with parameters $0<p<1, q>0$ on a domain in $\Z^2$ is
a random collection of edges $\eta$ of the domain whose probability
is proportional to
\begin{equation*}
  p^{o(\eta)}(1-p)^{c(\eta)}q^{k(\eta)},
\end{equation*}
where $o(\eta)$ is the number of edges in $\eta$, $c(\eta)$ is the
number of edges of the domain which are not in $\eta$ and $k(\eta)$
is the number of connected components in the graph whose vertices
are the vertices of the domain and whose edges are given by $\eta$.
For each $\eta$, one may draw a loop configuration $\omega_\eta$ (on
the so-called medial lattice) consisting of the loops marking the
boundaries of the connected components (these loops go around the
connected components and on the boundary of each ``hole'' that the
components surround); see Figure~\reffig{fig:random-cluster-loops}. It turns out that the probability of $\eta$ may be rewritten
using these loops so that the probability of $\eta$ is proportional
to
\begin{equation}\label{eq:random_cluster_loop}
  \lambda^{o(\eta)}(\sqrt{q})^{L(\omega_\eta)},
\end{equation}
where $\lambda := \frac{p}{\sqrt{q}(1-p)}$ and $L(\omega_\eta)$ is
the number of loops in $\omega_\eta$. This representation highlights
a self duality occurring when $p$ is such that $\lambda = 1$ and
this self-dual point has been proven to be the critical point
$p_c(q)$ for the random-cluster model~\cite{BefDum12}. The formula \eqref{eq:random_cluster_loop} may immediately
remind the reader of the formula for the probability of
configurations in the loop $O(n)$ model given in Definition~\ref{def:loop-model}.
However, we emphasize that $o(\eta)$ counts the number of edges in
$\eta$ and as such is quite different from the `length' of the loops
in $\omega_\eta$. In fact, the loop configuration $\omega_\eta$ is
necessarily fully packed in the domain for any given $\eta$, so that
$\lambda$ plays a different role from the parameter $x$ of the loop
$O(n)$ model. Still, the formula \eqref{eq:random_cluster_loop} does
suggest an analogy between the random-cluster model at criticality
(when $p = p_c(q)$) and the fully packed (i.e., $x=\infty$) loop
$O(n)$ model with $n=\sqrt{q}$.

\begin{figure}
    \centering
    \includegraphics[scale=1]{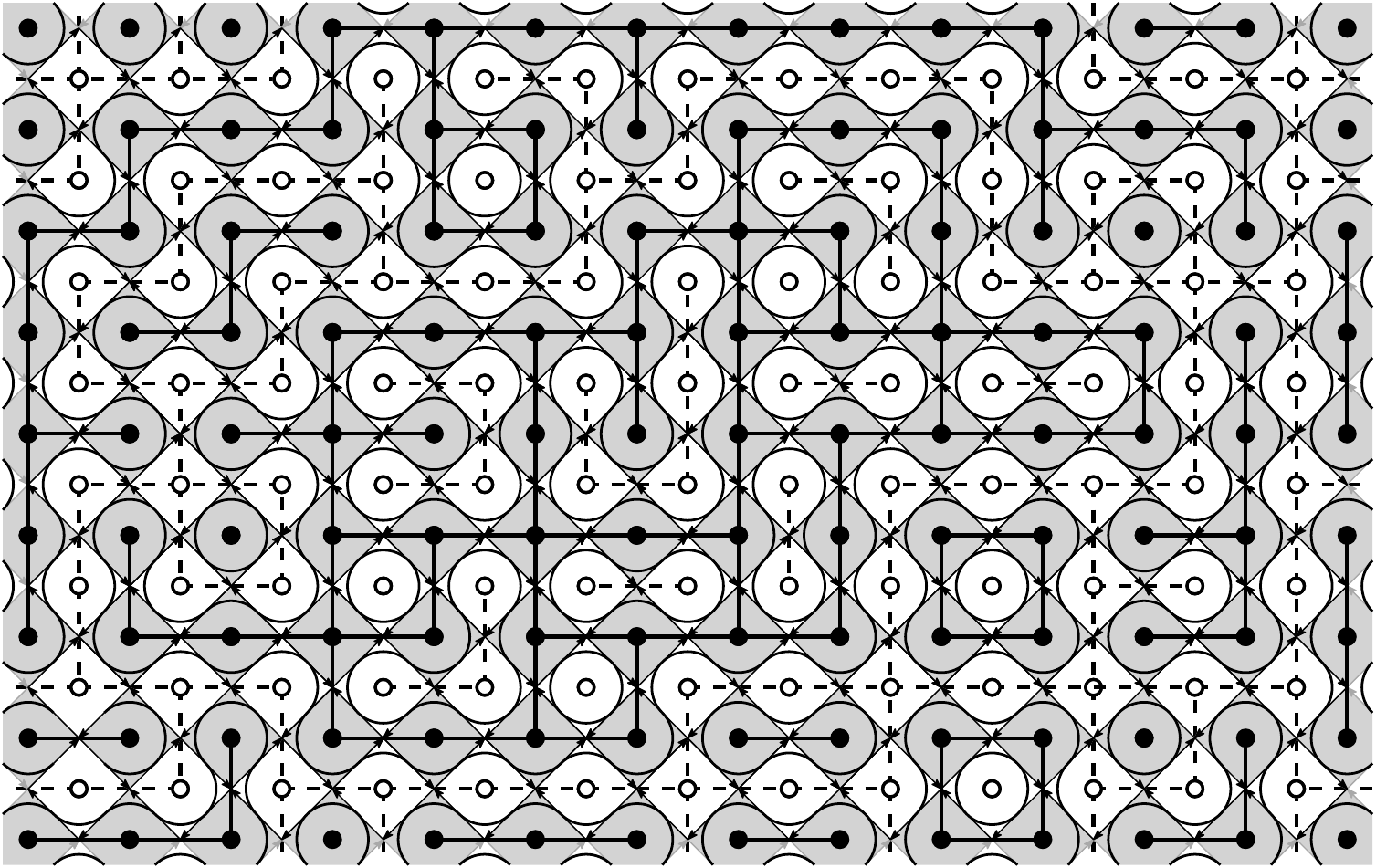}
    \caption{An illustration of a random-cluster configuration $\eta$ and its corresponding loop configuration $\omega_\eta$. The edges of $\eta$ are denoted by bold lines, the edges not in $\eta$ by dashed lines and the loops of $\omega_\eta$ by plain lines.}
    \label{fig:random-cluster-loops}
\end{figure}

Taken with periodic boundary conditions on a square domain, the
random-cluster model has two configurations $\eta$ which maximize
$L(\omega_\eta)$: one in which all the edges of the domain are
absent (yielding loops around the vertices) and one in which all of
them are present (yielding loops around the faces). These
configurations are equally probable at the critical point, but one
is preferred over the other whenever $p\neq p_c(q)$. Following a
proof of Koteck\'y and Shlosman~\cite{KotShl82} for the
closely-related Potts model, it has been proven by Laanait et
al.~\cite{LaaMesMir91} that for large $q$, the random-cluster model
exhibits a first-order phase transition, so that at criticality
there are two Gibbs states corresponding to the two ground states
described above. Our results on the existence of the ordered phase
for large $n$ and $x=\infty$ are quite analogous to this phenomenon.
In fact, it is predicted that the square-lattice random-cluster
model has a first-order phase transition if $q\ge 4$ and otherwise
has a second-order phase transition. This is in line with the
conjectured phase diagram for the loop $O(n)$ model, predicting that
the ordered phase at $x=\infty$ exists only for $n\ge 2$.

We again point out that the parameter $p$ of the random-cluster
model has no analogue in the loop $O(n)$ model and so the existence
of a first-order transition in $p$ does not suggest that such a
transition should occur also when varying $x$. As mentioned above,
it may well be that for large $n$, the transition in $x$ is of
second order by analogy with the situation for hard hexagons.

Lastly, we mention that Nienhuis~\cite{nienhuis1991locus} proposed a version
of the Potts model, termed the \emph{dilute Potts model}, with a
direct relationship to the loop $O(n)$ model. A configuration of the
dilute Potts model in a domain of the triangular lattice is an
assignment of a pair $(s_z, t_z)$ to each vertex $z$ of the domain,
where $s_z\in\{1,\ldots, q\}$ represents a spin and $t_z\in\{0,1\}$
denotes an occupancy variable. The probability of configurations
involves a hard-core constraint that nearest-neighbor occupied sites
must have equal spins (reminiscent of the Edwards-Sokal coupling of
the Potts and random-cluster models) and single-site,
nearest-neighbor and triangle interaction terms involving the
occupancy variables. With a certain choice of coupling constants,
the marginal of the model on the occupancy variables is equivalent
to the loop $O(n)$ model (with $n=\sqrt{q}$), with the loops being
the interfaces between occupied and unoccupied sites. Nienhuis
predicts this choice of parameters to be part of the critical
surface of the dilute Potts model. The properties of the dilute
Potts model appear not to have been studied in the mathematical
literature and it would be interesting to see whether they can shed
further light on the behavior of the loop $O(n)$ model.

\subsection*{Height representation for integer $n$}
When the loop parameter $n$ is an integer, the loop $O(n)$ model admits a height function representation \cite{DomMukNie81}. Let $T_n$ be the $n$-regular tree (so that $T_1 = \{+,-\}$ and $T_2=\Z$) rooted at an arbitrary vertex $\rho$. Let $\text{Lip}_n$ be the set of functions $\varphi\colon\T\to T_n$ satisfying the `Lipschitz condition':
\begin{equation*}
\text{If $y,z\in\T$ are adjacent then either $\varphi(y) = \varphi(z)$ or $\varphi(y)$ is adjacent to $\varphi(z)$ in $T_n$}
\end{equation*}
(in other words, $\varphi$ is a graph homomorphism from $\T$ to the graph $T_n'$ obtained from $T_n$ by adding a loop at every vertex). For a domain $H\subset\HH$, we further set $\text{Lip}_n(H)$ to be the set of $\varphi\in\text{Lip}_n$ satisfying the boundary condition $\varphi(z) = \rho$ for all hexagons $z$ which are not in the interior of $H$ (i.e., which are incident to a vertex in $V(\HH)\setminus V(H)$). Define the `level lines' of $\varphi\in\text{Lip}_n$ by
\begin{equation*}
\omega_\varphi := \big\{e \in E(\HH) \,:\, \text{the edge $e$ borders hexagons $y,z\in\T$ satisfying $\varphi(y)\neq \varphi(z)$}\big\}.
\end{equation*}
Observe that $\omega_\varphi$ is a loop configuration and that if $\varphi\in\text{Lip}_n(H)$ then $\omega_\varphi\in\LC(H,\emptyset)$.
For a real parameter $x>0$, define a probability measure $\nu_{H,n,x}$ on $\text{Lip}_n(H)$ by
\[
 \nu_{H,n,x}(\varphi) := \frac{x^{|\omega_\varphi|}}{Z^{\text{Lip}}_{H,n,x}}, \quad \varphi\in\text{Lip}_n(H),
\]
where $Z^{\text{Lip}}_{H,n,x}$ is the unique constant which makes $\nu_{H,n,x}$ a probability measure. The definition is extended to $x=\infty$ by $\nu_{H,n,\infty}(\varphi) := \lim_{x\to\infty} \nu_{H,n,x}(\varphi)$.

The fact that the loop $O(n)$ model admits a height function
representation is manifested in the relation between the measures
$\nu_{H,n,x}$ and $\Pr_{H,n,x}^\emptyset$. As is straightforward to
verify, if $\varphi$ is a random function chosen according to
$\nu_{H,n,x}$ then $\omega_\varphi$ is distributed according to
$\Pr_{H,n,x}^\emptyset$. In particular, the height function
representation of the loop $O(1)$ model is an Ising model (which may
be either ferromagnetic or antiferromagnetic according to whether
$x<1$ or $x>1$) and the height function representation of the loop
$O(2)$ model is a restricted Solid-On-Solid model. Our main result,
Theorem~\ref{thm:no-large-loops}, implies that long level lines
surrounding a given hexagon are exponentially unlikely in height
functions sampled according to $\nu_{H,n,x}$, when $H$ is a domain
of type $\clr\in\{0,1,2\}$ and $n$ is large. Our proof does not make
use of the height function representation and thus applies to real
$n$. It would be interesting to see whether the height function
representation may be used to provide further information for
integer $n$.

%
%%
%%%
%%%%%%%%%%%%%%
%%%
%%
%

\appendix

\section{Integrals}
\label{sec:integrals}

In this section, we present a detailed derivation of the formulas
approximating the partition function and the spin-spin correlations
in the spin $O(n)$ model on a finite subgraph $H$ of the hexagonal
lattice. Let $u,v\in V(H)$ be distinct vertices and let $H^+$ be the
(possibly multi-)graph obtained by adding an edge $e_{u,v}$ between
$u$ and $v$ to $H$. In the introductory section, the derivation was
reduced to computing integrals of the form
\[
I(\omega):=\int_\Omega \prod_{\{w,w'\} \in E(\omega)} \langle
\sigma_w, \sigma_{w'} \rangle \,d\sigma,
\]
where $\Omega = (\nSn)^{V(H)}$, $\omega$ is an arbitrary subgraph of
$H^+$,
and $d\sigma$ is the product of $|V(H)|$ uniform probability
measures on $\nSn$. Note first that, by symmetry, making the
substitution $\sigma_w \leftarrow -\sigma_w$ for some $w \in V(H)$
does not change the value of this integral and consequently
$I(\omega) = 0$ unless every vertex has even degree in $\omega$. In
other words, if $\omega \subset H$ then $I(\omega) = 0$ unless
$\omega$ is a loop configuration, i.e., $\omega \in
\LC(H,\emptyset)$, and $I(\omega + e_{u,v}) = 0$ unless the degrees
of $u$ and $v$ in $\omega$ are odd and the degrees of all other
vertices are even, i.e., $\omega \in \LC(H,\emptyset,u,v)$.

We shall repeatedly make use of the following identity. For every $x, y \in \R^n$,
\begin{equation}
    \label{eq:xzzy-xy}
    \int_{\nSn} \langle x, z \rangle \langle z, y \rangle \,dz = \langle x, y \rangle,
\end{equation}
where $dz$ is the uniform probability measure on $\nSn$. Note that
both sides of~\eqref{eq:xzzy-xy} are bilinear functions of $x$ and
$y$ and therefore it is enough to verify that~\eqref{eq:xzzy-xy}
holds when $x$ and $y$ are two vectors from the canonical basis
$\{e_1, \ldots, e_n\}$ of $\R^n$. By symmetry, for each $i$,
\[
\int_{\nSn}  \langle e_i, z \rangle \langle z, e_i \rangle \,dz = \frac{1}{n} \sum_{i=1}^n \int_{\nSn}  \langle z, e_i \rangle^2 \,dz = \frac{1}{n} \int_{\nSn} \|z\|^2dz = 1,
\]
If $i \neq j$, substituting $(z_1, \ldots, z_n) \leftarrow (z_1, \ldots, z_{i-1}, -z_i, z_{i+1}, \ldots, z_n)$ yields
\[
\int_{\nSn}  \langle e_i, z \rangle \langle z, e_j \rangle \,dz = - \int_{\nSn}  \langle e_i, z \rangle \langle z, e_j \rangle \,dz = 0.
\]

Suppose first that $\omega \in \LC(H, \emptyset)$. Since the loops of $\omega$ are vertex-disjoint, $I(\omega) = \prod_{L \subset \omega} I(L)$, where $L$ ranges over all loops of $\omega$. Suppose now that $L$ is a loop through vertices $v_0, \ldots, v_\ell$, where $v_\ell = v_0$. Invoking~\eqref{eq:xzzy-xy} repeatedly yields
\[
\begin{split}
I(L) & = \int_\Omega \langle \sigma_{v_0}, \sigma_{v_1} \rangle \cdots \langle \sigma_{v_\ell-1}, \sigma_{v_\ell} \rangle \,d\sigma = \int_\Omega \langle \sigma_{v_0}, \sigma_{v_0} \rangle \,d\sigma = n,
\end{split}
\]
giving $I(\omega) = n^{L_H(\omega)}$.

Suppose now that $\omega \in \LC(H, \emptyset, u, v)$, let $C$ be
the connected component of $u$ (and $v$) in $\omega$, and note that
$C$ must contain a simple path $P$ connecting $u$ and $v$. Since we
have already proved that $I(L) = n$ for every loop $L$, in order to
compute $I(\omega + e_{u,v})$, it is enough to compute
$I(C+e_{u,v})$. A simple case analysis shows that $C$ is either (i)
the path $P$, (ii) the path $P$ and a loop intersecting $P$ in one
of its endpoints, (iii) the path $P$ and two vertex-disjoint loops,
each intersecting $P$ in one of its endpoints, or (iv) the path $P$
and two other simple paths connecting $u$ and $v$, each pair of
paths sharing only the vertices $u$ and $v$. Since the edge
$e_{u,v}$ closes $P$ into a loop, invoking~\eqref{eq:xzzy-xy}
repeatedly to `contract' loops yields that $I(C+e_{u,v})$ equals $n$
in case (i), $n^2$ in case (ii), and $n^3$ in case (iii). In case
(iv), invoking~\eqref{eq:xzzy-xy} repeatedly only gives
\[
I(C+e_{u,v}) = \iint_{\nSn} \langle x, y \rangle^4 \,dxdy,
\]
which is somewhat more difficult to compute. Using symmetry and the
fact that the projection of the Lebesgue measure on
$\mathbb{S}^{n-1} \subset \R^n$ onto the first coordinate gives the
measure on $[-1,1]$ with density $(1-t^2)^{\frac{n-3}{2}}$ up to a
normalization constant, we obtain
\[
\begin{split}
I(C + e_{u,v}) & = \int_{\nSn} \langle x, \sqrt{n} e_1 \rangle^4 \, dx = n^4 \int_{\nSn} \langle x/\sqrt{n}, e_1 \rangle^4 \,dx \\
& = n^4 \cdot \frac{\int_{-1}^1 t^4 (1-t^2)^{\frac{n-3}{2}}
\,dt}{\int_{-1}^1 (1-t^2)^{\frac{n-3}{2}} \,dt} = \frac{3n^3}{n+2},
\end{split}
\]
where one may obtain the final identity using integration by parts.

%
%%
%%%
%%%%%%%%%%%%%%
%%%
%%
%

\section{Circuits and domains}
\label{sec:circuits-and-domains}

Here we prove some facts about circuits and domains.

\begin{proof}[Proof of Fact~\ref{fact:gamma-int-ext}]
    Let $\gamma$ be a circuit and denote by $\HH_\gamma$ the subgraph of $\HH$ obtained by removing from $\HH$ all edges in $\gamma^*$. Let $\Ext\gamma$ be the set of vertices that are the endpoint of some infinite simple path in $\HH_\gamma$.

First, we claim that $\Ext\gamma$ is a connected component of
$\HH_\gamma$. To see this, note first that by definition,
$\Ext\gamma$ is a union of connected components of $\HH_\gamma$.
Furthermore, since $\gamma^*$ is finite, there exists an $R$ and a vertex $u \in \VH$ such
that the complement of the ball of radius $R$ (in the graph distance
determined by $\HH$) centered at $u$ induces the same connected graph $\HH_R$
in both $\HH$ and $\HH_\gamma$. Finally, every infinite simple path
in $\HH$ intersects $\HH_R$ and therefore $\Ext\gamma$ consists of a
single connected component.

    Second, we claim that the set of endpoints of the edges in $\gamma^*$ intersects at most two connected components of $\HH_\gamma$, one of which is $\Ext\gamma$. To see this, suppose that $\gamma = (\gamma_0, \ldots, \gamma_m)$ as in the definition in Section~\ref{sec:definitions}. In order to prove the first part of our claim, it suffices to show that for each $i \in \{1, \ldots, m-1\}$, there are two disjoint $\HH_\gamma$-connected sets of vertices, each of which intersects both $\{\gamma_{i-1}, \gamma_i\}^*$ and $\{\gamma_i, \gamma_{i+1}\}^*$ (where we regard an edge as the set of its endpoints). To see this, note that $\{\gamma_{i-1}, \gamma_i\}^*$ and $\{\gamma_i, \gamma_{i+1}\}^*$ are the only two out of six edges surrounding the hexagon $\gamma_i$ that belong to $\gamma^*$. Consequently, the removal of $\gamma^*$ partitions the six vertices surrounding $\gamma_i$ into two $\HH_\gamma$-connected sets, each of which intersects both $\{\gamma_{i-1}, \gamma_i\}^*$ and $\{\gamma_i, \gamma_{i+1}\}^*$. For the second part of the claim, consider an arbitrary infinite simple path in $\HH$ which uses an edge from $\gamma^*$. Let $\{v, w\}$ be the last edge of $\gamma^*$ on this path and observe that either $v$ or $w$ belongs to $\Ext\gamma$. Hence, $\Ext\gamma$ is one of the $\HH_\gamma$-connected components that contains an endpoint of an edge of $\gamma^*$.

    Third, we claim that $\Ext\gamma \neq \VH$. If this were not the case, then in particular there would be a $\{v, w\} \in \gamma^*$ such that both $v$ and $w$ belong to the same connected component of $\HH_\gamma$. Consequently, there would be a simple path $P$ in $\HH_\gamma$ that connects $v$ and $w$. The edge $\{v,w\}$ and $P$ would then form a cycle in $\HH$ that contains exactly one edge of $\gamma^*$. This is impossible since the basic $6$-cycles surrounding the hexagons of $\T$ generate the cycle space of $\HH$ and each of these basic cycles intersects $\gamma^*$ in either $0$ or $2$ edges.

    Fourth, we claim that $\VH \setminus \Ext\gamma$ is $\HH_\gamma$-connected, that is, every two $v, w \notin \Ext\gamma$ are in the same connected component of $\HH_\gamma$. To see this, consider two infinite simple paths $P_v$ and $P_w$ in $\HH$ that start at $v$ and $w$, respectively. Since $v, w \notin \Ext\gamma$, both $P_v$ and $P_w$ contain an edge from $\gamma^*$. Let $v', w'$ be the first vertices in $P_v$ and $P_w$, respectively, which are incident to edges of $\gamma^*$. Clearly $v,v'$ and $w,w'$ lie in the same $\HH_\gamma$-connected components, other than $\Ext\gamma$. By our second claim, $v'$ and $w'$ must belong to the same $\HH_\gamma$-connected component. Hence, $v$ and $w$ also belong to the same $\HH_\gamma$-connected component, which we shall from now on denote by $\Int\gamma$.

    Finally, we show that both $\Ext\gamma$ and $\Int\gamma$, as $\HH_\gamma$-connected components, are induced subgraphs of $\HH$ and that $\Int\gamma$ is finite. The first assertion follows from the fact that the two endpoints of each edge of $\gamma^*$ belong to different $\HH_\gamma$-connected components, which we have already established above. If the second assertion were false, then $\Int\gamma$ would be an infinite $\HH_\gamma$-connected graph and hence it would contain an infinite simple path, contradicting the fact that $\Int\gamma \cap \Ext\gamma = \emptyset$.
\end{proof}

\begin{proof}[Proof of Fact~\ref{fact:circuit-domain-bijection}]
    Let $H$ be a domain and let $E$ be the set of edges of $\HH$ with exactly one endpoint in $V(H)$.
    Let $T$ be the auxiliary graph with vertex set $\T$ whose edges are all pairs $\{y,z\}$ such that $\{y,z\}^* \in E$. We first claim that all vertex degrees in $T$ are even. Indeed, to see this for the degree of a hexagon $z\in \T$, it suffices to traverse the vertices bordering $z$ in order and to consider which of them belong to $V(H)$. It follows that $T$ contains a circuit $\gamma$. By Fact~\ref{fact:gamma-int-ext}, $\gamma^*$ splits $\HH$ into exactly two connected components. As $\gamma^* \subseteq E$ and $\Int\gamma$ is finite and non-empty, and as $V(H)$ is finite, connected and with connected complement, it must be that $\VH \setminus V(H) \subseteq \ExtVert\gamma$ and $V(H) \subseteq \IntVert\gamma$. Consequently, $H = \Int\gamma$.
\end{proof}

\begin{proof}[Proof of Fact~\ref{fact:circuits-max}]
Denote $A := \IntVert\sigma$, $A' := \IntVert{\sigma'}$ and $B := A \cup A'$. Let us first show that $B$ is connected. If $A \cap A' \neq \emptyset$ then this is immediate. Otherwise, by assumption, there exists an edge $\{v,u\} \in \sigma^* \cap (\sigma')^*$. Assume without loss of generality that $v \in A$ and $u \notin A$. Then $u \in A'$ and $v \notin A'$, and thus, $B$ is connected.

Let $C$ be the unique infinite connected component of $\VH \setminus B$ and let $D := \VH\setminus C$. It is straightforward to check that $D$ is finite, $B \subset D$ and $\partial D \subset \partial B$. Since $B$ is connected, this implies that $D$ is connected. Thus, as $\VH\setminus D=C$ is connected, the subgraph of $\HH$ induced by $D$ is a domain.

By Fact~\ref{fact:circuit-domain-bijection}, there exists a circuit $\gamma$ such that $D=\IntVert\gamma$. It remains to check that $\gamma^* \subset \sigma^* \cup (\sigma')^*$. Let $\{v,u\} \in \gamma^*$ be such that $v \in D$ and $u \notin D$. In particular, $v \in B$ and $u \notin B$. Thus, either $v \in A$ so that $\{v,u\} \in \sigma^*$, or $v \in A'$ so that $\{v,u\} \in (\sigma')^*$.
\end{proof}

\bibliographystyle{amsplain}
%\nocite{*}
\bibliography{biblicomplete}

\end{document}